\documentclass[11pt,prd,aps,amsfonts,showpacs,nofootinbib,longbibliography,notitlepage,superscriptaddress]{revtex4-1}

\usepackage{graphicx} % Required for inserting images
\usepackage{amsmath}
\usepackage{amssymb}
\usepackage{amsthm}
\usepackage{mathrsfs}
\usepackage{extarrows} 
\usepackage{xcolor}
\usepackage[margin=0.8 in]{geometry}
\usepackage{tikz-cd}
\usepackage[normalem]{ulem}
\usepackage{scalerel}
\usepackage{bbm}
\usepackage{caption}
\usepackage{subcaption}
\usepackage{hyperref}

\let\bigopsize\bigoplus
\def\bigominus{{\scalerel*{\boldsymbol\ominus}{\bigopsize}}}

\usepackage{array}   % for \newcolumntype macro
\newcolumntype{L}{>{$}r<{$}} % math-mode version of "l" column type

\newcolumntype{C}{>{$}c<{$}} % math-mode version of "l" column type

\newcommand{\Z}{\mathbb Z}

\newcommand{\T}{\mathcal T}

\newcommand{\R}{\mathbb R}
\newcommand{\Q}{\mathbb Q}
\newcommand{\N}{\mathbb N}
\newcommand{\q}{\mathbf q}

\newcommand{\G}{\mathcal G}
\newcommand{\im}{\textup{im}\hspace{1pt}}

\newcommand{\cX}{\mathcal X}
\newcommand{\cZ}{\mathcal Z}

\newtheorem{theorem}{Theorem}[section]
\newtheorem{proposition}[theorem]{Proposition}
\newtheorem{lemma}[theorem]{Lemma}
\newtheorem{corollary}[theorem]{Corollary}

\theoremstyle{definition}
\newtheorem{defn}[theorem]{Definition}

\newcommand{\Span}{\operatorname{span}}

\newcommand{\fe}{{\mathfrak e}}
\newcommand{\fm}{{\mathfrak m}}

	\makeatletter
\newsavebox{\@brx}
\newcommand{\llangle}[1][]{\savebox{\@brx}{\(\m@th{#1\langle}\)}%
  \mathopen{\copy\@brx\kern-0.5\wd\@brx\usebox{\@brx}}}
\newcommand{\rrangle}[1][]{\savebox{\@brx}{\(\m@th{#1\rangle}\)}%
  \mathclose{\copy\@brx\kern-0.5\wd\@brx\usebox{\@brx}}}
\newcommand{\spn}{\text{span}\hspace{1pt}}
\makeatother

\newcommand{\overbar}[1]{\mkern 1.5mu\overline{\mkern-1.5mu#1\mkern-1.5mu}\mkern 1.5mu}

\definecolor{arpit}{RGB}{127,0,0}

\begin{document}

\title{Planon-modular fracton orders}

\author{Evan Wickenden}
\affiliation{Department of Physics, University of Colorado, Boulder, CO 80309, USA}
\affiliation{Center for Theory of Quantum Matter, University of Colorado, Boulder, CO 80309, USA}
\author{Marvin Qi}
\affiliation{Kadanoff Center for Theoretical Physics and Enrico Fermi Institute, University of Chicago, Chicago, IL 60637, USA}
\affiliation{Department of Physics, University of Colorado, Boulder, CO 80309, USA}
\affiliation{Center for Theory of Quantum Matter, University of Colorado, Boulder, CO 80309, USA}
\author{Arpit Dua}
\affiliation{Department of Physics, Virginia Tech, Blacksburg, Virginia 24060, USA}
\affiliation{Department of Physics and Institute for Quantum Information and Matter, Caltech, Pasadena, CA, USA}
\author{Michael Hermele}
\affiliation{Department of Physics, University of Colorado, Boulder, CO 80309, USA}
\affiliation{Center for Theory of Quantum Matter, University of Colorado, Boulder, CO 80309, USA}
\date{\today}

\begin{abstract}
There are now many examples of gapped fracton models, which are defined by the presence of restricted-mobility excitations above the quantum ground state.  However, the theory of fracton orders remains in its early stages, and the complex landscape of examples is far from being mapped out. Here we introduce the class of planon-modular ($p$-modular) fracton orders, a relatively simple yet still rich class of quantum orders that encompasses several well-known examples of type I fracton order. The defining property is that any non-trivial point-like excitation can be detected by braiding with planons. From this definition, we uncover a significant amount of general structure, including the assignment of a natural number (dubbed the weight) to each excitation of a $p$-modular fracton order. We identify simple new phase invariants, some of which are based on weight, which can easily be used to compare and distinguish different fracton orders. We also study entanglement renormalization group (RG) flows of $p$-modular fracton orders, establishing a close connection with foliated RG. We illustrate our general results with an analysis of several exactly solvable fracton models that we show to realize $p$-modular fracton orders, including $\mathbb{Z}_n$ versions of the X-cube, anisotropic, checkerboard, 4-planar X-cube and four color cube (FCC) models. We show that each of these models is $p$-modular and compute its phase invariants. We also show that each example admits a foliated RG at the level of its non-trivial excitations, which is a new result for the 4-planar X-cube and FCC models. We show that the $\mathbb{Z}_2$ FCC model is not a stack of other better-studied models, but predict that the $\mathbb{Z}_n$ FCC model with $n$ odd is a stack of 10 4-planar X-cubes, possibly plus decoupled layers of 2d toric code. We also show that the $\Z_n$ checkerboard model for $n$ odd is a stack of three anisotropic models. 
\end{abstract}

\maketitle

\newpage

\tableofcontents

\section {Introduction}
\label{section:intro}

The existence of fracton order \cite{Chamon_2005,Haah_2011,Nandkishore_2019,Pretko_2020_review} adds significant complexity to the problem of characterization and classification of gapped phases of matter in three dimensions. Fracton models support excitations of restricted mobility, and typically have ground state degeneracy on the three-torus that grows with system size. Unlike topologically ordered systems, fracton models are not fixed points under conventional entanglement renormalization flows \cite{Haah_2014,Shirley_2018,Dua_2020}, and their low-energy properties cannot be captured by topological quantum field theories. There is now a complex landscape of fracton models, but a bigger picture is lacking. The characterization of fracton orders in terms of universal properties is not well understood, and it is not at all clear that fracton orders can be usefully classified. It is thus an important problem to identify simpler yet non-trivial classes of fracton orders whose characterization is tractable.

In this paper, we introduce such a class of abelian fracton orders, dubbed planon-modular ($p$-modular) fracton orders. The defining property is that any non-trivial point-like excitation can be detected by braiding with some planon. Here, a planon is an excitation constrained to move within a plane, and a non-trivial excitation is one that cannot be created by a local operator. This definition encompasses some well-known models, such as the X-cube model \cite{Vijay_2016} and the Chamon model \cite{Chamon_2005}, but excludes others, notably Haah's cubic code \cite{Haah_2011}. Starting from the definition, we uncover a significant amount of structure, including many simple phase invariants, some of which generalize the notion of quotient superselection sectors \cite{Shirley_2019}. We also study entanglement renormalization group (RG) flows of $p$-modular fracton orders, and find a close connection with foliated RG and foliated fracton orders \cite{Shirley_2018}.
In general, $p$-modular fracton orders offer an appealing balance between simplicity and complexity, and their classification may be within the reach of future work.

Some previous works have also been concerned with identifying simpler classes of fracton orders. 
Ref.~\onlinecite{Vijay_2016} defined type I fracton orders to be those with some non-trivial mobile excitations, including fully mobile excitations, planons (mobile within planes) and lineons (mobile along lines). This contrasts with type II orders, where all the non-trivial excitations are fractons, which are immobile in isolation. Among type II fracton orders, exemplified by Haah's cubic code, little is understood about the characterization of phases in terms of universal properties (but see Refs.~\cite{Williamson_2016,Vijay_2016,Song_2024}). On the other hand, more is known about the characterization of certain type I fracton orders, especially, but not only, the X-cube model \cite{Vijay_2016,Shirley_2019_entanglement,Song_2019,Shirley_2019,Shirley_2019_gauging,Pai_2019,Qi_2020,Seiberg_2021,Gorantla_2021,Song_2024}.  However, as a class, type I orders are not simpler than type II orders. A rather trivial observation is that if we stack the X-cube model and Haah's code, the resulting model is type I, but still has all the complexity of Haah's code. Less trivially, there are type I fracton models that have both fractal structure and lineon excitations \cite{Castelnovo_2012, Yoshida_2013, Dua_2019}.

One important class of simpler fracton orders are the foliated fracton orders \cite{Shirley_2018}. We define a gapped system to be foliated if it is a fixed point of the foliated renormalization group (RG) \cite{Shirley_2018}, and call a gapped quantum order foliated if it can be realized by a foliated system. A foliated fracton order is a foliated quantum order with restricted mobility excitations. In standard entanglement RG \cite{Vidal_2007}, one treats decoupled qubits as a free resource that can be integrated out in the RG sense. Foliated RG is a generalization where decoupled two-dimensional layers are treated as a free resource and can be integrated out \cite{Shirley_2018}. Foliated RG was further generalized from a quantum circuit perspective in Ref.~\onlinecite{Wang_2023}; however, in this paper, we reserve the term foliated fracton order for systems that are fixed points under the original foliated RG of Ref.~\onlinecite{Shirley_2018}. Based on known examples that again include the X-cube \cite{Shirley_2018} and Chamon \cite{Shirley_2023} models, foliated fracton orders indeed appear to be simpler than their more general cousins. However, not much is understood about the general structure of foliated fracton orders, and it can be challenging to check whether a model is foliated, which requires exhibiting an RG circuit. We emphasize that the notions of foliated fracton orders and foliated fracton phases are distinct; this is explained in Section~\ref{sec:phase-invariants}, where we discuss different definitions of phases relevant to the results of our paper.

\begin{figure}
    \centering
    \includegraphics[width=0.7\textwidth]{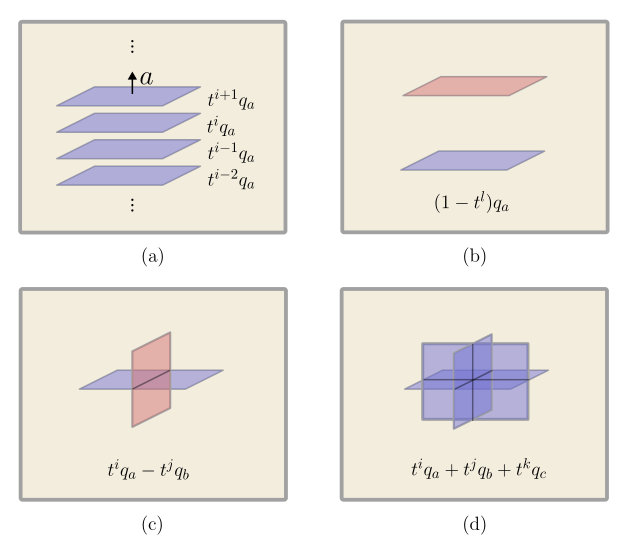}
    \caption{(a) $a$-oriented plane charges partitioning the points of a lattice. Here $q_a$ denotes the charge within a particular plane and $t^i$ represents translation in the normal direction by $i$ units.  The orientation $a$ is a two-dimensional vector subspace of $\R^3$ parallel to the planes shown.   (b) Plane charge content of a planon excitation carrying two non-zero plane charges of the same orientation. (c) Plane charge content of the $w=2$ lineon excitations of the X-cube model.(d) Depiction of the plane charge content of an X-cube model fracton, which has $w=3$. 
    }
    \label{fig:plane-charges}
\end{figure}

In contrast to foliated fracton orders, given a commuting Pauli Hamiltonian, it is generally not very difficult to check whether it realizes a $p$-modular fracton order. In addition, $p$-modularity leads to a significant degree of structure on which a general theory can be built. The key observation is that the fusion and mobility of excitations can be described completely in terms of conserved quantities, which we refer to as plane charges (Fig.~\ref{fig:plane-charges}), and which are defined in terms of mutual statistics between planons and arbitrary point-like excitations. A plane charge is characterized by an orientation and by a (relative) position normal to its orientation (Fig.~\ref{fig:plane-charges}a). By an orientation, we mean a two-dimensional plane in $\R^3$, \emph{i.e.} a two-dimensional vector subspace of $\R^3$. Each excitation of a $p$-modular fracton order carries a finite number of plane charges, which completely characterize its fusion and mobility properties (see Figs.~\ref{fig:plane-charges}b-d). An important quantity is the \emph{weight} $w$ of an excitation, which is the number of different orientations of plane charges it carries. Excitations with $w=1$ are planons, excitations with $w=2$ are lineons, and excitations with $w \geq 3$ may be lineons or fractons depending on whether the relevant planes intersect along a line or at a point.  We say that an excitation is \emph{intrinsic} if it cannot be obtained by combining excitations of lower weight.  This gives finer distinctions among excitations than earlier notions of intrinsic-ness \cite{Song_2019, Prem_2019}, where an excitation is said to be intrinsic if it cannot be obtained as a composite of excitations of higher mobility, \emph{e.g.} an intrinsic fracton is not a composite of planons and/or lineons.

In earlier work, Ref.~\onlinecite{Song_2019} used braiding with planons to label topological charges in the analysis of a class of exactly solvable fracton models, and related this labeling to mobility. However, the action of translation symmetry on these labels, and the resulting module structure of plane charges, which plays a key role in our analysis, was not considered.  Incorporating translation symmetry, the idea of describing fusion and mobility in terms of plane charges was discussed previously in Ref.~\onlinecite{Pai_2019}, where it was shown how to describe the excitations of the X-cube model in this manner. However, in contrast to Ref.~\onlinecite{Song_2019}, the results of Ref.~\onlinecite{Pai_2019} were obtained by writing down, in an \emph{ad hoc} manner, a module of plane charges and the associated isomorphism with the module of superselection sectors. Colloquially, the plane charge description was ``pulled out of thin air,'' instead of being derived in a manner that could be generalized. Here, building on the ideas of Refs.~\onlinecite{Song_2019, Pai_2019}, we define plane charges in a canonical way for any $p$-modular fracton order, incorporating translation symmetry and without making arbitrary choices or \emph{ad hoc} constructions. In particular, the weight of an excitation is a well-defined physical property.

\begin{table}
%\centering
\begin{tabular}{||c | c | c | c | c | c||} 
 \hline
 Type of order & Non-trivial excitations & Planarity & Foliated & $p$-modular & References \\ 
 \hline
 Trivial order & -- & 0 & Yes & Yes & -- \\ 
 \hline
  3d toric code & Fully mobile (point-like and loops) & 0 & Yes & No & -- \\ 
 \hline
  2d toric code layers & Planons & 1 & Yes & Yes & -- \\ 
 \hline
  Anisotropic model & Lineons, planons & 2 & Yes & Yes & \cite{Shirley_2019} \\ 
 \hline
  X-cube & Fractons, lineons, planons & 3 & Yes & Yes & \cite{Vijay_2016, Shirley_2018} \\ 
 \hline
  4-planar X-cube & Fractons, lineons, planons & 4 & Yes* & Yes & \cite{Slagle_2018, Shirley_2019, Shirley_2023} \\ 
 \hline
  FCC model & Fractons, lineons, planons & 7 & Yes* & Yes & \cite{Ma_2017} \\ 
 \hline
\end{tabular}
\caption{Known distinct foliated and $p$-modular quantum orders that can be realized in $\mathbb{Z}_2$ CSS Pauli codes. For each type of order, the table describes the mobility of non-trivial excitations, gives the planarity $k$ as defined in the text, and indicates whether the model is foliated and $p$-modular; all these properties also hold for $\Z_n$ generalizations of these models studied in this paper.  In some cases, references are given to works where the model was introduced or shown to be foliated.  The $p$-modularity of these models is shown in this paper.  All the models except the 4-planar X-cube and FCC model have been shown to be foliated in previous works by exhibiting an RG circuit.  The latter two models are shown to be foliated in this paper at the level of an algebraic description of the excitations, denoted by the asterisk. To our knowledge, the 4-planar X-cube and FCC models have not previously been shown to be foliated in the literature, but we note that Ref.~\onlinecite{Shirley_2023} showed foliation of the Chamon model which is very closely related to the 4-planar X-cube model.}
\label{tab:orders}
\end{table}

In the simplest setting of $\mathbb{Z}_2$ CSS Pauli codes,\footnote{These are commuting Pauli Hamiltonians written in terms of Pauli $X$ and $Z$ operators acting on qubits, where every Hamiltonian term is either a product of $X$ operators or of $Z$ operators.} only a few examples of foliated quantum orders are known and are given in Table~\ref{tab:orders}. We show in this paper that these examples are all $p$-modular, except for the 3d toric code, which is not $p$-modular for the trivial reason that there are no planons.  We also show that all these fracton orders are foliated at the level of an algebraic description of their excitations, but we do not give explicit RG circuits.  These properties are also shown for $\Z_n$ generalizations of the tabulated fracton orders, as well as the $\Z_n$ checkerboard model.

An important property of a $p$-modular fracton order given in Table~\ref{tab:orders} is the number $k$ of different orientations of lattice planes in which planons move. We refer to $k$ as the planarity and say the system is $k$-planar. If there are no planons, by definition $k=0$. In related terminology, a system is said to be $k$-foliated if it is a fixed point of foliated RG where $k$ distinct orientations of two-dimensional layers are integrated out.

One of our examples, the four color cube (FCC) model,  has not been extensively studied in previous work.  The $\Z_2$ FCC model was introduced in Ref.~\onlinecite{Ma_2017} by coupling four X-cube models and condensing membrane-like excitations, and has thus also been referred to in the literature as the membrane-coupled X-cube model.  To our knowledge the FCC model has not been shown or argued to be foliated previously.  In this paper, we introduce a $\Z_n$ generalization of the FCC model and fully characterize the mobility of its excitations in terms of plane charges.  Surprisingly, the FCC model is 7-planar, with planons mobile in three $\{100\}$ and four $\{111\}$ planes; in the coupled X-cube construction, there are only $\{100\}$ planons before the X-cube models are coupled.

The general relationship between foliated and $p$-modular fracton orders is an interesting question.  In this paper, we obtain a strong constraint on entanglement RG flows of $p$-modular fracton orders:  the system that is integrated out during an RG step can only have planons or bound states of planons as its non-trivial excitations.  One possibility is foliated RG, where one integrates out decoupled 2d layers.  However, there are fracton orders with only planon excitations, but which are not stacks of decoupled layers \cite{Shirley_2020}.  This leaves open the possibility of an RG where such non-decoupled-layer planon-only systems are integrated out, which would mean that there are $p$-modular fracton orders that are not foliated. Whether this actually occurs remains an interesting open question for future work.

We can make progress characterizing $p$-modular fracton orders by finding simple phase invariants.  These are universal properties of a phase that, ideally, can be used to efficiently compare and distinguish different models.  For example, we are interested in understanding whether the orders given in Table~\ref{tab:orders} are all distinct, or whether some of them can be decomposed as a stack of the others.  

An important example of a phase invariant is the quotient superselection sectors (QSS), where one takes a quotient of all excitations by planons \cite{Shirley_2019}.  In an abelian fracton order, this produces an abelian group $\mathcal Q$, which is finite for many type I fracton orders, resulting in an invariant that can easily be compared between different models.  However, because the QSS is just a single abelian group, its power to distinguish among fracton orders is rather limited; in particular, the QSS alone is not sufficient to distinguish among the examples of Table~\ref{tab:orders}.  Additional invariants are needed that, like the QSS, can easily be compared between models.

In this paper, for $p$-modular fracton orders, we generalize the QSS to the weighted quotient superselection sectors (weighted QSS), which consists of a sequence of abelian groups $\mathcal W = (\mathcal W_1, \dots, \mathcal W_k)$.  Here, $k$ is the planarity and $\mathcal W_i$ is obtained by taking a quotient by all excitations of weight $i$ and below; note that $\mathcal W_1 = \mathcal Q$, recovering the original QSS.  Moreover, $\mathcal W_k = S/S = 0$.  The weighted QSS tells us about the presence of intrinsic excitations of each weight, via the fact that $\mathcal W_i \ncong \mathcal W_{i+1}$ if and only if there is at least one intrinsic excitation of weight $i+1$.  Another important property is that if we stack two quantum systems, the weighted QSS combine by direct sum, that is $\mathcal W^{{\rm stack}}_i = \mathcal W^1_i \oplus \mathcal W^2_i$.

Using the weighted QSS, we show that each $\Z_2$ fracton order in Table~\ref{tab:orders} is distinct, in the sense that it is not equivalent to any stack of the other tabulated orders under a finite-depth unitary circuit. In particular, the $\Z_2$ FCC model is not equivalent to a stack of X-cube models.  More generally, because the weighted QSS tells us that the FCC model has intrinsic weight-seven fractons, it cannot be obtained as a stack of any fracton orders with planarity less than seven.  These statements also hold if we consider the coarser notion of foliated phase equivalence (see Sec.~\ref{sec:phase-invariants}), except that a system of 2d toric code layers is foliated-equivalent to a trivial system.  Table~\ref{tab:wqss} gives the weighted QSS for each fracton order of Table~\ref{tab:orders}.  To illustrate the power of the weighted QSS, we observe that if we only consider $\mathcal W_1$, we cannot distinguish the $\Z_2$ FCC model from a stack of 10 4-planar X-cube models, but the full weighted QSS easily makes this distinction.

\begin{table}[]
    \centering
    \begin{tabular}{||c|c||}
        \hline
        Type of order & Weighted QSS 
        \\
        \hline
         2d toric code layers & 0 
         \\
        \hline
         Anisotropic model & $(\Z_2^2, 0)$
         \\
        \hline
         X-cube & $(\Z_2^3,\Z_2,0)$
         \\
        \hline
         Checkerboard & $(\Z_2^6,\Z_2^2,0)$
         \\
        \hline
         4-planar X-cube & $(\Z_2^4,\Z_2,\Z_2,0)$
         \\
        \hline
         FCC model & $(\Z_2^{40},\Z_2^{16},\Z_2^{10},\Z_2^8,\Z_2^2,\Z_2^2,0)$
         \\
         \hline
    \end{tabular}
    \caption{Weighted quotient superselection sectors (QSS) for certain fracton orders realized in $\Z_2$ CSS Pauli codes. In addition to the fracton orders listed in Table~\ref{tab:orders}, we also show the weighted QSS for the $\Z_2$ checkerboard model, which is finite-depth-unitary-equivalent to a stack of two X-cube models \cite{Shirley_2019_checkerboard}; this is reflected in the weighted QSS.  Excluding the checkerboard model, these weighted QSS imply that none of these fracton orders is a stack of the others. The weighted QSS, as well as another invariant -- the constraint module $C$ -- are given for $\Z_n$ generalizations of these fracton orders in Table~\ref{tab:invariants} of Sec.~\ref{sec:phase-invariants}.}
    \label{tab:wqss}
\end{table}

We now outline the remainder of the paper, including a summary of results not mentioned above. In Sec.~\ref{section:definitions}, we introduce an algebraic structure called a $p$-theory, which abstracts away from specific models to describe the fusion, mobility and planon mutual statistics of an abelian fracton order.  Within this framework we describe some of the basic properties of $p$-modular fracton orders; in particular, we define plane charges and explain how to describe fusion and mobility in terms of them.  We also introduce the constraint module $C$, which plays an important role in the theory.  Each generator of $C$ is associated with a linear equation satisfied by physical configurations of plane charges; we prove that $C$ is finite, so only a finite number of linear equations are needed.  Various technical facts about $p$-theories are proved in Appendix~\ref{appendix:technical}, including the result that $C$ is finite.  Section~\ref{sec:phase-invariants} introduces the weighted QSS $\mathcal W$ and some other simple invariants based on the weight of excitations that can be obtained from $\mathcal W$.  Different definitions of phase are discussed that have been used in the context of fracton matter, and it is argued that $\mathcal W$ and $C$ are invariants for all these notions of phase.  In particular, $\mathcal W$ and $C$ are shown to be unchanged under adding decoupled 2d layers (as in foliated fracton phase equivalence), and under coarsening the translation symmetry, which refers to enlarging a system's unit cell by a finite amount. The invariants $C$ and $\mathcal W$ for the $\Z_n$ fracton models studied in Sec.~\ref{section:examples} are summarized in Table~\ref{tab:invariants} of Sec.~\ref{sec:phase-invariants}.

Entanglement RG flows of $p$-modular fracton orders are discussed in Sec~\ref{section:RG}.  We point out that while systems with fully mobile excitations are not $p$-modular but can be foliated, it is reasonable to conjecture that foliated abelian fracton orders without non-trivial fully mobile excitations may always be $p$-modular.  We obtain the constraint on RG flows of $p$-modular theories discussed above.  We then discuss how to determine the RG flow of a given $p$-modular fracton order, showing that upon making a physically reasonable assumption, it is sufficient to study the RG at the level of the $p$-theory without exhibiting a quantum circuit.  We give sufficient conditions for a $p$-theory to be an RG fixed point, and for the RG to be foliated, which are stated in detail and proved in Appendix~\ref{appendix:foliated-RG}.  

Section~\ref{section:examples} studies a series of examples, after a preliminary discussion where we discuss how to work out the plane charge structure of a commuting Pauli Hamiltonian.  In particular we introduce and discuss how to find the $p$-modular exact sequence, which is the main technical tool used to analyze the examples.  We study the X-cube model, layers of 2d toric code, the anisotropic model, the checkerboard model, the 4-planar X-cube model, and the FCC model.  In each case we study a $\Z_n$ version of the model in question for arbitrary $n \geq 2$.  For each model, we compute its weighted QSS and constraint module, and show it admits a foliated RG where layers of $\Z_n$ toric code are integrated out at each step.  While the $\Z_2$ checkerboard model is known to be a stack of two X-cube models, a continuum quantum field theory analysis predicted that the analogous statement does not hold for $n > 2$ \cite{Gorantla_2021}.  We verify this, and also use the phase invariants to predict that, when $n$ is odd, the $\Z_n$ checkerboard model is a stack of three $\Z_n$ anisotropic models.  We verify this prediction by exhibiting a quantum circuit in Appendix~\ref{appendix:checkerboard}. The weighted QSS imply that the FCC model is not a stack of other models analyzed in the literature for even $n$, but for odd $n$ they are consistent with the FCC model being a stack of 10 4-planar X-cube models. Finally, the paper concludes with a discussion of several open questions in Sec.~\ref{section:conclusion}.

\section{Basic properties of $p$-modular fracton orders}
\label{section:definitions}

Throughout the paper, we focus on the point-like excitations of translation-invariant gapped lattice systems in infinite three-dimensional space, describing the fusion and mobility properties largely following Ref.~\onlinecite{Pai_2019}. We only consider bosonic systems, \emph{e.g.} spin models. We denote the set of superselection sectors (also referred to as excitation types or topological charges) by $S$; elements of $S$ are equivalence classes of excitations that can be transformed into one another by acting with local operators. All point-like excitations are taken to be abelian, which makes $S$ into an additive abelian group under fusion of excitations. Moreover, we assume translation invariance with lattice symmetry group ${\cal T} \cong \Z^3$; this makes $S$ into a module over the group ring $\Z \T$, via the action of translation on excitations. We assume that $S$ is finitely generated as a $\Z\T$-module; that is, we assume that there are a finite number of generating excitation types from which all others can be obtained by acting with translations and taking arbitrary finite sums and differences. We further assume that every excitation type has finite fusion order, \emph{i.e.} for every element $x \in S$ there is a positive integer $m$ such that $mx = 0$. Under these assumptions, there is a smallest positive integer $n$ such that $nx = 0$ for all $x \in S$. It is more convenient to work with the group ring $R = \Z_n \T$ and view $S$ as an $R$-module, and we will do this throughout the paper.\footnote{Formally, this can be understood by noting that the ideal $(n) \subset \Z \T$ acts trivially on $S$, which is thus naturally a $\Z\T/(n)$-module, and $\Z_n\T \cong \Z\T / (n)$.} Elements $r \in R$ are formal linear combinations of group elements $g \in \T$. That is, $r = \sum_{g \in \T} c_g \, g$, where $c_g \in \Z_n$ and only finitely many $c_g$ are non-zero. There is an antipode operation mapping $R \to R$; the antipode of $r$ is denoted $\bar{r} = \sum_{g \in \T} c_g \, g^{-1}$, so in particular $\bar{g} = g^{-1}$.

The above assumptions hold in the setting of translation-invariant exact Pauli codes (see Sec.~\ref{section:examples}), but are not limited to fracton orders that can be realized with commuting Pauli Hamiltonians. However, our setup is not entirely general; there are fractonic infinite-component Chern-Simons theories that violate the finite fusion order assumption \cite{Ma_2022}. Because we will be mainly interested in $p$-modular fracton orders, we make one further assumption, namely that there are no non-trivial fully mobile excitations; the formal statement of this assumption is given below. This is used to prove some useful technical properties before assuming $p$-modularity, and it must hold in any $p$-modular fracton order, because a fully mobile excitation must have trivial braiding with any planon. However, this is a weaker assumption than $p$-modularity; for example, Haah's cubic code has no non-trivial fully mobile excitations but is not $p$-modular.

As in Ref.~\onlinecite{Pai_2019}, we use translation symmetry as a tool to describe the mobility of excitations. However, we do not want to study translation symmetry enriched fracton orders \cite{Pretko_2020}. To address this issue, we can impose translation symmetry in a ``coarse'' manner, where we allow for breaking of $\T$ down to subgroups $\T' \subset \T$, where $\T'$ is arbitrary so long as $\T' \cong \Z^3$ \cite{Haah_2013,Haah_2014,Pai_2019,Hermele_KITP_talk, Hermele_UQM_talk, Hermele_inprep}. Such symmetry breaking corresponds to a finite (but possibly large) enlargement of the crystalline unit cell. This is discussed further below in Sec.~\ref{sec:phase-invariants} when we discuss phases and phase invariants.

To bring in lattice geometry, we choose an embedding $\eta : {\cal T} \hookrightarrow \R^3$. The map $\eta$ is required to be an injective group homomorphism, where the group structure on $\R^3$ is the standard vector addition. We further assume that $\eta(t_1), \eta(t_2), \eta(t_3)$ form a vector space basis for $\R^3$ where $\{ t_1, t_2, t_3 \}$ is any set of generators for ${\cal T}$. We need the embedding $\eta$ because $\T$ is an abstract group without any spatial information, but we want to use translation symmetry to describe mobility of excitations.

For any element $x \in S$, we denote its stabilizer group by $\T_x = \{ g \in \T | gx = x\}$. Note that any non-trivial subgroup of $\T$ is isomorphic to $\Z$, $\Z^2$ or $\Z^3$. We say an element $p \in S$ is a \textit{planon} if $\T_p \cong \Z^2$, while fully mobile excitations $x \in S$ have $\T_x \cong \Z^3$. Note that in our terminology, fully mobile superselection sectors are not considered planons. As mentioned above, we assume that there are no non-trivial fully mobile excitations; formally, this is the statement that for any $x \in S$, $\T_x \cong \Z^3$ implies $x=0$. Lineons are excitations with $\T_x \cong \Z$, while fractons have $\T_x = 0$.

Using the embedding $\eta$, we define the orientation of any excitation $x \in S$ to be the subspace $a(x) \equiv \text{span}\hspace{1pt} (\eta(\T_x)) \subset \R^3$.
Then, by Proposition~\ref{prop:sbgrp}, an element $p \in S$ is a planon if and only if $a(p)$ is a two-dimensional subspace of $\R^3$. The same statement holds for fractons, lineons and fully mobile excitations, where $a(x)$ is a zero, one and three dimensional subspace, respectively. Whether two excitations have the same orientation is independent of the choice of $\eta$.
We denote by $A$ the set of distinct orientations of all planons in $S$. For each $a \in A$, $S_a \subset S$ is defined to be the submodule generated by $a$-oriented planons. It is proved in Appendix~\ref{appendix:technical} that every non-zero element of $S_a$ is an $a$-oriented planon, using the assumption of no non-trivial fully mobile excitations. The submodule generated by all planons is denoted $S^{(1)} = \sum_{a \in A} S_a$.\footnote{If $M$ is a module and $M_i \subset M$ a (possibly infinite) collection of submodules, it is standard to define the sum $\sum_i M_i \equiv \{ \sum_i m_i | m_i \in M_i, \text{ finitely many } m_i \neq 0 \}$. It is easily seen that $\sum_i M_i$ is a submodule of $M$; in fact, it is the smallest submodule that contains all the $M_i$.} In Appendix~\ref{appendix:technical} we prove that in fact $S^{(1)}$ is a direct sum, $S^{(1)} = \bigoplus_{a\in A} S_a$. Because $R$ is a Noetherian ring, submodules of finitely generated modules (such as $S$) are finitely generated, and it follows that $A$ is a finite set (see Appendix~\ref{appendix:technical} for details). The cardinality of $A$ is independent of the embedding $\eta$.

So far, with the data $(\T, S, \eta)$, we have specified a \emph{fusion theory of excitations}, which captures both fusion and mobility properties. To study $p$-modularity, we also need to include information on mutual statistics between planons and other excitations; the resulting structure will be referred to as a $p$\emph{-theory of excitations}. We use this terminology because a $p$-theory of excitations does not include all statistical information -- for example, the X-cube and semionic X-cube model have the same $p$-theory, but they are known to be distinct due to lineon-lineon exchange statistics \cite{Pai_2019}. A \emph{theory of excitations}, not defined in this paper, should include all the data of a $p$-theory of excitations, together with enough additional data to specify the outcome of any statistical process.

On physical grounds, there is a bilinear form $\langle \cdot, \cdot \rangle_a : S_a \times S \to \Z_n$
that records the statistical phase acquired when a planon in $S_a$ is moved in a large loop around a charge in $S$.\footnote{The elements of $S_a$ and $S$ are of course equivalence classes of the actual excitations involved in such a statistical process, but the statistical phase is independent of the representative excitations.} To make the mutual statistics and thus the bilinear form well-defined, we need to specify a sense of motion (clockwise or counterclockwise) for each planon remote detection process. For each orientation $a \in A$ we make an arbitrary two-fold choice of unit vector $\hat{n}_a \in \R^3$ normal to the given orientation. The bilinear form then gives the phase accumulated when the planon is moved counterclockwise, looking toward the ``tip'' of the vector $\hat{n}_a$, \emph{i.e.} viewing the process from the $+\hat{n}_a$-direction. If we change $\hat{n}_a \to -\hat{n}_a$, then to describe the same physical system we should also change $\langle \cdot, \cdot\rangle_a \to -\langle \cdot, \cdot\rangle_a$. Sometimes we omit the $a$-subscript on the bilinear form when it is clear from context.

For any planons $p, q \in S_a$ and charges $x, y \in S$, the $a$-oriented bilinear form is taken to satisfy (omitting subscripts):
\begin{enumerate}
    \item (Translation Invariance). $\langle p, x \rangle = \langle gp, gx \rangle$ for all $g \in \T$.
    \item (Bilinearity). 
     $\langle p + q, x \rangle = \langle p,x\rangle + \langle q, x \rangle$ and $\langle p, x + y\rangle = \langle p, x \rangle + \langle p, y \rangle$.
    \item (Symmetry of planon mutual statistics). $\langle p, q\rangle = \langle q, p \rangle$.
    \item (Locality). 
    Fix $p \in S_a$ and $x \in S$, and let $g \in \T$ be a translation such that $g^i p = p$ implies $i = 0$. That is, $g$ is transverse to the plane of mobility of $p$. Then there is an integer $m$ such that $|i| \geq m$ implies $\langle g^i p, x \rangle = 0$. Note that $m$ may depend on $p$, $x$ and $g$.
\end{enumerate}
The first three properties hold on general grounds for any translation-invariant, gapped abelian system. The fourth property is less obvious, and we give an argument that it holds in exact Pauli codes.
To see this, we consider the Wilson loop operator $W$ that transports a planon $p$ around a large counterclockwise loop of linear size $L$. (Note that $L$ is not the system size.) The operator $W$ is a ribbon supported on the boundary of a disc-shaped region $D$ whose boundary is the loop. Moreover, the thickness of $W$ in the direction normal to the disc is a constant not depending on $L$. By exactness (see Sec.~\ref{section:examples}), since $W$ produces no excitations, it is necessarily a product of stabilizers. In addition, as illustrated in Fig.~\ref{fig:smaller-loops}, $W$ can be written as a product of ${\cal O}(L^2)$ operators $W_i$ transporting $p$ around smaller loops of length ${\cal O}(1)$ within the disc $D$, where the operators $W_i$ are related to one another by translation symmetry. Each $W_i$ is product of ${\cal O}(1)$ stabilizers, which implies that we can write $W$ as a product of stabilizers within a slab $V = D \times [-\ell, \ell]$; here $\ell$ is a length of ${\cal O}(1)$ and the region $V$ is simply a thickened version of $D$. Therefore, the superselection sectors that $p$ can detect are those whose excitations are supported within $V$, and the locality property follows. In fact, we expect the locality property to hold in general, because it seems unreasonable for the $\Z_n$-valued statistics between $g^i p$ and $x$ to be non-zero for arbitrarily large transverse separation between the two excitations. We note that if one removes the assumption of finite fusion order, then there are systems that violate the locality property, but in this case the mutual statistics angle can be arbitrarily small and even irrational \cite{Ma_2022}.

\begin{figure}
    \centering
    \includegraphics[width=0.5\linewidth]{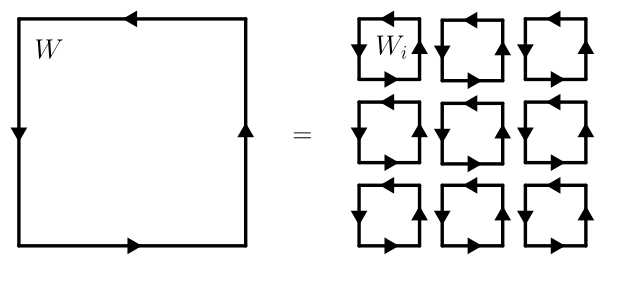}
    \caption{Birds-eye view of the ribbon operator $W$ and its decomposition into a product of ${\cal O}(L^2)$ smaller ribbon operators $W_i$.}
    \label{fig:smaller-loops}
\end{figure}

Recall that a system is $p$-modular when every non-trivial point-like excitation can be detected by braiding with some planon. Formally, $p$-modularity is the property that for any non-zero superselection sector $x \in S$, there exists an orientation $a \in A$ and a planon $p \in S_a$ such that $\langle p, x \rangle_a \neq 0$. This is a kind of remote-detectability property; every non-trivial excitation is remotely detectable by braiding with some planon. It is straightforward to show that $p$-modularity can equivalently be stated as a remote distinguishability property: given $x, y \in S$ such that $x \neq y$, then there exists an orientation $a \in A$ and a planon $p \in S_a$ such that $\langle p, x \rangle_a \neq \langle p, y \rangle_a$.

To summarize our discussion, a $p$-theory of excitations consists of the 5-tuple of data \\ $P = (\T, S, \eta, \{ \hat{n}_a \},  \{ \langle \cdot, \cdot\rangle_a \})$, where the last two entries denote the choice of unit normal vector for each orientation and the set of $a$-oriented bilinear forms. The conditions satisfied by the data are given in the discussion above. In Appendix~\ref{appendix:technical} we give a definition of isomorphism of $p$-theories, and write $P_1 \cong P_2$ when two theories are isomorphic. Isomorphism essentially tells us that two $p$-theories are identical up to arbitrary choices in how $\T$ and $S$ are presented. It is part of a formal definition of phase equivalence of $p$-theories, but is not the full story, in the sense that isomorphic $p$-theories should be phase equivalent (but not necessarily vice versa) for any notion of phase equivalence. We give a further discussion of phase equivalence in Sec.~\ref{sec:phase-invariants}, but defer a full formal treatment to future work.

An important operation on quantum systems is stacking, where we combine two systems into a new one by taking a tensor product of on-site Hilbert spaces and choosing the Hamiltonian to be a decoupled sum of the original Hamiltonians. It will be useful to define the corresponding stacking operation on $p$-theories. We consider two $p$-theories $P_1 = (\T, \eta, S_1, \{ \hat{n}^1_a \}, \{ \langle \cdot , \cdot \rangle^1_a \} )$ and $P_2 = (\T, \eta, S_2, \{ \hat{n}^2_{b} \}, \{ \langle \cdot , \cdot \rangle^2_b \} )$, with sets of orientations $A$ and $B$, respectively. Note that the translation group $\T$ and map $\eta : \T \to \R^3$ is the same in both theories; stacking is only defined when these data agree. The stacked theory is denoted $P_1 \ominus P_2$, and its module of superselection sectors is the direct sum $S_1 \oplus S_2$.  The full data of the stacked theory is $P_1 \ominus P_2 = (\T, \eta, S_1 \oplus S_2, \{ \hat{n}^s_a \}, \{ \langle \cdot, \cdot \rangle^s_a \} )$. Here the set of orientations is $A_s = A \cup B$, and we choose $\hat{n}^s = \hat{n}^1_a$ if $a \in A$ and $\hat{n}_a^s = \hat{n}^2_a$ if $a \in B$. (If $a \in A$ and $a \in B$, we require $\hat{n}^1_a = \hat{n}^2_a$ as a condition for $P_1 \ominus P_2$ to be defined.) The bilinear forms are given by
\begin{equation}
\langle p, x \rangle^s_a = \langle p_1, x_1 \rangle^1_a + \langle p_2, x_2 \rangle^2_a \text{,}
\end{equation}
for any $a$-oriented planon $p \in S_1 \oplus S_2$ and any $x \in S_1 \oplus S_2$, for which we have unique decompositions $p = p_1 + p_2$ and $x = x_1 + x_2$ with $p_i, x_i \in S_i$. The first (second) term is defined to be zero if $a \notin A$ ($a \notin B$). This definition encodes the desired property that excitations in $P_1$ are transparent to planons in $P_2$, and vice versa.

In the setting of a $p$-theory of excitations, we now give a formal definition of plane charges. Choosing an orientation $a \in A$, we define two superselection sectors $x,y \in S$ to be $a$-indistinguishable, written $x \sim_a y$, if $\langle p, x\rangle_a = \langle p, y \rangle_a$ for all $p \in S_a$. This is clearly an equivalence relation, and the equivalence classes are referred to as $a$-oriented plane charges, or just $a$-plane charges.  The set of $a$-plane charges is denoted $Q_a$, and is an $R$-module. It is straightforward to show that $Q_a = S / K_a$,
where
\begin{equation}
    K_a = \{ x \in S | \text{$\langle p, x \rangle_a = 0$ for all $p \in S_a$} \}
\end{equation}
is referred to as the submodule of $a$-transparent superselection sectors. The associated quotient map is written $\pi_a : S \to Q_a$, and we think of $\pi_a(x)$ as the $a$-plane charge carried by $x \in S$. 

In the notation above, a $p$-theory is $p$-modular if and only if, for any $x,y \in S$, $x \sim_a y$ for all $a \in A$ implies $x = y$. It is often convenient to work with a slightly different characterization of $p$-modularity. We define the \emph{plane charge module} $Q = \oplus_{a \in A} Q_a$, and the map $\pi : S \to Q$ by $\pi = \sum_{a \in A} i_a \circ \pi_a$, where $i_a : Q_a \hookrightarrow Q$ is the natural inclusion. Then it is easy to see that $p$-modularity is equivalent to the condition that $\pi$ is injective. This means that $\pi(S)$ is an isomorphic copy of $S$ as a submodule of the plane charge module, and $\pi(x)$ uniquely represents each superselection sector $x \in S$ in terms of the plane charges it carries. We often identify $S$ with $\pi(S)$ in $p$-modular theories. The property of $p$-modularity plays a key role in our results, and, {\bf for the remainder of the main text after this section, we always assume $p$-theories to be $p$-modular unless stated otherwise.} This assumption is not always made in the Appendices.

The injection $\pi : S \to Q$ is useful in large part because it gives a geometrical picture of the mobility properties of superselection sectors. This is rooted in the fact that plane charges have the same mobility properties as planons. Recall that non-zero elements $p \in S_a$ are planons, and by definition $\T_p \cong \Z^2$. In Appendix~\ref{appendix:technical} it is shown that for any non-zero $a$-plane charge $q_a \in Q_a$, we also have $\T_{q_a} \cong \Z^2$.  Moreover, $a$-plane charges are mobile only within $a$-oriented planes; formally this is the statement that $\spn \eta (\T_{q_a}) = a$.  If we act with a translation $t$ transverse to $a$, then $t q_a \neq q_a$ for any non-zero $q_a \in Q_a$. 

A key quantity characterizing a $p$-modular $p$-theory is the number of orientations of planons $k = |A|$. We say that a theory with $k = |A|$ is $k$-planar, and refer to $k$ as the planarity. In the literature, the term $k$-foliated has sometimes been used instead of $k$-planar; we introduce and use the latter term here, preferring to reserve the term ``foliated'' for properties having to do with foliated RG or foliated fracton phases. Indeed, it is not clear if every $k$-planar $p$-modular fracton order can be realized as a fixed point of foliated RG; some results pertaining to this question are discussed in Sec.~\ref{section:RG}.

Given any $R$-module $M$, we define the stabilizer group
\begin{equation}
\T_M = \{ g \in \T | g m = m  \text{ for all } m \in M \} \text{.}
\end{equation}
For the module of $a$-oriented planons, we have $\T_{S_a} \cong \Z^2$. Moreover, in a $p$-modular theory, $\T_{Q_a} = \T_{S_a}$. These results are proved in Appendix~\ref{appendix:technical}.
It is convenient to define
\begin{equation}
\T_a \equiv \T_{Q_a} = \T_{S_a} \text{.}
\end{equation}
Because $\T_a$ acts trivially on $Q_a$, we can think of $Q_a$ as an $R_a \equiv \Z_n \T / \T_a$ module. We recall the following mathematical fact: if $\rho : G \to H$ is a group homomorphism, and $M$ is a $\Z_n H$-module, then there is a $\Z_n G$-module $\rho^* M$ which is identical to $M$ as an abelian group and with $G$-action given by $g \cdot m = \rho(g) m$ for any $m \in M$ and $g \in G$. In the present context we have the quotient map $\rho_a : \T \to \T / \T_a$, and if we view $Q_a$ as an $R_a$-module, the $R$-module structure is recovered by considering $\rho^*_a Q_a$. We sometimes view $Q_a$ as an $R$-module and sometimes as an $R_a$-module, slightly abusing notation by using the same symbol $Q_a$ in both cases.

For any $p$-theory of excitations (not necessarily $p$-modular), a key structural property is that the $R$-module $C \equiv \operatorname{coker} \pi = Q / \pi(S)$ has finitely many elements; this result is proved in Appendix~\ref{appendix:technical}.\footnote{While $C$ is guaranteed to be a finitely generated $R$-module, finitely generated $R$-modules can have infinitely many elements. An example is $S$ itself for any of the fracton models we consider in this paper.} The associated surjective quotient map is denoted $\omega : Q \to C$. We refer to $C$ as the \emph{constraint module} and $\omega$ as the \emph{constraint map}, because $\omega$ expresses $\Z_n$-linear constraints that vanish on elements of $\pi(S) \subset Q$; that is, for $q \in Q$, $\omega(q) = 0$ if and only if $q \in \pi(S)$. Plane charge configurations $q \in Q$ with $q \notin \pi(S)$ are unphysical; they cannot be realized by any excitation. So the constraints encoded by $C$ and $\omega$ are satisfied precisely by physical plane charge configurations. 

We conclude this section by discussing some useful equivalent ways of expressing the planon statistics data of a $p$-modular $p$-theory. First, we can equivalently think of the bilinear forms $\langle \cdot, \cdot \rangle_a$ as maps
\begin{equation}
\langle \cdot , \cdot \rangle_a : \operatorname{ker} \omega|_{Q_a} \times Q_a \to \Z_n \text{.}
\end{equation}
This is valid because $\pi_a |_{S_a}$ gives an isomorphism $S_a \cong Q_a \cap \pi(S) = \operatorname{ker} \omega|_{Q_a}$, and because $\langle p, x \rangle_a$ only depends on $\pi_a(x) \in Q_a$. Viewed in this way, the bilinear forms satisfy the same properties enumerated above, where the symmetry property becomes $\langle q, q' \rangle_a = \langle q', q \rangle_a$ if both $q,q' \in \operatorname{ker} \omega|_{Q_a}$.

A more non-trivial repackaging of the statistical data is accomplished by defining new bilinear forms taking values in $R_a$, namely $\llangle \cdot, \cdot \rrangle_a : \ker \omega|_{Q_a} \times Q_a \to R_a$, according to
\begin{equation}
    \llangle 
    p, q \rrangle_a
    \equiv
    \sum_{[g] \in \T / \T_a} \langle g^{-1} p , q \rangle_a [g].
    \label{eq:planon-pairing}
\end{equation}
Here $p \in \operatorname{ker} \omega|_{Q_a}$, $q \in Q_a$, and we express the coset of $\T_a$ containing $g \in \T$ by $[g]$; this expression is well-defined because $g^{-1} p$ only depends on $[g]$. By locality, the right hand side of Eq.~(\ref{eq:planon-pairing}) contains only a finite number of non-zero terms and is thus well-defined in $R_a$. We note that $R_a$ is naturally an $R$-module, where the $R$-action is given by the surjective ring homomorphism $R \to R_a$ induced by the quotient map $\rho_a : \T \to \T / \T_a$. The value of $\llangle p, \pi_a(x) \rrangle_a$ for $x \in S$ encodes the mutual statistics of $p$ and all its distinct translates with $x$. For any $r \in R$, the new bilinear forms satisfy
\begin{equation}
 \llangle r p, q \rrangle_a =  \llangle  p, \bar{r} q \rrangle_a
 = r \llangle  p, q \rrangle_a \text{.}
\end{equation}

Specifying the bilinear form $\llangle \cdot, \cdot \rrangle_a$ is equivalent to giving the induced $R$-module homomorphism $\Phi_a : \ker \omega|_{Q_a} \to Q_a^*$, defined by $\Phi_a(p) = \llangle p, \cdot \rrangle_a$. Here $Q^*_a \equiv \overline {\text{Hom}}_R(Q_a, R_a)$, where the bar indicates that the elements are $R$-module anti-homomorphisms, \emph{e.g.} $\Phi_a(p)(r q) = \bar{r} \Phi_a(p)(q)$. It is straightforward to check that $p$-modularity implies the map $\Phi_a$ is injective. We use the maps $\Phi_a$ in our treatment of RG for $p$-modular fracton orders (see Sec.~\ref{section:RG}), and to specify the statistical data of the examples studied in Sec.~\ref{section:examples}.

\section {Phase invariants}
\label{sec:phase-invariants}

Here we describe a family of phase invariants of $p$-modular fracton orders. The key property underlying most of the invariants is the weight of superselection sectors, as described in Sec.~\ref{section:intro} and defined formally below. The invariants capture universal aspects of mobility of excitations; some of them generalize the notion of quotient superselection sectors developed in previous work \cite{Shirley_2019}. Because our invariants only depend on the data of a $p$-theory, they cannot distinguish all fracton orders; for example, the $\Z_2$ X-cube and semionic X-cube models have the same $p$-theory but are distinct phases (under some definitions of phase) due to different lineon-lineon exchange statistics \cite{Pai_2019}. However, the invariants are sensitive enough to distinguish the examples discussed in this paper, which is not possible with previously available invariants.  More generally, our invariants provide a powerful set of tools with which to distinguish $p$-modular fracton orders. As noted in Sec.~\ref{section:definitions}, throughout this section (and throughout the remainder of the main text) we always assume that $p$-theories are $p$-modular.

Any discussion of phase invariants for fracton orders raises the question of what we mean by a fracton phase of matter.  This is a  subtle and non-trivial question, and requires some discussion. At the end of this section, we review the different relevant definitions of gapped phases following recent work of one of us in Refs.~\onlinecite{Hermele_KITP_talk, Hermele_UQM_talk, Hermele_inprep}. We argue that our invariants are phase invariants for \emph{all} the relevant definitions of phase.  This means that, for the purposes of this paper, we do not need to choose among the available definitions.  

The weight of a superselection sector $x \in S$ is defined to be the number of orientations of planons that detect $x$, \emph{i.e.}
\begin{align}
    w(x) \equiv |A_x|, && A_x = \{a \in A | \pi_a(x) \neq 0 \} \text{.}
\end{align}
Note that $w(x) = 0$ if and only if $x = 0$. 
In addition, we say $x$ is intrinsic if $x \neq 0$ and $x$ cannot be written as a finite sum $x = \sum_i y_i$ with each $y_i \neq 0$ and $w(y_i) < w(x)$. For example, a composite of two planons of different orientations is a $w = 2$ lineon, but is not intrinsic. However, X-cube model fracton excitations are intrinsic with $w = 3$, because they cannot be formed as a composite of lower weight excitations (lineons and planons). 

Now we introduce our phase invariants. The simplest of these is the \emph{intrinsic weight sequence}
\begin{equation}
W_{int} \equiv \Big\{ w \in \{ 2,3,\dots\} \Big| \exists \text{ intrinsic } x \in S \text{ with } w = w(x) \Big\} \text{.}
\end{equation}
That is, $W_{int}$ is the set of weights $w \geq 2$ carried by some intrinsic excitation. We do not include $w=1$, because all weight-1 excitations (planons) are trivially intrinsic. Any non-trivial $p$-modular fracton order has planons, so including $w=1$ conveys little information. Moreover, if we did include $w=1$, then $W_{int}$ would not be an invariant of foliated fracton phases. The \emph{maximum intrinsic weight} $w_{max}$ is the largest element of $W_{int}$.

To give some examples for $\Z_n$ fracton models, the X-cube model has $W_{int} = \{2,3 \}$ as there are intrinsic $w=2$ lineons and $w=3$ fractons, while the anisotropic model has $W_{int} = \{2 \}$ since there are intrinsic lineons and no excitations with weight $w > 2$. The 4-foliated X-cube model has $w_{max} = 4$, while the FCC model for even $n$ has $w_{max} = 7$, where in both cases the maximum intrinsic weight is carried by the simplest ``elementary'' fracton excitations. This is already enough to show that the FCC model (for even $n$) is not a stack of previously studied examples. These statements are proved by computing the weighted quotient superselection sector invariants introduced below, from which the intrinsic weight sequence can be obtained.

A more refined invariant generalizes the quotient superselection sector (QSS) invariant of \cite{Shirley_2019}, constructed there as an invariant of foliated fracton phases. The QSS invariant is the quotient module $S / S^{(1)}$, \emph{i.e.} one considers topological charges modulo planons. Using the weight structure of excitations, we define the \emph{weighted quotient superselection sectors} (QSS) in a $k$-planar system as the sequence of quotients
\begin{align}
    \mathcal W = ({\cal W}_1, \dots, {\cal W}_k) \equiv  (S / S^{(1)}, S/ S^{(2)}, \dots, S / S^{(k)}), 
\end{align}
where $S^{(i)}$ is the submodule generated by excitations of  weight at most $i$. We observe that $S^{(1)} \subset S^{(2)} \subset \cdots \subset S^{(k)} = S$. The first entry of $\mathcal W$ is the same as the QSS invariant. Since $S^{(k)} = S$, the sequence $\mathcal W$ always terminates with $0$. There are no intrinsic excitations of weight $i$ if and only if $S^{(i)} = S^{(i-1)}$, so ${\cal W}_i = {\cal W}_{i-1}$ if and only if $i$ does not appear in $W_{int}$. In this manner we can reconstruct the intrinsic weight sequence $W_{int}$ from the weighted QSS. The maximum intrinsic weight is the smallest $i \leq k$ for which ${\cal W}_i = 0$. In the weighted QSS invariant, we forget about the module structure and think of the ${\cal W}_i$ as abelian groups; we do this because the module structure of ${\cal W}_i$ is not an invariant for all the types of phases we wish to consider.

Identifying $S$ with $\pi(S)$, the modules $S^{(i)}$ can be computed using the constraint map, per the expression
\begin{align}
    S^{(i)}
   = S^{(i-1)} +
    \sum_{\{a_1,\dots, a_i \} \subset A} \ker \omega|_{Q_{a_1} + \dots + Q_{a_i}} ,
    \label{eq:Si-formula}
\end{align}
where the sum is over the distinct $i$-element subsets of $A$. This is a direct sum for $i=1$, but not in general for $i > 1$. While the $S^{(i-1)}$ term on the right-hand side is superfluous as it is automatically contained within the second term, it is often a helpful reminder in doing explicit computations. The computation of the entries in Table~\ref{tab:invariants} is described in Sec.~\ref{section:examples}.

One final invariant was already introduced above in Sec.~\ref{section:definitions}, and is the constraint module $C = Q / \pi(S)$. As with the weighted QSS, the module structure of $C$ is not a phase invariant, so in this context we forget the module structure and think of $C$ as an abelian group. As proved in Appendix~\ref{appendix:technical}, $C$ is finite. The invariants $C$ and $\mathcal W$ are quoted in Table~\ref{tab:invariants} for the examples of $p$-modular $\Z_n$ fracton models studied in Sec.~\ref{section:examples}. We observe that $C \cong \mathcal W_1$ as abelian groups for all the examples studied in this paper. At present, we do not have an explanation for this phenomenon and we do not know how general it may be; these questions will be interesting to consider in future work.

It is important to understand when the invariants are trivial, and also how they behave under stacking of quantum systems. We call a $p$-theory \emph{essentially 1-planar} if it has $S = S^{(1)}$, \emph{i.e.} if all non-trivial excitations are planons and composites thereof. This property is trivially the same as $\mathcal W_1 = 0$; moreover, $\mathcal W_1 = 0$ implies all the $\mathcal W_i = 0$.
In Appendix~\ref{appendix:technical}, we prove that $C = 0$ if and only if a $p$-theory is essentially 1-planar. Next, we consider the stack $P_s = P_1 \ominus P_2$ (see Sec.~\ref{section:definitions} for the definition of stacking in terms of $p$-theories). The intrinsic weight sequences combine by taking the union, \emph{i.e.} $W^s_{int} = W^1_{int} \cup W^2_{int}$. Moreover, because $S^{(i)}_s = S^{(i)}_1 \oplus S^{(i)}_2$, the weighted QSS are given by the direct sums ${\cal W}^s_i = {\cal W}^1_i \oplus {\cal W}^2_i$. The constraint module is also a direct sum, $C_s = C_1 \oplus C_2$.

\begin{table*}[]
    \centering
    \begin{tabular}{c c l} 
    \text{Model} & \text{Constraint Module} & \text{Weighted QSS} 
    % & \text{Foliated ER}
    \\
    % \hline
    Toric code layers &  0 & 0 
    \\
    Anisotropic model
    &
    $\Z_n^2$ & $(\Z_n^2, 0)$ 
    \\
    X-cube
    &
    $\Z_n^3$ & $(\Z_n^3, \Z_n, 0)$ 
    \\
    Checkerboard
    &
    $\Z_n^6$
    &
    $(\Z_n^6, \Z_{(n,2)}^2, 0)$
    \\
    4-foliated X-cube
    & $\Z_n^4$ & $(\Z_n^4, \Z_n, \Z_n, 0)$
    \\
    Four Color Cube
    &
    $\Z_n^{40}$
    &
    $(\Z_n^{40}, \Z_{(n,2)}^6\times \Z_{n}^{10}, \Z_n^{10}, \Z_{(4,n)}^2 \times \Z_{(2,n)}^6, \Z_{(2,n)}^2, \Z_{(2,n)}^2, 0)$
    \end{tabular}
    \caption{
        Constraint modules and weighted QSS 
        for the $p$-modular $\Z_n$ fracton models studied in Sec.~\ref{section:examples}. Here, $(r,s)$ denotes the greatest common divisor of integers $r$ and $s$.
    }
    \label{tab:invariants}
\end{table*}

Now we turn to the issue of what we mean by a fracton phase and thus by a phase invariant. In the most familiar definition of gapped phase, referred to here as the standard definition, two gapped systems are in the same phase if their ground states are related by a finite-depth unitary circuit, perhaps after stacking with trivial product states.  With only these ingredients it is not clear how to describe mobility of excitations, and the simplest way to incorporate this is to impose lattice translation symmetry.  More generally, even apart from fracton systems, some kind of homogeneity condition needs to be added to the standard definition to avoid pathologies;\footnote{An example of the kind of pathology to be avoided is a system comprised of a gapped boundary between two different semi-infinite gapped phases. While such a system is perfectly physical, there is no way to decide which phase it belongs to.} translation invariance is the simplest way to achieve this.  This leads to the notion of \emph{translation-invariant gapped phases}, or, more specifically for fracton systems, \emph{translation-invariant fracton phases}.  

While translation-invariant fracton phases have the advantage of a simple definition, there are some significant disadvantages.  One stems from the fact that imposing translation symmetry leads to symmetry-enriched quantum orders \cite{Wen_2002,Wen_2003,Kitaev_2006,Essin_2013,Mesaros_2013,Barkeshli_2019}.  The basic picture is that in the presence of some symmetry, including translation, there can be multiple phases that all become the same when the symmetry is broken. In the context of topological phases, these are referred to as symmetry enriched topological (SET) phases.  Similar phenomena occur in fracton systems \cite{You_2020, Pretko_2020}, and, in particular, there are translation-symmetry enriched fracton orders \cite{Pretko_2020}.  The phenomenon of symmetry enrichment presents an obstacle to understanding fracton orders. It is conceptually simpler to first understand phases in the absence of symmetry, where there are fewer possibilities, and then later to understand the possible symmetry enrichments. Indeed, this is how the theory of SET phases is organized. The most obvious solution is to break or ignore translation symmetry, but it is then unclear how to describe excitation mobility, which we believe is crucial in fracton systems.

One way to address this issue is to impose translation symmetry in a coarse manner \cite{Haah_2013,Haah_2014,Pai_2019,Hermele_KITP_talk,Hermele_UQM_talk,Hermele_inprep}. This means that given a translation group $\T \cong \Z^3$, we always allow for breaking the symmetry down to a subgroup $\T' \subset \T$, so long as also $\T' \cong \Z^3$.  This corresponds to allowing for finite, but perhaps large, enlargements of the crystalline unit cell.  Coarse translation symmetry still allows us to describe mobility of excitations over long distances, but eliminates at least some invariants associated with translation symmetry enrichment.

A separate disadvantage of translation-invariant fracton phases -- even if translation symmetry is only imposed coarsely -- is that they do not interact nicely with the foliated renormalization group (RG) under which some simple fracton models are fixed points \cite{Shirley_2018,Wang_2023,Hermele_UQM_talk, Hermele_inprep}.  In particular, foliated RG coarse graining operations, which involve integrating out two-dimensional topologically ordered layers, do not stay within the same phase.  This is undesirable because it spoils the useful association between phases and RG fixed points that underlies much of our thinking about universality in condensed matter physics.  

An elegant way to resolve this issue is to consider \emph{foliated fracton phases} \cite{Shirley_2018}, where decoupled two-dimensional layers are treated as a free resource.  This means that any stack of decoupled two-dimensional layers is trivial as a foliated fracton phase, even though it is not finite-depth-unitary equivalent to a product state.  However, foliated fracton phases have their own disadvantages.  First, foliated phase-equivalence is coarser than our usual intuitive notions of phase equivalence; this is a disadvantage if one desires a definition of phase that both hews to our familiar intuition and interacts nicely with RG.  Second, foliated phase-equivalence is only a useful definition for certain gapped systems, namely those with foliated fracton order.  For instance, if one wants to study topological phases, or Type II fracton orders such as Haah's cubic code, it is more useful to choose a different definition.  This raises the prospect of using different definitions of phase depending on which kind of system one is studying, which, again, is far from familiar intuition about phases of matter.

In Refs.~\onlinecite{Hermele_UQM_talk,Hermele_inprep}, these issues led one of us to propose the notion of \emph{deformable-lattice (DL) phases}.  The key idea is to allow certain continuous deformations of the crystalline lattice, crucially including uniform scale transformations, while also imposing translation invariance in a coarse manner.  This results in a definition of gapped phase motivated by physical intuition that also interacts nicely with RG for a large class of gapped systems, including topological phases, foliated fracton orders, and at least some Type II fracton orders (including Haah's code).

Let us extract the key information from the above discussion for our present purposes of constructing phase invariants.  It is convenient to think of each of the above definitions of phase as an equivalence relation generated by a set of operations (and inverse operations) that relate one system to another. Rather than enumerating and giving a formal definition of all the relevant operations for each type of phase, we focus on two that have a particularly significant effect on the $p$-theory.

The first operation is \emph{coarsening}, where we lower the translation symmetry to a subgroup as described above. This operation is involved in the definition of DL phases, and could also be used in a formal treatment of foliated fracton phases. In terms of $p$-theories, this is a move we denote $P \to_c P'$ where $P = (\T, \eta, S, \{ \hat{n}_a \}, \{ \langle \cdot , \cdot \rangle_a \} )$ and $P' = (\T', \eta', S, \{ \hat{n}_a \}, \{ \langle \cdot , \cdot \rangle_a \} )$ are $p$-theories, related by the conditions $\T' \subset \T$ and  $\eta' = \eta |_{\T'}$, \emph{i.e.} $\eta'$ is the restriction of $\eta$ to $\T'$. The ring $\Z_n \T'$ acts on $S$ because $\Z_n \T' \subset \Z_n \T$, so $S$ can be viewed as a $\Z_n \T'$-module. In different language, the $\Z_n \T'$-module appearing in $P'$ is $i^* S$, where $i : \T' \hookrightarrow \T$ is the inclusion.

A special choice of the subgroup $\T' \subset \T$, which will be used when discussing RG in Sections~\ref{section:RG} and~\ref{section:examples}, is
\begin{equation}
\T^{(m)} = \{ g \in \T \big| g = h^m \text{ for some } h \in \T \} \text{,}
\end{equation}
where $m \geq 2$. Note that $\T^{(m)}$ is the subgroup of $\T$ corresponding to an isotropic enlargement of the unit cell by a factor of $m$ in linear dimension. For $\T' = \T^{(m)}$, we denote the coarsening operation by $P \to_m P'$.

We observe that, given two $p$-theories related by coarsening, if $x \in S$ is a planon in one theory, then it is a planon in the other theory. In fact, more generally, the dimension of the subspace in which $x$ is mobile is the same in both theories. This can be proved using Corollaries~\ref{cor:dim} and~\ref{cor:int-dim} of Appendix~\ref{appendix:technical}. Letting $\T_x \subset \T$ and $\T'_x \subset \T'$ be the subgroups of $\T$ and $\T'$ leaving $x \in S$ invariant, we have $\T_x \cong \Z^m$, $\T'_x \cong \Z^{m'}$ and $\T'_x = \T_x \cap \T'$. By the results of Appendix~\ref{appendix:technical},
\begin{equation}
m' = \operatorname{dim}\Span (\eta(\T'_x)) = 
\operatorname{dim}\Span (\eta(\T_x)) \cap \Span (\eta(\T'))
= \operatorname{dim}\Span (\eta(\T_x)) = m \text{,}
\end{equation}
where the third equality follows because $\Span (\eta(\T')) = \R^3$. In addition, the orientation $a(x)$ is the same in both theories. This means that a planon $p \in S$ has the same plane of mobility in both theories, so both $p$-theories have the same set of orientations, and it makes sense to choose the normal vectors and bilinear forms to be the same in both $P$ and $P'$.

The weighted QSS are invariant under coarsening, as is the underlying abelian group of the constraint module.  The discussion above tells us that $S_a$, the submodule of $a$-oriented planons, has the same underlying abelian group in both $p$-theories. (The module structure of $S_a$ is different because the translation groups are different.) Forgetting about module structure, $Q$ and $\pi : S \to Q$ are also the same, with the same decomposition $Q = \oplus_{a \in A} Q_a$. This implies that $C = Q / \pi(S)$ is unaffected by coarsening if we ignore the module structure. Moreover, we have $w(x) = w'(x)$, where $w$ and $w'$ denote the weight of $x \in S$ in the two $p$-theories.  It follows that $S^{(i)} \subset S$ are also the same in both theories, again in the sense that their underlying abelian groups are equal. This gives the same weighted QSS ${\cal W}_i$ before and after coarsening.

The second operation we consider is \emph{foliated stacking}, \emph{i.e.} stacking with decoupled gapped 2d layers. This operation is not part of the definition of standard or DL phases, but plays the key role in foliated fracton phases. Formally, in terms of $p$-theories, the input data for the operation consists of two $p$-theories $P = (\T, \eta, S, \{ \hat{n}_a \}, \{ \langle \cdot , \cdot \rangle_a \} )$ and $P_\ell = (\T, \eta, S_\ell, \{ \hat{n}^\ell_b \}, \{ \langle \cdot , \cdot \rangle^\ell_b \} )$. The two sets of orientations are $A$ and $A_\ell$, respectively, with $a \in A$ and $b \in A_\ell$. Note that we require $\T$ and $\eta$ to be the same in both theories. If $a \in A$ and $a \in A_\ell$, we require $\hat{n}_a = \hat{n}^\ell_a$. While $P$ is arbitrary, $P_\ell$ is a $p$-theory of decoupled 2d layers. More formally, $S_\ell$ is generated by planons, and moreover the mutual statistics is required to be that of a stack of decoupled 2d layers.\footnote{This last condition is needed, because there are fracton models with only planon excitations that are not unitarily equivalent to a stack of decoupled 2d layers \cite{Shirley_2020}.} Given this input data, foliated stacking is simply the operation $P \to_f P' = P \ominus P_\ell$. That is, $P'$ is nothing but a stack of $P$ and $P_\ell$, as defined in Sec.~\ref{section:definitions}. Because $P_\ell$ has trivial weighted QSS and constraint module, the invariants do not change under foliated stacking.

\section {Entanglement renormalization flows of $p$-modular theories}
\label{section:RG}

The renormalization group (RG) underpins the understanding of universality in condensed matter physics, in part by relating quantum systems to long-wavelength effective descriptions that encode universal properties of a phase. Several fracton models are known to be fixed points of exotic entanglement renormalization flows, which differ from standard ones in that one must integrate out some non-trivial quantum systems in order to make the fracton model an RG fixed point \cite{Shirley_2018, Dua_2020, Wang_2023}. Foliated RG is an example of this, where one is allowed to integrate out decoupled 2d layers with non-trivial topological order \cite{Shirley_2018}. In this section, we consider entanglement RG flows of $p$-modular fracton orders.

What can we say about the relationship between foliated and $p$-modular fracton orders? First of all, not every foliated quantum order is $p$-modular. A simple counterexample is the 3d toric code, which is a fixed point of foliated RG,\footnote{The 3d toric code is a fixed point of a trivial foliated RG where one never needs to integrate out non-trivial 2d layers, but only integrates out decoupled qubits.} but has fully mobile point-like excitations that cannot be detected by braiding with planons; indeed, there are no planons at all. While the 3d toric code lacks restricted mobility excitations, if we stack the 3d toric code and the X-cube model we obtain a fracton order that is foliated but not $p$-modular. However, if we restrict to fracton orders without non-trivial fully mobile point-like excitations, it is reasonable to conjecture that every foliated fracton order is $p$-modular. 

This section focuses on the opposite question, namely whether a $p$-modular fracton order is a fixed point of foliated RG. In general, entanglement RG \cite{Vidal_2007} starts with a quantum system schematically denoted ${\cal A}$. The system is an RG fixed point if ${\cal A} \simeq {\cal A}^\kappa \ominus {\cal B}$, where ${\cal B}$ is another quantum system, ${\cal A}^\kappa$ is a stack of $\kappa \geq 1$ copies of ${\cal A}$ (each on a lattice with larger lattice constant than the original system), $\ominus$ denotes stacking, and $\simeq$ indicates that two systems are related to one another by a finite depth unitary quantum circuit. By integrating out the system ${\cal A}^{\kappa-1} \ominus {\cal B}$, we again obtain ${\cal A}$, but on a coarse-grained spatial lattice. RG for fracton systems is exotic in the sense that ${\cal A}^{\kappa-1} \ominus {\cal B}$ is in a non-trivial phase; in any conventional RG procedure, we do not integrate out such systems. If ${\cal A}^{\kappa-1} \ominus {\cal B}$ consists of decoupled 2d layers, we say the RG is foliated.

Now suppose that the system ${\cal A}$ is characterized by some $p$-modular fracton order, and we assume that ${\cal A}$ is an RG fixed point, so that ${\cal A} \simeq {\cal A}^\kappa \ominus {\cal B}$. It follows that ${\cal B}$ must also be $p$-moduar. To focus on the most interesting cases, we further suppose that ${\cal A}$ has some intrinsic excitation of weight $w > 1$; that is, ${\cal A}$ has some non-trivial excitation that is not a composite of planons. It follows that
\begin{equation}
{\cal A} \simeq {\cal A} \ominus {\cal B}_{p} \text{,}
\end{equation}
where the $p$ subscript indicates that ${\cal B}_p$ is essentially 1-planar (\emph{i.e.} its non-trivial exciations are planons or composites of planons). This follows from the fact that the invariant $C = 0$ if and only if a $p$-theory is essentially 1-planar.

This result is a very strong constraint on RG fixed points of $p$-modular fracton orders. Not only must we have $\kappa = 1$, but the system to be integrated out, namely ${\cal B}_p$, is essentially 1-planar. One possibility is that ${\cal B}_p$ consists of decoupled 2d layers, in which case the RG is foliated. However, it is not guaranteed that the RG is foliated because there are 1-planar fracton orders that are not simply decoupled layers \cite{Shirley_2020}, and  ${\cal B}_p$ could be such a system, although we are not presently aware of an example of RG where this phenomenon arises.

Now we turn to the issue of how to study the RG for some particular $p$-modular fracton order. For greater simplicity, we would like to work with $p$-theories instead of quantum systems. Recall that all $p$-theories are assumed to be $p$-modular unless stated otherwise. If ${\cal A}$ is a system with $p$-modular fracton order and ${\cal A} \simeq {\cal A} \ominus {\cal B}_p$ as above, then in terms of $p$-theories we have
\begin{equation}
P_A \to_m P'_A \cong P_A^{(m)} \ominus P_B \text{.} \label{eqn:ptr}
\end{equation}
This expression requires some explanation. Here, $P_A$ is the $p$-theory of ${\cal A}$, which is coarsened to $P'_A$ using $\T' = \T^{(m)} \subset \T$ defined in Sec.~\ref{sec:phase-invariants}, for $m \geq 2$. The coarsening reflects the fact that, in entanglement RG of a translation-invariant system, the unit cell of the finite depth circuit, and thus the new system after applying the circuit, needs to be larger than that of the original system. The coarsened $p$-theory is isomorphic to a stack, where $P_A^{(m)}$ is a scaled up copy of $P_A$, corresponding to taking the system ${\cal A}$ and changing the lattice constant $a \to m a$. Formally, given a $p$-theory $P = (\T, \eta, S, \{ \hat{n}_a \}, \{ \langle \cdot, \cdot \rangle_a \} )$, we define $P^{(m)} = (\T^{(m)}, \eta|_{\T^{(m)}}, S^{(m)}, \{ \hat{n}_a \}, \{ \langle \cdot, \cdot \rangle_a \} )$, where the module $S^{(m)} \equiv \alpha^{-1*}_m S$, with $\alpha_m : \T \to \T^{(m)}$ the isomorphism $\alpha_m(g) = g^m$.\footnote{This notation was explained in Sec.~\ref{section:definitions}; in more detail, $S^{(m)}$ is identical to $S$ as an abelian group but has the $\T^{(m)}$ action given by $g \cdot x = h x$ for any $x \in S^{(m)}$, where $g = h^m$ is an arbitrary element of $\T^{(m)}$ and the expression $hx$ on the right-hand side uses the $\T$-action on $S$.} The unit vectors $\hat{n}_a$ and the bilinear forms are the same in the two $p$-theories. Finally, $P_B$ is an essentially 1-planar $p$-theory.

While it is clear that if ${\cal A} \simeq {\cal A} \ominus {\cal B}_p$, then Eq.~\ref{eqn:ptr} holds, understanding the converse would be more useful, because this would let us work entirely with $p$-theories. Namely, if ${\cal A}$ is a system whose $p$-theory satisfies Eq.~\ref{eqn:ptr}, is it true that ${\cal A} \simeq {\cal A} \ominus {\cal B}_p$? We will argue this indeed holds, relying on the following assumption: given a system ${\cal A}$ with $p$-theory $P_A \cong P_1 \ominus P_2$, if there is no process by which excitations of $P_1$ can remotely detect those of $P_2$ and vice versa, then ${\cal A}$ is phase-equivalent to a stack ${\cal A}_1 \ominus {\cal A}_2$, where ${\cal A}_i$ has $p$-theory $P_i$. We would actually like a finite-depth-circuit equivalence ${\cal A} \simeq {\cal A}_1 \ominus {\cal A}_2$, which might not hold in general; however, we assume this does hold if ${\cal A}$ has a commuting projector Hamiltonian as in all the cases of interest for this paper. 

Now, given a system ${\cal A}$ with $p$-theory satisfying $P \cong P_1 \ominus P_2$, where $P_2$ is essentially 1-planar, we would like to argue ${\cal A} \simeq {\cal A}_1 \ominus {\cal A}_2$. Remote detection processes where excitations of $P_2$ detect other excitations are nothing but planon braiding processes, so the fact that $P$ is a stack tells us that excitations of $P_1$ cannot be detected in this way. The desired result follows from the assumption above if we can argue that, in addition, there is no process by which excitations of $P_1$ can remotely detect excitations of $P_2$.

A general process by which excitations of $P_1$ detect a planon $p \in S_2$ can be described as follows and is illustrated in Fig.~\ref{fig:generic-detection-process}.  We let $| \psi \rangle$ be the state with a single planon $p$ and no other excitations above the ground state. The planon is located at position $\boldsymbol{r}_i$, a large but finite distance away from the origin $\boldsymbol{r} = 0$. (1) Starting from $| \psi \rangle$, we move $p$ to the origin. (2) Next, create, move and destroy excitations of $P_1$ within a finite but large thickened shell $\Sigma$ consisting of lattice points $\boldsymbol{r}$ with $r_1 < |\boldsymbol{r}| < r_2 \ll |\boldsymbol{r}_i|$; in Fig.~\ref{fig:generic-detection-process} we represent $\Sigma$ as the surface of a cube for ease of illustration. (3) Finally, we move $p$ back to its initial position $\boldsymbol{r}_i$. Each step is achieved by acting with operators ${\cal S}_p$, ${\cal D}$ and ${\cal S}^{-1}_p$, respectively, where ${\cal S}_p$ is a string operator moving $p$ from $\boldsymbol{r}_i$ to the origin, and ${\cal D}$ is supported within $\Sigma$. The final state after the process described is thus ${\cal S}^{-1}_p {\cal D} {\cal S}_p | \psi \rangle$, which is the same as the initial state $|\psi\rangle$ up to a phase factor; however, this phase factor is not interesting because it includes local phase factors accumulated during step (2). Instead we compare to the state ${\cal D} {\cal S}^{-1}_p {\cal S}_p | \psi \rangle$, corresponding to a process where we change the order of the steps by first doing (1), then (3), and then (2). This process accumulates the same local phase factors as the first one, but there is no remote detection of $p$ by ${\cal D}$. We have
\begin{equation}
{\cal S}^{-1}_p {\cal D} {\cal S}_p | \psi \rangle = e^{i \phi} {\cal D} {\cal S}^{-1}_p {\cal S}_p | \psi \rangle = e^{i \phi} {\cal D} | \psi \rangle \text{,}
\end{equation}
where $\phi$ is a statistical phase factor, and $\phi \neq 0$ indicates non-trivial remote detection of $p$. This expression can be rewritten
\begin{equation}
{\cal D} {\cal S}_p | \psi \rangle = e^{i \phi} {\cal S}_p {\cal D} | \psi \rangle \text{;} 
\end{equation}
that is, $\phi \neq 0$ means ${\cal S}_p$ and ${\cal D}$ do not commute acting on $|\psi\rangle$.  We will assume $\phi \neq 0$ and obtain a contradiction.

The operator ${\cal D}$ may be a finite depth circuit or more generally a finite-time adiabatic evolution under some local Hamiltonian; in either case it can be sensibly truncated to a subregion of its support. Acting on $| \psi \rangle$, or on $S_p | \psi \rangle$, or on a ground state, ${\cal D}$ creates no excitations. However, if we truncate ${\cal D} \to {\cal D}_F$, where $F$ is a subregion of $\Sigma$ as shown in Fig.~\ref{fig:generic-detection-process}, then ${\cal D}_F$ creates excitations of $P_1$ on the region $T \subset \partial F$ where non-trivial truncation occurred.\footnote{This is where we use the assumption that ${\cal D}$ implements a remote detection process using excitations of $P_1$.} We choose $F$ to contain the intersection of $\Sigma$ with the support of ${\cal S}_p$, which implies
\begin{equation}
{\cal D}_F {\cal S}_p | \psi \rangle = e^{i \phi} {\cal S}_p {\cal D}_F | \psi \rangle \text{.}
\end{equation}
This implies that we can construct a process where $p$ braids with some of the excitations created by ${\cal D}_F$ with mutual statistical angle $\phi \neq 0$, a contradiction. That process is constructed using a different string operator ${\cal S}'_p$ that moves $p$ from the origin back to $\boldsymbol{r}_i$ without passing through $F$, so that ${\cal S}_{loop} = {\cal S}'_p {\cal S}_p$ transports $p$ around a closed loop, and ${\cal D}_F {\cal S}_{loop} | \psi \rangle = e^{i \phi} {\cal S}_{loop} {\cal D}_F | \psi \rangle$.

\begin{figure}
    \centering
    \includegraphics[width=0.5\linewidth]{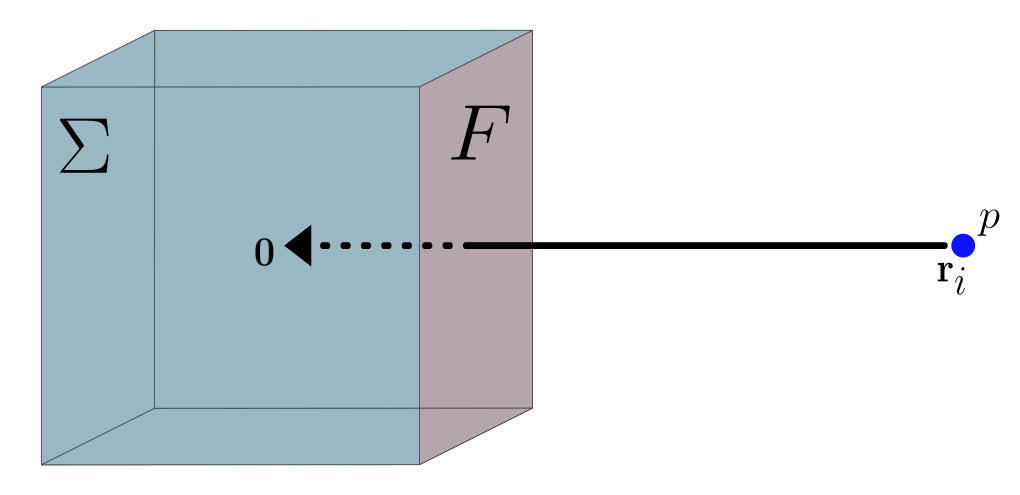}
    \caption{Illustration of a generic process for remotely detecting a planon $p$ by other excitations. The planon is first moved from its initial position at $\mathbf r_i$ to the origin (black arrow), then other excitations are created, moved and annihilated within a shell of finite thickness, shown for ease of illustration as the surface of a cube $\Sigma$. Finally the planon is returned to its initial position in the inverse of the first step (not shown). The region $F$ is the face of the cubical shell that intersects the path of the planon string operator. See text for more details.}
    \label{fig:generic-detection-process}
\end{figure}

Now that we have justified working with $p$-theories, we would like to understand when a $p$-theory is an RG fixed point, in the sense that Eq.~(\ref{eqn:ptr}) holds.  Moreover we would like to be able to determine whether the theory $P_B$ consists of decoupled 2d layers. We do not  know the answer in full generality, but in Appendix~\ref{appendix:foliated-RG} we provide  sufficient conditions that apply in all the examples considered in this paper, which we briefly describe here.

We can always coarsen the translation symmetry so that $\T / \T_a \cong \Z$ for each orientation $a \in A$. Upon doing this we have $R_a = \Z_n \T / \T_a \cong \Z_n[t^\pm]$. In Proposition~\ref{prop:rg-2} of Appendix~\ref{appendix:foliated-RG}, we show that if, for each orientation $a$:
\begin{enumerate}
    \item The module of $a$-oriented plane charges $Q_a$ is a free $R_a$-module; and
    \item There is an $R_a$-module basis $q_1, \dots, q_r$ for $Q_a$ together with a basis $\Delta_1 q_1, \dots, \Delta_r q_r$ for $\ker \omega|_{Q_a}$, for some $\Delta_i \in R_a$ whose highest and lowest degree coefficients are units in $\Z_n$, such that $\langle t^k \Delta_i q_i, t^\ell q_j \rangle_a = \delta_{k \ell} F_{ij}$,
\end{enumerate}
Then there is an integer $m$ such that upon further coarsening the translation symmetry to $\T^{(m)} \subset \T$ we have $P \to_m P' \cong P^{(m)} \ominus P_2$, where $P_2$ is essentially $1$-planar. The integer $m$ is any $m > 1$ such that $\T^{(m-1)}$ acts trivially on the constraint module $C$; because $C$ is finite, such an integer always exists. Physically, the second condition gives a division of planons and plane charges into layers, so that planons only detect plane charges in the same layer; note that $t^k \Delta_i q_i$ is  a translation of the basis planon $\Delta_i q_i$ by $k$ layers away from a reference layer.

Proposition~\ref{prop:foliated-rg} gives a condition that ensures $P_2$ consists of layers of 2d toric codes, and hence the RG is foliated. Again we first coarsen translation symmetry so that $\T / \T_a \cong \Z$ for each $a \in A$. If we choose a basis for $Q_a$ and $q \in Q_a$ is a basis element, then $q^* \in Q_a^*$ denotes the dual basis element. That is, $q^*(q) = 1 \in R_a$ and $q^*(q') = 0$ if $q' \neq q$ is another basis element. If, for each orientation $a \in A$:
\begin{enumerate}
\item $Q_a$ is a free $R_a$-module; and
\item There is an $R_a$-module basis $q_{\fe 1}, q_{\fm 1}, \dots, q_{\fe r_a}$, $q_{\fm r_a}$ for $Q_a$, and elements $\Delta_{ai} \in R_a$ for $i=1,\dots,r_a$ whose highest and lowest degree coefficients are units in $\Z_n$, such that $\Delta_{a1} q_{\fe 1}, \bar{\Delta}_{a1} q_{\fm 1},
\dots, \Delta_{a r_a} q_{\fe r_a}$, $\bar{\Delta}_{a r_a} q_{\fm r_a}$ is an $R_a$-module basis for $\operatorname{ker} \omega|_{Q_a}$, and also $\Phi_a(\Delta_{ai} q_{\fe i}) = q^*_{\fm i}$ 
and $\Phi_a(\bar{\Delta}_{ai} q_{\fm i}) = q^*_{\fe i}$,
\end{enumerate}
Then, up to isomorphism, $P \to_m P' \cong P^{(m)} \ominus P_2$, where $m > 1$ is as above, and $P_2 = \ominus_{a \in A} P_{\ell,a}$, where $P_{\ell,a}$ is a stack of $r_a(m-1)$ copies of $a$-oriented $\Z_n$ toric code layers per unit cell of $\T^{(m)}$ translation symmetry. Here we have assumed that each $Q_a$ breaks up into ``electric'' ($\fe$) and ``magnetic'' ($\fm$) sectors such that electric excitations only detect magnetic excitations and vice versa. This is a natural set of assumptions in CSS Pauli codes, where $Z$-sector excitations only detect $X$-sector excitations and vice versa. Of course, that is not the only assumption made, and there may be $p$-modular CSS codes where the assumptions of Proposition~\ref{prop:foliated-rg} do not hold; however, they do hold for the examples studied in this paper.

\section {Examples}
\label{section:examples}

In this section we study a handful of fracton models and show they are $p$-modular. For each example, we describe the module of superselection sectors in terms of plane charges and compute the weighted QSS and constraint module. We also use the results of Appendix~\ref{appendix:foliated-RG}, summarized above in Sec.~\ref{section:RG}, to show that the $p$-theory admits a foliated RG, \emph{i.e.} we find a decomposition of the form Eq.~\ref{eqn:ptr}. Subject to the assumptions discussed in Sec.~\ref{section:RG}, this shows that these examples are all foliated fracton orders. We note that the $\Z_n$ X-cube model \cite{Shirley_2018} and $\Z_2$ anisotropic model \cite{Shirley_2019} were previously shown to be foliated. The $\Z_2$ checkerboard model is foliated as it is a stack of two X-cube models \cite{Shirley_2019_checkerboard}, and the toric code layers model is trivially foliated. Ref.~\onlinecite{Shirley_2023} showed foliation of the Chamon model, which is very closely related to the 4-planar X-cube model.

All of the models we consider are $\Z_n$ translationally invariant commuting Pauli Hamiltonians, also referred to as Pauli codes. Such models are written in terms of clock ($Z$) and shift ($X$) $\Z_n$ Pauli operators, satisfying the commutation relation $Z X = e^{2\pi i / n} X Z$. The Hamiltonian is of the form $H = - \sum_\alpha (h_\alpha + h^\dagger_\alpha)$, where each stabilizer $h_\alpha$ is a product of Pauli operators, and all the $h_\alpha$'s are mutually commuting. Because the stabilizers satisfy $h_\alpha^n = 1$, their eigenvalues are $n$th roots of unity, \emph{i.e.} they are of the form $e^{2 \pi i m / n}$ for $m \in \Z_n$. We often refer to stabilizer eigenvalues by $m \in \Z_n$ rather than its exponential. Throughout the paper, we consider systems defined in infinite space.

We always assume that the stabilizers have the property that if $(h_\alpha)^\ell = 1$, then $\ell$ is a multiple of $n$, which always holds if $n$ is prime, and holds in all our examples for arbitrary $n$. This is equivalent to assuming that the set of eigenvalues of each stabilizer is all the $n$th roots of unity. 

We employ the formalism developed in Ref.~\onlinecite{Haah_2013} to describe Pauli codes, which we now review. Throughout this section, as above, we set $R = \Z_n \T$ with $\T \cong  \Z^3$ the lattice translational symmetry group, where $n$ is the Hilbert space dimension of a single qudit. Usually, we take $\T$ to be the free abelian multiplicative group on three variables $x,y,z$, or equivalently the multiplicative group of ``Laurent monomials'' in $x,y,z$, where arbitrary integer powers are allowed to appear.  The element $x^a y^b z^c$ with $a,b,c \in \Z$ represents the cubic lattice translation $a \hat x + b \hat y + c \hat z$. With this choice, $R = \Z_n[x^{\pm}, y^{\pm}, z^{\pm}]$, the ring of Laurent polynomials in the variables $x,y,z$ with $\Z_n$ coefficients. For the FCC and checkerboard models, which are defined on a face centered cubic (fcc) lattice, we instead take $\T$ to be the free abelian group generated by the monomials $xy, yz, zx$. In this case $R$ is a subring of the above-mentioned Laurent polynomial ring.

Translation invariant, locally topologically ordered Pauli codes can be specified by an exact sequence of free,  finitely generated $R$-modules\footnote{Here we use the assumption that the set of eigenvalues of each stabilizer is all $n$th roots of unity. If there is a stabilizer that has only a proper subset of $n$th roots of unity as eigenvalues, which is possible if $n$ is not prime, then $G$ and $E$ are not free. An example is a stabilizer that is a product of $Z^2$ Pauli operators when $n=4$.}
\begin{align}
    G \xrightarrow{\sigma} P \xrightarrow{\epsilon = \sigma^\dagger \lambda} E . \label{eqn:exact-pauli-sequence}
\end{align}
The module $P$ is the Pauli module, and its elements correspond bijectively to finite products of Pauli operators modulo phase factors. $P$ can be obtained as the abelianization of the multiplicative Pauli group $\mathcal P$ generated by $Z$ and $X$ Pauli operators; addition in $P$ corresponds to (abelianized) operator multiplication in $\mathcal P$. It is important to be aware that elements of $P$ are not precisely Pauli operators, because information about the overall phase is lost. In order for Eq.~\ref{eqn:exact-pauli-sequence} to unambiguously specify a Pauli code, this phase information needs to be recovered somehow for the stabilizers $h_\alpha$; one option is to choose a Pauli operator in $\mathcal P$ for each element of $P$, and we do this by putting all $Z$ Paulis first, followed by $X$ operators, with an overall phase of unity. If there are $q$ qudits per unit cell, then $P = R^{2 q}$, and elements can be written as $2q$-component column vectors with entries in $R$. We take the first $q$ entries to correspond to $Z$ Pauli operators, with the second $q$ entries corresponding to $X$ operators.

Next, elements of $G$ are formal sums of the local stabilizer operators that appear in the Hamiltonian. That is, a general element of $G$ is of the form $\sum_\alpha c_\alpha h_\alpha$, where $c_\alpha \in \Z_n$ and in this expression the $h_\alpha$'s should be understood as formal objects and not operators acting on Hilbert space. The map $\sigma$ sends each stabilizer $h_\alpha$ to the corresponding product of Pauli operators in $P$. If there are $t$ local stabilizers per unit cell, $G = R^t$ and elements can be written as $t$-component column vectors. In this paper, we only consider models where it is possible to make a translation-invariant choice of Hamiltonian stabilizers $h_\alpha$ that do not satisfy any local constraints. That is, there is no non-trivial equation of the form $h^{n_1}_{\alpha_1} h^{n_2}_{\alpha_2} \cdots h^{n_k}_{\alpha_k} = e^{i \phi} \mathbbm{1}$, where the $n_i$ are non-zero elements of $\Z_n$. Equivalently, the map $\sigma$ is injective. This is not possible for arbitrary Pauli codes, as in the 3d toric code.

The excitation module $E$ is identical to $G$, \emph{i.e.} $E = G = R^t$. However the physical interpretation of its elements is different. Elements of $E$ are patterns of eigenvalues of the $h_\alpha$ stabilizers, with only finitely many non-zero eigenvalues. As expected from the lack of local constraints on stabilizers, and as can be verified in the examples we consider, elements of $E$ are in bijective correspondence with excitations above the ground state (energy eigenstates) of bounded spatial support. The excitation map $\epsilon : P \to E$ sends an abelianized Pauli operator to the excitation that it creates when acting on the ground state. The maps $\epsilon$ and $\sigma$ are related by $\epsilon = \sigma^\dagger \lambda$, where $\lambda : P \to P$ is a map coming from the symplectic form on $P$ describing commutations among Pauli operators (see Ref.~\onlinecite{Haah_2013}). Here, for an $R$-valued matrix $f$, $f^\dagger$ is the transpose and element-wise antipode of $f$. Moreover, with our conventions, $\lambda$ is a $2q \times 2q$ matrix which can be written as follows in block form:
\begin{equation}
\lambda = \left( \begin{array}{cc}
 0 &  \mathbf 1 \\ - \mathbf 1 & 0
 \end{array}\right) \text{.}
\end{equation}

By choosing $G$ such that stabilizers obey no local constraints, and $E$ such that all elements are physical excitations in infinite space, we sidestep some of the technical complications that arise in Ref.~\onlinecite{Haah_2013}. In particular, the module of superselection sectors (referred to in Ref.~\onlinecite{Haah_2013} as topological charges) is given by $S = \operatorname{coker} \epsilon \equiv E / \operatorname{im} \epsilon$, without any need to take the torsion submodule of $\operatorname{coker} \epsilon$. Given an excitation $x \in E$, we denote the superselection sector it belongs to by $[x] \in S$.

Exactness of the sequence Eq.~\ref{eqn:exact-pauli-sequence} at $P$ is the condition $\operatorname{im} \sigma = \operatorname{ker} \epsilon$. Physically, the condition $\operatorname{im} \sigma \subset \operatorname{ker} \epsilon$ is the statement that stabilizers create no excitations, and thus commute with one another; this property holds by definition for any Pauli stabilizer code. The model is \emph{locally topologically ordered} when also $\operatorname{ker} \epsilon \subset \operatorname{im} \sigma$, which tells us that any local operator creating no excitations is a stabilizer (or a product of a finite number of stabilizers). In particular, when working with large finite systems satisfying periodic boundary conditions, this implies that no local operator can distinguish between degenerate ground states. Codes satisfying the local topological order property are also referred to as exact codes, because this property is the same as exactness of Eq.~\ref{eqn:exact-pauli-sequence}. We are only interested in exact codes in this paper.

For a $p$-modular Pauli code, the above sequence extends to an exact sequence of the form
\begin{align}
    0 \to G \xrightarrow{\sigma} P \xrightarrow{\epsilon} E \xrightarrow{\pi} Q \xrightarrow{\omega} C \to 0,
    \label{eq:code-sequence}
\end{align}
where $Q$ and $C$ are not free $R$-modules, but instead have the structure described in Sec.~\ref{section:definitions}. 
We call Eq.~\ref{eq:code-sequence} the $p$-modular sequence of the code.  In a slight abuse of notation, we use the same symbol $\pi : E \to Q$ as for the induced map $\pi : S = E / \operatorname{im} \epsilon \to Q$. Exactness at $E$ is the statement that excitations with trivial plane charges are exactly locally createable excitations, and is equivalent to injectivity of $\pi : S \to Q$ and thus to $p$-modularity. Exactness at $Q$ means that exactly those configurations $q$ of plane charges realized by some excitation satisfy the constraints expressed by $\omega(q) = 0$.

All the examples discussed in this paper are CSS codes, which means that each stabilizer is a product either of $Z$ operators or of $X$ operators. The $p$-modular sequence of a CSS code decomposes into a direct sum of two exact sequences, with $G = G^Z \oplus G^X$, $P = P^Z \oplus P^X$, $\sigma = \sigma^Z \oplus \sigma^X$, $\epsilon = \epsilon^Z \oplus \epsilon^X$, and so on for the other modules and maps. Elements of $P^Z$ correspond to products of $Z$ Pauli operators, while elements of $G^Z$ label $Z$-stabilizers, and $E^Z$ describes configurations of $Z$-excitations (\emph{i.e.} excitations of $Z$-stabilizers), which in turn carry only $Z$-plane charges (elements of $Q^Z$). The corresponding statements hold for the modules with $X$ superscripts. Maps with $Z$ ($X$) superscripts map between modules with the same superscripts, with the exception of $\epsilon^Z : P^X \to E^Z$ and $\epsilon^X : P^Z \to E^X$, reflecting the fact that $Z$-Paulis create $X$-excitations and vice versa. While this results in two separate exact sequences, the plane charge structure of $Z$-excitations is determined by braiding with $X$-planons and vice versa. The module of superselection sectors is also a direct sum
\begin{equation}
S = \frac{E}{\operatorname{im} \epsilon} = \frac{E^Z}{\operatorname{im} \epsilon^Z} \oplus \frac{E^X}{\operatorname{im} \epsilon^X} = S^Z \oplus S^X \text{.}
\end{equation}

For each example, we obtain the sequence Eq.~\ref{eq:code-sequence} and use it to determine the phase invariants and the behavior under RG. The general strategy we pursue is first to identify a set of orientations $A$ and putative submodules of $a$-oriented planons $\tilde{S}_a \subset S_a$. We do not know \emph{a priori} that we have found all orientations of planons, or all the planons of a given orientation; that is, it may be the case that $\tilde{S}_a \neq S_a$. Next, we use the submodules $\tilde{S}_a$ to obtain a putative plane charge structure. Letting $\tilde{K_a} \subset S$ be the submodule of topological charges transparent to braiding with planons in $\tilde{S}_a$, we define the putative plane charge modules $\tilde{Q}_a = S / \tilde{K_a}$, with associated quotient maps $\tilde{\pi}_a : S \to \tilde{Q}_a$. Moreover we have $K_a \subset \tilde{K_a} \subset S$, implying
\begin{equation}
\tilde{Q}_a = \frac{S/K_a}{\tilde{K}_a / K_a} = \frac{Q_a}{\tilde{K}_a / K_a} \text{,}
\end{equation}
so there is a quotient map $\tau_a : Q_a \to \tilde{Q}_a$, and $\tilde{\pi}_a = \tau_a \circ \pi_a$. We define $\tilde{Q} = \bigoplus_{a \in A} \tilde{Q}_a$, which comes with inclusions $\tilde{\iota}_a : \tilde{Q}_a \hookrightarrow \tilde{Q}$, and moreover we define  $\tilde{\pi} : S \to \tilde{Q}$ by $\tilde{\pi} = \sum_{a \in A} \tilde{\iota}_a \circ \tilde{\pi}_a$. Defining $\tau : Q \to \tilde{Q}$ by $\tau = \bigoplus_{a \in A} \tau_a$, we have $\tilde{\pi} = \tau \circ \pi$. It follows immediately that if $\tilde{\pi}$ is injective, then so is $\pi$, and thus the model is $p$-modular. At this point we can complete the putative $p$-modular sequence by defining $\tilde{C} = \tilde{Q} / \operatorname{im} \tilde{\pi}$, with associated quotient map $\tilde{\omega} : \tilde{Q} \to \tilde{C}$.

In fact, we can prove that if $\tilde{\pi}$ is injective, then $\tilde{Q}_a = Q_a$, and moreover the entire putative $p$-modular sequence is identical to the true one. As noted above, $K_a \subset \tilde{K}_a$ for all $a \in A$. Suppose that for some $a \in A$, $K_a$ is a proper submodule of $\tilde{K}_a$, so there exists $x \in \tilde{K}_a$ with $x \notin K_a$. Then $\tilde{\pi}_a(x) = 0$ but $\pi_a(x) \neq 0$. By Corollary~\ref{cor:alpha_a-planon}, there exists $\alpha_a \in R$ such that $\alpha_a x$ is a non-trivial $a$-planon. But $\tilde{\pi}_a(\alpha_a x) = \alpha_a \tilde{\pi}_a(x) = 0$, implying $\tilde{\pi}(\alpha_a x) = 0$ and contradicting injectivity of $\tilde{\pi}$. Therefore $K_a = \tilde{K}_a$ for all $a \in A$ and the $p$-modular sequences agree.  Finally, if desired, once we have the $p$-modular sequence, we can go back and check that we indeed found all the planons.

Before going into the examples, we introduce some notation. For a subgroup $\G \subset \T$, we denote $R_\G \equiv \Z_n \T / \G$, and let $p_\G : R \to R_\G$ be the natural projection of group rings. We view $R_\G$ as an $R$-module, with the $R$-action $R \times R_\G \to R_\G$ given by $(r, r_\G) \mapsto p_\G(r) r_\G$. We also let $c_\mathcal G : R_\mathcal G \to R_\T$ project the rest of the way down; note that $R_\T$ is just $\Z_n$ with a trivial action of $R$. Explicitly, a general element $r \in R$ is given by $\sum_{g \in \T} a_g g$ with $a_g \in \Z_n$. Then $p_\G(r) = \sum_{[g] \in \T/\G} (\sum_{g \in [g]} a_g) [g]$, and $c_\G(p_\G(r)) = \sum_{g \in \T} a_g$. It is useful to note that the kernel of $p_\mathcal G$ is the ideal generated by elements of $1 -\mathcal G$. In particular, if $\mathcal G$ is generated by $g_1, \dots, g_l$, then 
\begin{align}
    \ker p_\mathcal G = (1 - g_1, \dots, 1 - g_l).
    \label{eq:ker-p}
\end{align}

We let the indices $\mu,\nu,\rho$ run over $x,y,z$, subject to the constraint that $(\mu,\nu,\rho)$ is a cyclic permutation of $(x,y,z)$. We also use the symbols $\mu, \nu, \rho$ to denote the corresponding monomials in $\Z_n[x^{\pm}, y^{\pm}, z^{\pm}]$. To construct the plane charge module $Q$, we start with a family of subgroups $\T_i \subset \T$, with $\T_i \cong \Z^2$ for each subgroup. In some examples, but not always, the $\T_i$'s are the same as $\T_a = \T_{S_a}$, the groups leaving the superselection sector of $a$-oriented planons invariant. We denote $R_i = R_{\T_i}$, $p_i = p_{\T_i}$, $c_i = c_{\T_i}$. We also define
$R_{i_1 \cdots i_k} \equiv R_{\T_{i_1}} \oplus \cdots \oplus R_{\T_{i_k}}$. Given a group homomorphism $\varepsilon_i : \T \to \Z_n^\times$, we can define an $R$-module homomorphism $\varepsilon_i p_i : R \to R_i$ by $g \mapsto \varepsilon_i(g) p_i(g)$.

\subsection {X-cube}
\label{section:x-cube}

We begin with the $\Z_n$ generalization of X-cube model \cite{Vijay_2016, Slagle_2017, Ma_2017, Shirley_2019, Radicevic_2020}. It is not the simplest of our examples, but it is a key example of a foliated and $p$-modular fracton order, and it has sufficient complexity to be a good first illustration of our analysis. We therefore analyze the plane charge structure in some detail. Of course, for $n=2$, the plane charge structure has already been determined in Ref.~\onlinecite{Pai_2019},  but in an \emph{ad hoc} manner; here, we determine the plane charge structure following the ideas laid out in the sections above.

The Hamiltonian is the negative sum of the cube and star stabilizers shown in Fig. \ref{fig:x-cube}(a), plus Hermitian conjugates. The translation group is $\T = \langle x,y,z \rangle$. The Pauli module is $P = R^6$ and the excitation module is $E = R^3$, and the excitation map $\epsilon : R^6 \to R^3$ is given by the $3 \times 6$ matrix
\begin{equation}
\epsilon = \left(\begin{array}{ccc|ccc}
            &&&&  - d_y &  d_z
            \\
            &&&     - d_x &&  d_z
            \\ \hline
            - \bar d_{yz} & - \bar d_{zx} & - \bar d_{xy}
        \end{array}\right) \text{,}
\end{equation}
where $d_\mu = 1-\mu$ and $d_{\mu\nu} = d_\mu d_\nu \in R$. The lower-left block is $\epsilon^X : P^Z \to E^X$ and the upper right block is $\epsilon^Z : P^X \to E^Z$.
This information completely specifies the model. We denote the six standard basis vectors of $P = R^6$ by $\mathbf{z}_x, \mathbf{z}_y, \mathbf{z}_z, \mathbf{x}_x, \mathbf{x}_y, \mathbf{x}_z$, where $\mathbf{z}_\mu$ ($\mathbf{x}_\mu$) represents a $Z$-Pauli ($X$-Pauli) operator living on a $\mu$-oriented edge of the simple cubic lattice. The standard basis vectors of $G = R^3$ are denoted respectively by $\mathbf{A}_x, \mathbf{A}_y, \mathbf{B}$, which represent two star stabilizers and one cube stabilizer, respectively, as shown in Fig.~\ref{fig:x-cube}(a). In $E = R^3$, we denote the standard basis vectors by $\mathbf{a}_x, \mathbf{a}_y, \mathbf{b}$, which represent excitations of the corresponding stabilizers.

\begin{figure}
    \centering
    \includegraphics[width=0.9\textwidth]{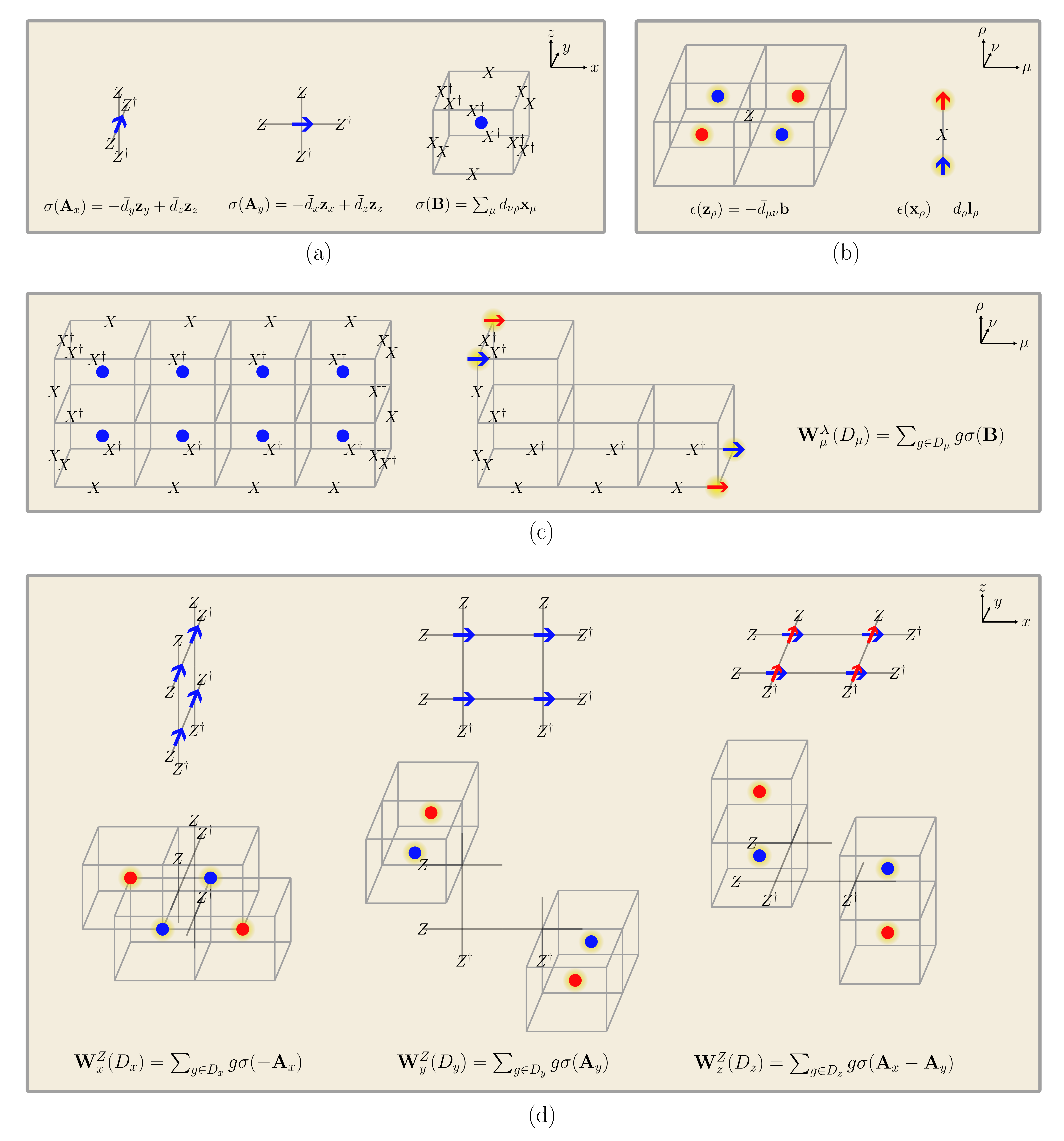}
    \caption{(a) $Z_n$ X-cube stabilizers; each blue symbol
    represents the corresponding product of surrounding Pauli operators. 
    (b) Locally generated excitations; blue (red) circles with yellow highlight denote $+1$ ($-1$) excitations of corresponding cube stabilizer, while blue (red) arrows with yellow highlight denote $+ \mathbf{l}_\rho$ ($-\mathbf{l}_\rho$) lineon exciations at the indicated position. The arrows indicate the lineon's direction of mobility. 
    (c)-(d) Closed and open planon string operators. We always braid planons counterclockwise according to the right-hand rule.
    }
    \label{fig:x-cube}
\end{figure}

It is well known that a nearest neighbor $\mu$-dipole of cube excitations is a planon mobile in the $\nu\rho$-plane. We introduce three orientations of such fracton dipoles denoted $\mathbf{p}^X_\mu = d_{\mu} \mathbf{b} \in E$, where the $X$ superscript reminds us that these are $X$-excitations. We can see that $\mathbf{p}^X_\mu$ is indeed mobile in the $\nu \rho$-plane by observing $d_{\mu\nu} \mathbf b , d_{\rho\mu} \mathbf b \in \im \epsilon$, which implies $[\mathbf{p}^X_\mu] = \nu [\mathbf{p}^X_\mu] = \rho [\mathbf{p}^X_\mu]$.

Turning to $Z$-excitations, we observe that $d_x \mathbf{a}_y = \epsilon(- \mathbf{x}_x)$ and $d_y \mathbf{a}_x = \epsilon(- \mathbf{x}_y)$, implying $\mathbf{l}_x \equiv -\mathbf{a}_y$ and $\mathbf{l}_y \equiv -\mathbf{a}_x$ are mobile along the $x$- and $y$-directions, respectively. In fact, these excitations are lineons, as is well known. In addition, $d_z (\mathbf{a}_x + \mathbf{a}_y) = \epsilon(\mathbf{x}_z)$, implying that $\mathbf{l}_z \equiv (\mathbf{a}_x + \mathbf{a}_y)$ is mobile along $z$ (in fact, it is a lineon). Moreover, the lineon dipoles $\mathbf{p}^Z_\mu = \overbar{d_{\mu}} \mathbf{l}_{\nu}$ are $\nu \rho$-planons, satisfying $[\overbar{d_\mu} \mathbf{l}_\nu ] = [-\overbar{d_\mu} \mathbf{l}_\rho ] $.

To determine the plane charge structure, we need to consider braiding planons around a general excitation ${\mathbf e} = r_{a_x} \mathbf{a}_x + r_{a_y} \mathbf{a_y} + r_b \mathbf{b}$, where $r_{a_x}, r_{a_y}, r_b \in R$. Suppose we transport the lineon dipole $\mu^m \mathbf{p}^Z_\mu$ around the boundary of a disc $D_\mu$ within the planon's plane of mobility. We denote the associated string operator by $\mathbf{W}^X_\mu(D_\mu) \in P$, where the $X$ superscript reflects the fact that $X$-Pauli operators are needed to transport $Z$-excitations. This string operator, as illustrated in Fig.~\ref{fig:x-cube}(c), is a product of cube stabilizers within $D_\mu$; that is,
$\mathbf{W}^X_\mu(D_\mu) = \sum_{g \in D_\mu} g \sigma( \mathbf{B})$. Taking the limit where $D_\mu$ is large, we see that braiding  $\mu^m \mathbf{p}^Z_{\mu}$ around $\mathbf{e}$  simply counts excitations of cube stabilizers within the plane of mobility, by adding up the eigenvalues modulo $n$ of stabilizers within the plane. The statistical phases for braiding with $\mathbf{p}^Z_{\mu}$ and all its translates along the $\mu$-direction are recorded in the value of the map $p_{\mu}(r_b) \in R_\mu$, where $\T_\mu = \T_{S_\mu} = \langle \nu, \rho \rangle$. We thus identify the plane charge modules for $\mu$-oriented $X$-excitations as $Q^X_\mu = R_\mu$.

Similarly, the string operator transporting the fracton dipole $\mu^m \mathbf{p}^X_\mu$ around the boundary of a disc $D_\mu$ is denoted $\mathbf{W}^Z_\mu(D_\mu)$, and this operator is a product of star stabilizers within $D_\mu$, as shown in Fig.~\ref{fig:x-cube}(d). We have $\mathbf{W}^Z_x(D_x) = \sum_{g \in D_x} g \sigma( -\mathbf{A}_x)$ and $\mathbf{W}^Z_x(D_y) = \sum_{g \in D_y} g \sigma( \mathbf{A}_y)$, and finally $\mathbf{W}^Z_\mu(D_z) = \sum_{g \in D_z} g \sigma(\mathbf{A}_x - \mathbf{A}_y)$. The statistical phases for braiding with ${\mathbf p}^X_\mu$ and all its translates are encoded in the values $-p_x(r_{a_x}) \in R_x$, $p_y(r_{a_y}) \in R_y$, and $p_z(r_{a_x}) - p_z(r_{a_y}) \in R_z$, for $\mu = x,y,z$, respectively. We thus identify $Q^Z_\mu = R_{\mu}$.

Putting these results together, we have $Q = Q^Z \oplus Q^X = R_{xyz} \oplus R_{xyz} = R_{xyzxyz}$, with the plane charge map $\pi : E \to Q$ specified in matrix form by
\begin{equation}
\pi =  \left(\begin{array}{cc|c}
       -p_x &
       \\
       & p_y
       \\
        p_z & -p_z
       \\ \hline
       && p_x
       \\
       && p_y
       \\
       && p_z
    \end{array}\right) \text{,}
\end{equation}
where the upper (lower) block is $\pi^Z$ ($\pi^X$), and the $i$th standard basis vector of $Q = R_{xyzxyz}$ is the identity element of the $i$th direct summand.\footnote{Even though we view the direct summands $R_\mu$ as modules over $R$, they come with a distinguished element (the identity) because they are also themselves rings.} We also denote the six standard basis vectors by $\mathbf{q}^Z_x, \mathbf{q}^Z_y, \mathbf{q}^Z_z,\mathbf{q}^X_x, \mathbf{q}^X_y, \mathbf{q}^X_z$. We observe that $\im \pi$ is generated by $\pi(\mathbf{a}_x) = -\mathbf{q}^Z_x + \mathbf{q}^Z_z$, $\pi(\mathbf{a}_y) = \mathbf{q}^Z_y - \mathbf{q}^Z_z$, and $\pi(\mathbf{b}) = \mathbf{q}^X_x + \mathbf{q}^X_y + \mathbf{q}^X_z$. This implies that the planons introduced above are represented in terms of plane charges by $\pi(\mathbf p^Z_\mu) = (1-\bar{\mu}) \mathbf q^Z_\mu$ and $\pi(\mathbf p^X_\mu) = (1-\mu)\mathbf q^X_\mu$. That is, each ``elementary'' planon is represented by adjacent parallel plane charges of opposite signs. 

To complete the $p$-modular sequence, we need to identify the constraint module $C$, and the map $\omega : Q \to C$ that vanishes exactly on physical configurations of plane charges, \emph{i.e.} those lying in $\im \pi \subset Q$.  A constraint is an $R$-linear map $\omega_i : Q \to \R_{\T} \cong \Z_n$ that vanishes on the generators $\pi(\mathbf a_x), \pi(\mathbf a_y), \pi(\mathbf b)$. We find three such constraints, so we choose $C = \R^3_{\T} \cong \Z^3_n$ and express the constraints as the rows of the matrix 
\begin{equation}
\omega =   \left(\begin{array}{ccc|ccc}
            c_x & c_y &  c_z
            \\ \hline
            &&& c_x && - c_z
            \\
            &&&& c_y & - c_z
        \end{array}\right) \text{,}
\end{equation}
where the upper (lower) block is $\omega^Z$ ($\omega^X$).
By construction we have $\im \pi \subset \ker \omega$, which is easily verified by showing that $\omega \pi = 0$ using matrix multiplication and the identity $c_x p_x = c_y p_y = c_z p_z$.

We have thus constructed a $p$-modular sequence of the form Eq.~\ref{eq:code-sequence}. We know the sequence is a chain complex, but we still need to show it is exact. To do this, we  replace $R = \Z_n \T$ with $R' = \Z \T$ in the construction of the modules in the sequence, noting that with this replacement, $R'_\mu = \Z [ \T / \T_\mu ]$ and $R'_\T \cong \Z$. The maps between modules are the same matrices given above. Ignoring the translation action, this gives a sequence of free $\Z$-modules, which has the property that tensoring with $\Z_n$ recovers the original sequence. The new sequence is convenient to work with because $\Z \T \cong \Z[x^\pm, y^\pm, z^\pm]$ is a Laurent polynomial ring over a Euclidean domain, so standard commutative algebra techniques apply \cite{Atiyah_1969,Eisenbud_1995}. If we can show the new sequence is exact, then exactness of the original sequence follows from a universal coefficient theorem, specificially Theorem~3A.3 of Ref.~\onlinecite{hatcherbook}.

We illustrate how to show the new sequence is exact by discussing exactness at $E$, which is particularly important because it is equivalent to injectivity of $\pi : S \to Q$ and thus to $p$-modularity.  Exactness at $E$ is equivalent to $\im \epsilon^X = \ker \pi^X$ and $\im \epsilon^Z = \ker \pi^Z$. Letting $\mathfrak p_\mu = \ker p_\mu$, by Eq.~\ref{eq:ker-p}, we have $\mathfrak p_\mu = (d_\nu, d_\rho)$. The first statement can be shown by computing the following intersection of ideals in $R'$,
\begin{equation}
    \mathfrak p_x \cap \mathfrak p_y \cap \mathfrak p_z
    = (d_y, d_z) \cap (d_z, d_x) \cap (d_x, d_y)
    = (d_{xy}, d_{yz}, d_{zx}),
\end{equation}
which can be shown using Gr\"{o}bner bases. Similarly, the second statement reduces to computing the following intersection of submodules of $E$:
\begin{equation}
(\mathfrak p_y \mathbf a_x + R \mathbf a_y) 
    \cap (\mathfrak p_x \mathbf a_y + R \mathbf a_x) \cap (\mathfrak p_z \mathbf a_x + \mathfrak p_z \mathbf a_y + R(\mathbf a_x + \mathbf a_y))
    = (d_x) \mathbf a_x + (d_y) \mathbf a_y + (d_z) (\mathbf a_x + \mathbf a_y) \text{.}
\end{equation}
We evaluated these intersections using the Macaulay2 software system \cite{M2}. In fact, Macaulay2 works with modules over $\Z[x,y,z]$, \emph{i.e.} with ordinary polynomials of non-negative degree. However, because $\Z\T$ is a localization of $\Z[x,y,z]$, and because localization commutes with taking intersections of submodules, one can do the necessary calculations in $\Z[x,y,z]$; see Chapter 3 of Ref.~\onlinecite{Atiyah_1969} for more information. We also checked directly in Macaulay2 that $\im \epsilon = \ker \pi$, using the fact that localization is an exact functor; exactness of the rest of the $p$-modular sequence was also checked in this way.

Next we work out the weighted QSS. We need to understand the modules $S^{(i)} \subset S$; to avoid cumbersome notation, we identify $S \cong \pi(S) \subset Q$. According to Eq.~\ref{eq:Si-formula},
\begin{equation}
S^{(1)} = \Big( \bigoplus_{\mu} \ker \omega |_{Q^Z_\mu} \Big) \oplus \Big( \bigoplus_\mu \ker \omega |_{Q^X_\mu} \Big) \text{.}
\end{equation}
These restrictions of $\omega$ are the individual columns of $\omega$ in matrix form. For example,
\begin{equation}
\omega|_{Q^Z_\mu} = \left( \begin{array}{c}
c_\mu \\ 0 \\ 0 \end{array}\right).
\end{equation}
It is straightforward to see that $S^{(1)}$ is given precisely by $R$-linear combinations of the planons $\pi(\mathbf p^Z_\mu)$ and $\pi(\mathbf p^X_\mu)$. 

To compute the quotient ${\cal W}_1 = S / S^{(1)}$, we observe that any element $\mathbf q \in Q$ can be written in the form $\mathbf q = \mathbf q_0 + \mathbf p$, where $\mathbf p \in S^{(1)}$ is an $R$-linear combination of $\pi(\mathbf p^Z_\mu)$ and $\pi(\mathbf p^X_\mu)$, and $\mathbf q_0 = \sum_{\mu} a_\mu \mathbf{q}^Z_\mu + \sum_{\mu} b_\mu \mathbf{q}^X_{\mu}$ for $a_\mu, b_\mu \in \Z_n$. We require $\mathbf q \in S$ by imposing the constraint $\omega(\mathbf q) = 0$, which reduces to the equations $a_z = -(a_x + a_y)$ and $b_x = b_y = b_z \equiv b$, so for $\mathbf q \in S$ we have $\mathbf q = \mathbf q_0 + \mathbf p$, where
\begin{equation}
\mathbf q_0 = a_x \mathbf{q}^Z_x + a_y \mathbf{q}^Z_y - (a_x+a_y) \mathbf{q}^Z_x + b \sum_\mu \mathbf{q}^X_\mu \text{.}
\end{equation}
In fact, $\mathbf q_0$ and $\mathbf p$ are uniquely determined by $\mathbf q$, giving a direct sum decomposition $S \cong S^{(1)} \oplus \mathcal W_1$. The key observation is that $\mathbf q_0 \in S^{(1)}$ implies $\mathbf q_0 = 0$, because requiring the maps $\omega|_{Q^{Z,X}_\mu} \circ P_{Q^{Z,X}_{\mu}}$ to annihilate $\mathbf q_0$, where $P_{Q^{Z,X}_\mu} : Q \to Q^{Z,X}_\mu$ is projection onto the direct summand,  gives $a_x = a_y = b = 0$.
This establishes that $\mathcal W_1 \cong \Z_n^3$, where $(a_x, a_y, b)$ parametrize the elements of $\mathcal W_1$.

To compute $\mathcal W_2$, we observe that $\mathbf q^Z_x - \mathbf q^Z_z$ and $\mathbf q^Z_y - \mathbf q^Z_x$ are intrinsic weight-2 excitations, as they are not annihilated by the maps $\omega|_{Q^{Z,X}_\mu} \circ P_{Q^{Z,X}_\mu}$. The above discussion implies that any $\mathbf q \in S$ can be written
\begin{equation}
\mathbf q = b \sum_\mu \mathbf q^X_\mu + \mathbf w_2 \text{,}
\end{equation}
where $\mathbf w_2 \in S^{(2)}$. This decomposition is also unique; that is, given $\mathbf q$, both $b \in \Z_n$ and $\mathbf w_2$ are uniquely determined. To establish this we first observe that $b \sum_\mu \mathbf q^X_\mu \in S^{(1)}$ implies $b = 0$. Moreover, it is straightforward to show that $\ker \omega|_{Q^X_x + Q^X_y} = \ker \omega|_{Q^X_x} \oplus \ker \omega|_{Q^X_y}$, and similarly for the other restrictions to sums of two $Q^x_\mu$; this implies there are no intrinsic $w=2$ $X$-excitations, so $b \sum_\mu \mathbf q^X_\mu \in S^{(2)}$ implies $b \sum_\mu \mathbf q^X_\mu \in S^{(1)}$ and thus $b=0$. So we have $S \cong \mathcal{W}_2 \oplus S^{(2)}$ with $\mathcal W_2 \cong \Z_n$. Finally it is clear that $\mathcal W_3 = 0$, because there are no excitations of weight higher than three. Therefore we obtain the weighted QSS
\begin{equation}
\mathcal W = (\Z_n^3, \Z_n, 0).
\end{equation}

It is well known that the X-cube model admits a foliated RG where layers of toric codes are exfoliated. We recover this using the results of Appendix~\ref{appendix:foliated-RG} as summarized at the end of Sec.~\ref{section:RG}. Identifying $R_\mu = \Z_n[t^{\pm}]$, we let $\Delta = 1-t$, then $\{ \Delta \mathbf q^X_{\mu}, \bar{\Delta} \mathbf q^Z_{\mu} \}$ forms a basis for $\operatorname{ker} \omega|_{Q_\mu}$, and the statistics described above are encoded in maps $\Phi_\mu : \operatorname{ker} \omega|_{Q_\mu} \to Q_{\mu}^*$ defined by 
\begin{align}
\Delta \mathbf q^X_\mu \mapsto ( \mathbf q^Z_\mu)^* \\
\bar{\Delta} \mathbf q^Z_\mu \mapsto ( \mathbf q^X_\mu)^* \text{.}   
\end{align}
The conditions of Proposition~\ref{prop:foliated-rg} are thus satisfied. Since $\T$ acts trivially on $C$, we set $m = 2$ to obtain the well known foliated RG that exfoliates one layer of $\Z_n$ toric code per unit cell of $\T^{(2)} \subset \T$ translation symmetry, for each orientation $\mu = x,y,z$.

\subsection{Toric code layers}

\begin{figure}
    \centering
    \includegraphics[width = 0.5\textwidth]{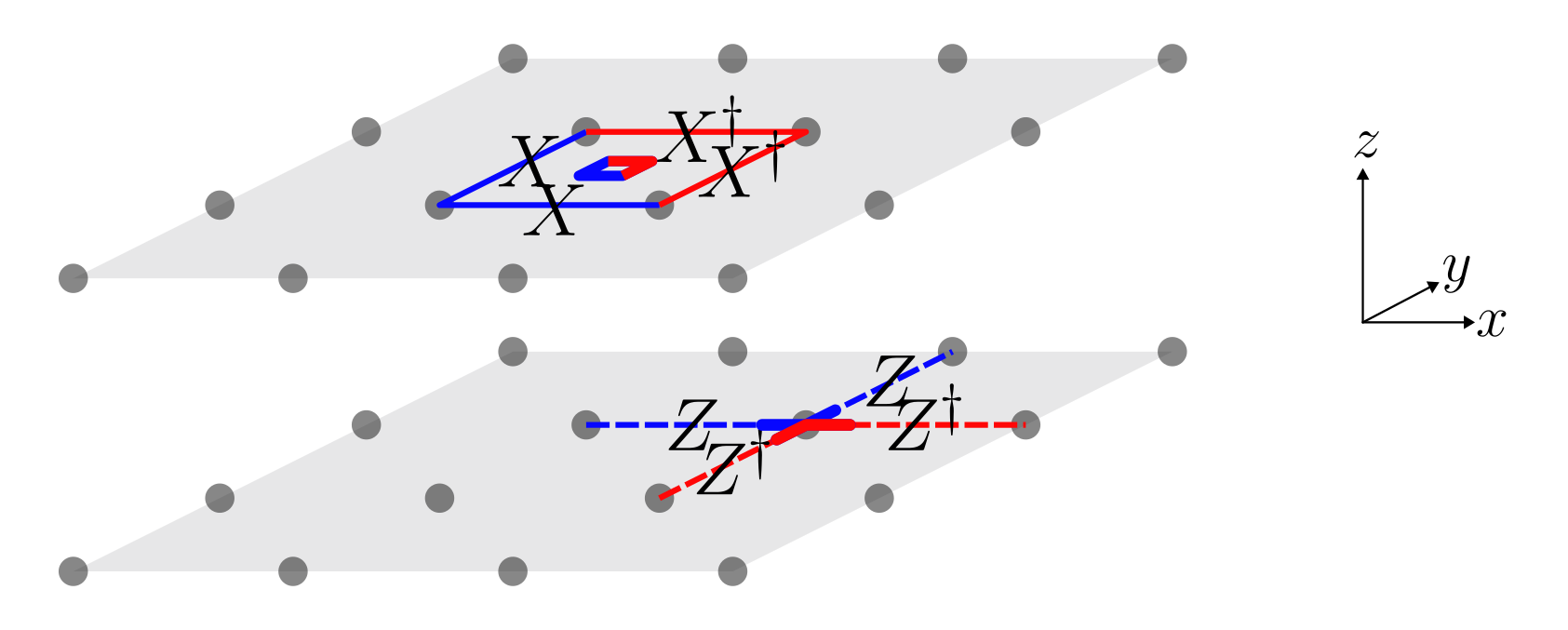}
    \caption{Plaquette and star stabilizers of $\Z_n$ Toric code layers.}
    \label{fig:TCL-stabilizers}
\end{figure}

\begin{figure}
    \centering
    \includegraphics[width = 0.8\textwidth]{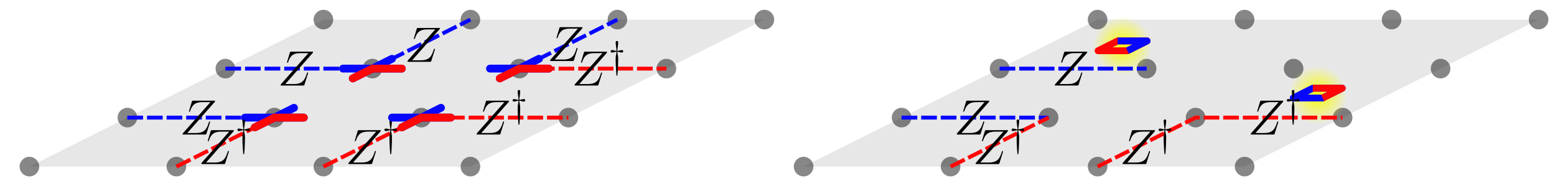}
    \caption{Coplanar product of toric code star stabilizers is closed string operator for plaquette excitations, yielding
    the standard mutual statistics of $e$ and $m$ toric code anyons. 
    }
    \label{fig:TCL-loop}
\end{figure}

We now briefly discuss the simplest example, namely layers of $\Z_n$ toric code in a three dimensional simple cubic lattice. We partition the lattice into planes labeled by $z$-coordinate and put the usual plaquette and star stabilizers of the $\Z_n$ toric code on each plane, as shown in Fig. \ref{fig:TCL-stabilizers}. We have $P = R^4$ and $E = G = R^2$. The excitation map $\epsilon : P \to E$ is given by
\begin{equation}
\epsilon = \left(
    \begin{array}{cc|cc}
        & & - d_x & d_y
        \\ \hline
        - \bar d_y & - \bar d_x & &
    \end{array}\right) \text{,}
\end{equation}
where $d_\mu = 1-\mu$.
As usual, we refer to elementary star (plaquette) excitations as $e$ ($m$) anyons.  Transporting an $e$ anyon around a large loop within its layer of mobility counts the number of $m$ anyons (modulo $n$) in the same layer, and vice versa. The result of such braiding around an arbitrary excitation is given by the value of the map $\pi : E \to Q = R_{zz}$, where
\begin{equation}
\pi = \left( \begin{array}{cc}
p_z & \\
& p_z
\end{array}\right) \text{.}
\end{equation}
It is easy to see this map is surjective, so the constraint module $C = 0$. The resulting $p$-modular sequence is exact. Therefore we can identify $S \cong Q = R_{zz}$. Recalling that elements of $R_z$ are Laurent polynomials in $z$ with $\Z_n$ coefficients, we see that elements of $S^Z = R_z$ ($S^X = R_z$) simply record the total number of $e$ ($m$) anyons on each layer, where the coefficient of $z^l$ gives the number of anyons in the corresponding layer modulo $n$. Clearly $S = S^{(1)}$ since all non-trivial excitations are planons, so $\mathcal W_1 = 0$. Note that we could have anticipated $\mathcal W_1 = C = 0$ from the fact that the $p$-theory of any stack of 2d layers is essentially 1-planar.

Finally, there is an obvious foliated RG which simply corresponds to partitioning the system into even and odd layers. This is also the RG predicted by Proposition~\ref{prop:foliated-rg} if we denote the standard basis vectors of $Q = R_{zz}$ by $\mathbf q_{\fe}, \mathbf q_{\fm}$, set $\Delta = 1$, and choose $m = 2$.

\subsection{Anisotropic model}

\begin{figure}
    \centering
    \includegraphics[width=0.4\textwidth]{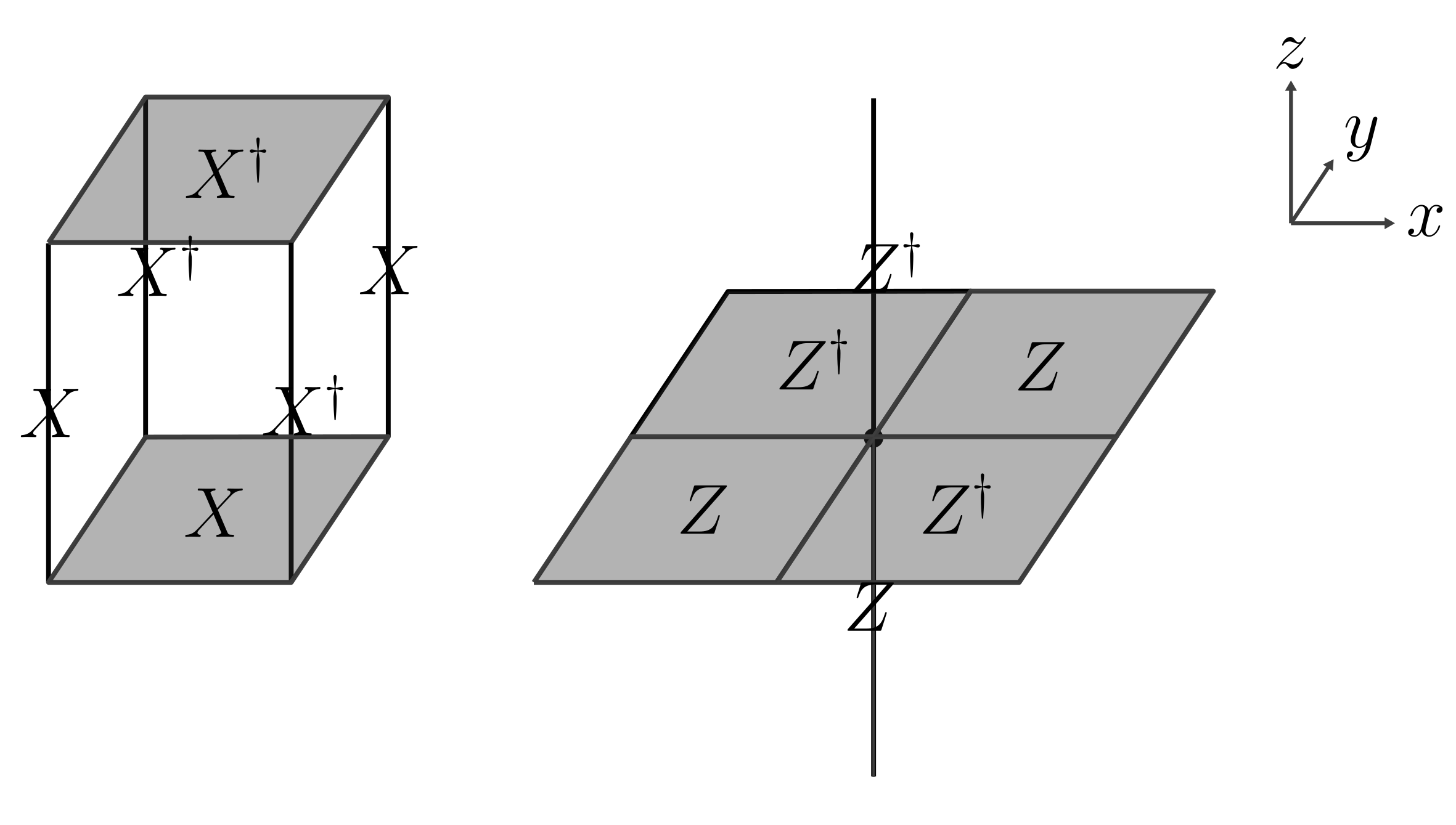}
    \caption{$\Z_n$ Anisotropic model stabilizers. Qudits live on vertical edges and horizontal faces of simple cubic lattice.} 
    \label{fig:aniso-stab}
\end{figure}

Here we record results for the $\Z_n$ Anisotropic model \cite{Shirley_2018}. This is a 2-planar model with lineon and planon excitations, but no fractons. The stabilizers are shown in Fig. \ref{fig:aniso-stab}. We set $\T_x = \langle y, z\rangle$ and $\T_y = \langle x, z \rangle$.
The $p$-modular sequence is
\begin{align}
    0 \to 
    R^2
    \xrightarrow{
        \lambda_2 \epsilon^\dagger
        % \begin{pmatrix}
        %     \bar d_z &
        %     \\
        %     - \bar d_{xy}
        %     \\
        %     &  d_{xy}
        %     \\
        %     & d_z
        % \end{pmatrix}
    }
    R^4 
    \xrightarrow{
			\begin{pmatrix}
				&&  - d_z & d_{xy}
				\\
				- \bar d_{xy} & - \bar d_z
			\end{pmatrix}
    }
    R^2
    \xrightarrow{
    \begin{pmatrix}
        p_x
        \\
        -p_y
        \\ & p_x \\ & p_y
    \end{pmatrix}
    }
    R_{xyxy}
    \xrightarrow{
    \begin{pmatrix}
        c_x & c_y 
        \\
         && c_x & - c_y
    \end{pmatrix}
    }
    R_\T^2  
    \to 0 \text{,}
    \label{eq:2xc}
\end{align}
where $d_\mu = 1-\mu$.
Note that $Q = R_{xyxy}$, and $C \cong \Z_n^2$. We verified exactness of the $p$-modular sequence along the lines described for the X-cube model. For instance,  exactness at $E = R^2$ is due to the equality of ideals $(d_z, d_{xy})= (d_y, d_z) \cap (d_z, d_x) = \ker p_x \cap \ker p_y$. 

Denoting the standard basis vectors of $Q$ by $\mathbf q^Z_x, \mathbf q^Z_y, \mathbf q^X_x,\mathbf q^X_y$, the module $S^{(1)}$ is generated by planons $\mathbf p^{Z,X}_\mu = (1-\mu)\mathbf q^{Z,X}_\mu$, where $\mu = x,y$. Given $\mathbf q \in S$, we have $\mathbf q = \mathbf q_0 + \mathbf p$, where $\mathbf p \in S^{(1)}$ and $\mathbf q_0 = a( \mathbf q^Z_x - \mathbf q^Z_y) + b (\mathbf q^X_x + \mathbf q^X_y )$, with $a,b \in \Z_n$, and where $\mathbf q_0$ and $\mathbf p$ are unique given $\mathbf q$. Therefore we have a direct sum decomposition $S \cong \mathcal W_1 \oplus S^{(1)}$ with $\mathcal W_1 \cong \Z_n^2$, and the weighted QSS
\begin{equation}
\mathcal W_\mathrm{An} = (\Z_n^2, 0) \text{,}
\end{equation}
where $\mathcal W_2 = 0$ since there are no excitations of weight greater than two. The two $\Z_n$ factors reflect the well-known fact that there are both $Z$ and $X$ lineons mobile along the $z$-axis.

If we define $\mathbf q_{\fe \mu} = \mathbf q^Z_{\mu}$ and $\mathbf q_{\fm \mu} = - \mathbf q^X_{\mu}$, identify $R_\mu = \Z_n[t^{\pm}]$, and define $\Delta = 1-t$, then the statistics is encoded by maps $\Phi_\mu$ defined by $\Phi_\mu(\Delta \mathbf q_{\fe \mu} ) = \mathbf q^*_{\fm \mu}$ and $\Phi_\mu(\bar{\Delta} \mathbf q_{\fm \mu} ) = \mathbf q^*_{\fe \mu}$. Since $\T$ acts trivially on $C$, we set $m=2$ and obtain a foliated RG where we coarsen the translation symmetry to $\T^{(2)} \subset \T$ and exfoliate one layer each of $x$- and $y$-oriented $\Z_n$ toric code per $\T^{(2)}$ unit cell.

\subsection {Checkerboard}
\label{section:checkerboard}

\begin{figure}
    \centering
    \includegraphics[width=0.6\textwidth]{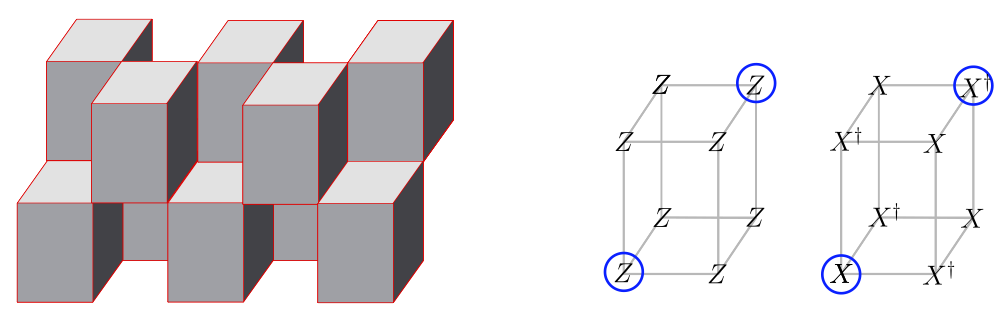}
    \caption{$\Z_n$ Checkerboard stabilizers. 
    Qudits live on vertices of simple cubic lattice. Stabilizers live on the fcc sublattice (shaded cubes on the left). There are two qudits per fcc unit cell; we take the basis sites to be those circled on the right.
    }
    \label{fig:check-stab}
\end{figure}

Next, we consider the $\Z_n$ Checkerboard model \cite{Vijay_2016}, whose stabilizers we depict in Figure \ref{fig:check-stab}. It is known that the $n=2$ model is unitarily equivalent to a stack of two $\Z_2$ X-cube models \cite{Shirley_2019_checkerboard}. However, a field theory analysis predicted that this statement does not hold for $n > 2$ \cite{Gorantla_2021}. Our analysis verifies this prediction, and provides additional information on the relationship between the $n > 2$ Checkerboard model and other fracton models.

Working over the fcc translational symmetry group $\T = \langle xy, yz, zx \rangle$, the model is described by the $p$-modular code sequence
\begin{widetext}
    \begin{align}
    0 \to R^2
    \xrightarrow{
    % \begin{pmatrix}
    %     \tau \\ \bar \tau 
    %     \\
    %     & - \tau \\
    %     & \bar \tau
    % \end{pmatrix}
    \lambda \epsilon^\dagger
    }
    R^4
    \xrightarrow{
    \begin{pmatrix}
        & & \bar \tau & \tau
        \\
         - \bar \tau &  \tau && 
    \end{pmatrix}
    }
    R^2
    \xrightarrow{
    \begin{pmatrix}
        \epsilon_x p_x
        \\
        \epsilon_y p_y
        \\
        \epsilon_z p_z
        \\
        & 
        \epsilon_x p_x
        \\ & 
        \epsilon_y p_y
        \\ & 
        \epsilon_z p_z
    \end{pmatrix}
    }
    R_{xyzxyz}
    \xrightarrow{
    \omega
    % \begin{pmatrix}
    % \kappa_x & \kappa_y & \kappa_z
    % \\
    % &&& \kappa_x & \kappa_y & \kappa_z
    % \end{pmatrix}
    }
    \Z_n^6
    \to 0,
\end{align} 
\end{widetext}
where $\tau = 1 + xy + yz + zx$, 
$\epsilon_\mu : \T \to \Z_n^\times$ is the homomorphism
given by $\nu\rho, \nu\rho^{-1} \mapsto -1$, $\mu\nu \mapsto 1$,
and $R_\mu = \Z_n \T / \T_\mu$ where $\T_\mu = \langle \nu \rho, \nu \bar{\rho} \rangle$. We have $\ker \epsilon_\mu p_\mu = (1 + \nu\rho, 1 + \nu\bar \rho)$. Exactness of the sequence at $E = R^2$ is the same as the equality of ideals $(\tau, \bar \tau) = \bigcap_\mu (1 + \nu\rho, 1 + \nu\rho^{-1})$.

To express the constraint map, it is convenient to coarsen the translation symmetry to the subgroup $\T' = \langle x^2, y^2, z^2\rangle$. This enlarges the unit cell by a factor of four, which is reflected in the fact that $R^c \cong (R')^4$ where $R' = \Z_n \T'$, and where we denote the coarsening operation by a $c$ superscript; that is, $R^c$ is $R$ viewed as an $R'$-module, and $R^c \cong (R')^4$ as $R'$-modules. We identify $R^c$ with $(R')^4$ by choosing the ordered basis $(1, yz, xz, xy)$ of elements of $R^c$. We thus have $G^c = E^c = (R')^8$ and $P^c = (R')^{16}$. The maps can also be coarsened to obtain a $p$-modular sequence. For instance, the identity $1 : R \to R$ becomes the $4 \times 4$ identity matrix $(R')^4 \to (R')^4$. A more interesting result is
\begin{equation}
(yz)^c = \left( \begin{array}{cccc}
& y^2 z^2 & & \\
1 & & & \\
& & & y^2 \\
& & z^2 & \end{array}\right) \text{.}
\end{equation}
Using these and similar expressions for $(xz)^c$ and $(xy)^c$, each entry of the $2 \times 4$ matrix representing $\epsilon : R^4 \to R^2$ is expanded to a $4 \times 4$ block. As an illustration, we explicitly write down $(\epsilon^X)^c : (P^Z)^c =(R')^8 \to (E^X)^c = (R')^4$:
\begin{equation}
(\epsilon^X)^c = \left( \begin{array}{cccccccc} 
-1 & -1 & -1 & -1 & 1 & y^2 z^2 & x^2 z^2 & x^2 y^2 \\
- \bar{y}^2 \bar{z}^2 & -1 & - \bar{y}^2 & -\bar{z}^2 & 1 & 1 & x^2 & x^2  \\
- \bar{x}^2 \bar{z}^2 & - \bar{x}^2 & - 1 & -\bar{z}^2 & 1 & y^2 & 1 & y^2  \\ 
- \bar{x}^2 \bar{y}^2 & - \bar{x}^2 & - \bar{y}^2 & -1 & 1 & z^2 & z^2 & 1  
\end{array}\right) \text{.}
\end{equation}
The map $(\epsilon^Z)^c$ takes the same form, but without the minus signs in the left-hand $4 \times 4$ block.

To coarsen the plane charge module $Q = R_{xyzxyz}$, we observe that $R^c_\mu \cong R'_{\mu \mu}$ (as $R'$-modules), where $R'_\mu = \Z_n \T' / \T'_\mu$ and $\T'_\mu = \langle \nu^2 , \rho^2\rangle$. We identify $R^c_\mu$ with $R'_{\mu \mu}$ by choosing the ordered basis $( p_\mu(1), p_\mu (\mu \nu) )$ of elements of $R^c_\mu$, so $(Q^{Z,X})^c = R'_{xxyyzz}$. The coarsened maps $(\epsilon_\mu p_\mu)^c : R^c \to R^c_\mu$ can thus be represented as $2 \times 4$ matrices, and we have
\begin{equation}
(\pi^Z)^c = (\pi^X)^c = \left( \begin{array}{cccc} 
p'_x & -p'_x & & \\
& & -p'_x & p'_x \\
p'_y &  & -p'_y & \\
& p'_y & & -p'_y \\
p'_z & & & -p'_z \\
& -p'_z & p'_z & 
\end{array}\right) \text{.}
\end{equation}

The constraint module $C^c = (C^Z)^c \oplus (C^X)^c$ and $(C^{Z,X})^c = (R'_{\T'})^3$, with the constraint map $\omega^c$ given by
\begin{equation}
(\omega^Z)^c = (\omega^X)^c = \begin{pmatrix}
        &  c'_x  &  -  c'_y & &  c'_z
        \\
        c'_x  & & & c'_y &   - c'_z
        \\
        - c'_x & & c'_y & & & c'_z 
    \end{pmatrix} \text{,}
\end{equation}
where $c'_\mu : R'_\mu \to R'_{\T'}$ is the natural projection. It is straightforward to check that $\omega^c \pi^c = 0$. We checked exactness of the coarsened $p$-modular sequence along the same lines described for the X-cube model.

We now work out the weighted QSS. Because $(\omega^Z)^c = (\omega^X)^c$, we have $\mathcal W^Z_i \cong \mathcal W^X_i$, and we focus on $Z$-excitations. We work with the coarsened $p$-modular sequence but omit ``$c$'' superscripts to avoid notational clutter. We denote the standard basis vectors for $Q^Z = R'_{xxyyzz}$ by $\mathbf q^Z_{\mu i}$, with $i=1,2$, in the order $\mathbf q^Z_{x1}, \mathbf q^Z_{x2}, \mathbf q^Z_{y 1}, \dots$. The module $S^{Z,(1)}$ is generated by planons $\mathbf p^Z_{\mu i} = (1- \mu^2) \mathbf q^Z_{\mu i}$. A general element $\mathbf q \in S$ can be uniquely expressed as $\mathbf q = \mathbf q_0 + \mathbf p$, where $\mathbf p \in S^{Z,(1)}$ and $\mathbf q_0 = \sum_\mu a_\mu \mathbf f_\mu$, where $a_\mu \in \Z_n$ and
\begin{equation}
\mathbf f_\mu = \mathbf q^Z_{\mu 1} - \mathbf q^Z_{\nu 2} + \mathbf q^Z_{\rho 2} \text{.}
\end{equation}
This gives a direct sum decomposition $S^Z \cong \mathcal W^Z_1 \oplus S^{Z,(1)}$, with $\mathcal W^Z_1 \cong \Z_n^3$. Moreover we have an isomorphism $\mathcal W^Z_1 \to \Z_n^3$ given by
\begin{equation}
\mathbf q_0 = \sum_\mu a_\mu \mathbf f_\mu \mapsto (a_x,a_y,a_z) \text{,}
\end{equation}
and we make this identification below.

By examining restrictions such as $\omega^Z|_{Q^z_x \oplus Q^z_y}$, we find that $S^{Z,(2)} \cong \tilde{S}^{(2)} \oplus S^{Z,(1)}$, where $\tilde{S}^{(2)}$ is the subgroup of $\Z^3_n$ generated by $\{ (0,1,1), (1,0,1), (1,1,0) \}$.
We have $\mathcal W^Z_2 \cong \Z_n^3 / \tilde{S}^{(2)}$, and we need to evaluate the quotient. An arbitrary element $(a_x, a_y, a_z) + \tilde{S}^{(2)} \in \Z_n^3 / \tilde{S}^{(2)}$ can be written in the form $(a,0,0) + \tilde{S}^{(2)}$ for some $a \in \Z_n$, because $(1,-1,0), (1,0,-1) \in \tilde{S}^{(2)}$ (these elements are differences of generators). Therefore $W^Z_2$ is cyclic and generated by $(1,0,0) + \tilde{S}^{(2)}$. By expressing $(1,0,0)$ as a linear combination of the generators of $\tilde{S}^{(2)}$ and solving for the coefficients, it is straightforward to show that $(1,0,0) \in \tilde{S}^{(2)}$ if and only if $n$ is odd. On the other hand, $(2,0,0) \in \tilde{S}^{(2)}$ for all $n$. Therefore we have
\begin{equation}
\mathcal W^Z_2 \cong \left\{ \begin{array}{ll}
0 , & n \text{ odd} \\
\Z_2 , & n \text{ even} 
\end{array} \right.
\end{equation}
Because there are no excitations above weight three, this gives the weighted QSS of the checkerboard model:
\begin{equation}
\mathcal W_\mathrm{Ch} = \left\{ \begin{array}{ll}
(\Z_n^6, 0 ) , & n \text{ odd} \\
(\Z_n^6, \Z_2^2, 0) , & n \text{ even}
\end{array}\right.
\end{equation}

We observe that
\begin{align}
    \mathcal W_\text{Ch}
    \cong
    \begin{cases}
        \mathcal W_\text{XC}^2 ,
        % \times
        % \mathcal W_\text{An}(\Z_{n/2})^3
        % & {\text{$n \equiv 2$ (mod $4$)}}
        & n = 2
        \\
        \mathcal W_\text{An}^3 , &  \text{$n$ odd}
    \end{cases}.
\end{align}
For $n=2$, this is consistent with the well known finite depth unitary equivalence between the $\Z_2$ checkerboard model and a stack of two $\Z_2$ X-cube models \cite{Shirley_2019_checkerboard}. For $n$ odd, this result suggests that the checkerboard model is a stack of three anisotropic models. Indeed, we establish this in Appendix~\ref{appendix:checkerboard}, showing first in terms of $p$-theories that
$P_\text{Ch} \cong P_\text{An}^x \ominus P_\text{An}^y \ominus P_\text{An}^z$, where $P^\mu_\text{An}$ is the $p$-theory of the anisotropic model oriented with lineons mobile along the $\mu$-direction. Then we go on to exhibit a finite depth circuit relating the $n$-odd checkerboard model to a stack of three anisotropic models.

When $n = 2 \operatorname{mod} 4$, since $\Z_n \cong \Z_2 \oplus \Z_{n/2}$, the weighted QSS suggest that the checkerboard model is a stack of a $\Z_2$ and a $\Z_{n/2}$ checkerboard model. Indeed, this follows from the Chinese remainder theorem for Pauli codes (Proposition~29 of Ref.~\onlinecite{Ruba_2022}). On the other hand, checkerboard models with $n = 2^k$ for $k > 1$ do not appear to be stacks of other known fracton models and are worthy of further study in future work.

Turning to foliated RG, the statistics are captured in maps $\Phi_\mu$ defined by
\begin{align}
\Delta \q^Z_{\mu 1} \mapsto - (\q^X_{\mu 2})^* \\
\Delta \q^Z_{\mu 2} \mapsto - t (\q^X_{\mu 1})^* \\
\bar{\Delta} \q^X_{\mu 1} \mapsto - \bar{t} (\q^Z_{\mu 2})^* \\
\bar{\Delta} \q^X_{\mu 2} \mapsto - (\q^Z_{\mu 1})^* \text{,}
\end{align}
where we identify $R'_\mu = \Z_n[t^{\pm}]$ and $\Delta = 1-t$. If we define new basis elements by $\q_{\fe \mu i} = \q^Z_{\mu i}$ ($i=1,2$), $\q_{\fm \mu 1} = - \q^X_{\mu 2}$ and $\q_{\fm \mu 2} = 
- t \q^X_{\mu 1}$, then the conditions of Proposition~\ref{prop:foliated-rg} are satisfied. We obtain a foliated RG where we coarsen the $\T'$ symmetry to $\T^{'(2)}$ and exfoliate two layers of $\Z_n$ toric code in each orientation per $\T^{'(2)}$ unit cell.

\subsection {4-planar X-cube}

\begin{figure}
    \centering
    \includegraphics[width=0.8\textwidth]{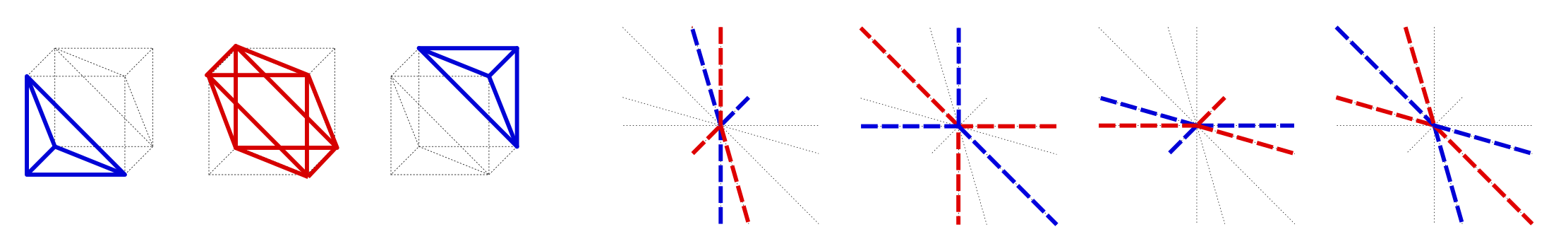}
    \caption{4-planar X-cube stabilizers. Qudits live on edges.  Solid (dashed) blue lines denote Pauli $X$ ($Z$). Red lines denote Hermitian conjugates.
    Each of the star stabilizers involves Pauli $Z$ operators that are coplanar within a $(100)$, $(010)$, $(001)$ or $(111)$ lattice plane. Only three of the four star stabilizers are independent. 
    }
    \label{fig:4xc-stab}
\end{figure}

Now we consider the 4-planar X-cube model \cite{Slagle_2018,Shirley_2019,Shirley_2023}, as presented in Appendix A of Ref.~\onlinecite{Shirley_2023}; the stabilizers are depicted in Fig.~\ref{fig:4xc-stab}. This model was referred to in Ref.~\onlinecite{Shirley_2023} as the 4-foliated X-cube model, but we prefer the term 4-planar, because this property is easier to check. While this model differs microscopically from a model of Ref.~\onlinecite{Shirley_2019} on the pyrochlore lattice (referred to in Ref.~\onlinecite{Shirley_2019} as the ``hyperkagome'' X-cube model), the models have identical modules of superselection sectors and plane charge structure \cite{Evan_unpublished}.  More generally, we expect, but have not shown, that any 4-planar X-cube model constructed following Ref.~\onlinecite{Slagle_2018} is in the same phase as the specific model considered here, as long as the four orientations of planes have the property that any triple of orientations intersects at a point. Here we state the main results without detailed derivations.

Before proceeding, we observe that the analysis here implies that the Chamon model is also $p$-modular, using the results of Ref.~\onlinecite{Shirley_2023}. After coarsening, the Chamon model is equivalent to a fermionic gauge theory with identical fusion to the $\Z_2$ 4-planar X-cube, with lineon flux excitations and fracton charge excitations. The planon remote detection properties between charge and flux sectors are identical to those of the 4-planar X-cube model, which is enough to imply $p$-modularity.

The translation group is $\T = \langle x,y,z \rangle$. We find the $p$-modular sequence:
\begin{align}
    0
    &\to 
    R^6
    \xrightarrow{
        \lambda \epsilon^\dagger
    % \begin{pmatrix}
    % & \bar d_x & - \bar d_x &
    % \\
    % - \bar d_y &  & \bar d_y
    % \\
    % \bar d_z & - \bar d_z & 
    % \\
    % \bar y - \bar z & &
    % \\
    % & \bar z - \bar x &
    % \\
    % && \bar x - \bar y
    % \\
    % &&& 1 & -(y + z) &      yz 
    % \\
    % &&& 1 & -(z + x) &      zx
    % \\ 
    % &&& 1 & -(x + y) &      xy
    %     \\
    % &&& 1 & -(1 + x)  &x 
    %     \\
    % &&& 1 & -(1 + y) & y
    %     \\ 
    % &&&    1 & -(1 + z) & z
    % \end{pmatrix}
    }
    R^{12}\dots
    \nonumber
    \\
    &\xrightarrow{
    % \sigma^\dagger \lambda
    % \begin{pmatrix}
    %     &&&&&
    %     1 & 1  & 1  &  1  & 1 & 0
    %     \\
    %     &&&&&
    %     - (y + z) & - (z + x)  & - (x + y) &  - (1 + x ) & - (1 + y) & - (1 + z) 
    %     \\
    %     &&&&&
    %     yz & zx & xy & x & 
    %     \\
    %     & - d_y & d_z &  y - z 
    %     \\
    %     d_x & & - d_z & &  z - x
    %     \\
    %     - d_x & d_y &&& x - y
    % \end{pmatrix}
    \left(
    \begin{tabular}{L L L L L L L L L L L L}
    &&&&&&  0 & d_y & - d_z & y- z & 0 & 0
    \\
    &&&&&&  -  d_x  & 0 & d_z & 0 & z - x& 0
    \\
    &&&&&& d_x &  - d_y & 0 & 0 & 0 & x - y
    \\
    -1 &  -1 & -1 &  -1 &  -1 & -1 & &&&&&
    \\
    (\bar y + \bar z)  &  (\bar z + \bar x) & (\bar x + \bar y) & (1 + \bar x) &  (1 + \bar y) & (1 + \bar z) &&&&&&
    \\
    - \overline {yz} &  - \overline {zx}  & - \overline {xy} & - \bar x & - \bar y & - \bar z &&&&&&
    \end{tabular}
    \right)
    }
    R^6
    \dots
    \nonumber
    \\
    &
    \xrightarrow{
        \begin{pmatrix}
            p_x & & &  
            \\
            &  p_y &&  
            \\
            && p_z & 
            \\
            - p_w & - p_w & - p_w & 
            \\
            &&& p_x & p_x & p_x
            \\
            &&&  p_y & p_y & p_y
            \\
            &&&  p_z & p_z & p_z
            \\
            &&&  \bar z p_w & p_w & z p_w
        \end{pmatrix}
    }
    R_{xyzwxyzw}
    \xrightarrow{
        \begin{pmatrix}
            c_x & c_y & c_z & c_w     
            \\
            &&&& c_x  &&& - c_w
            \\
            &&&&& c_y && - c_w
            \\
            &&&&&& c_z & - c_w
        \end{pmatrix}
    }
    R_\T^4
    \to 0 \text{,}
    \label{eq:4xc}
\end{align}
where $d_\mu = 1-\mu$. We define $\T_w \equiv \langle xy^{-1}, yz^{-1} \rangle$, the group of translations within $(111)$ lattice planes, and $R_w = \Z_n \T / \T_w$. As in other examples we define $\T_\mu = \langle \nu, \rho \rangle$ and $R_\mu = \Z_n \T / \T_\mu$. 
The model is 4-planar with planons moving within $(100)$, $(010)$, $(001)$ and $(111)$ lattice planes. Exactness of this sequence can be shown similarly to the standard X-cube case above. We find the weighted QSS sequence
\begin{equation}
\mathcal W_{4\mathrm{XC}} = (\Z_n^4, \Z_n, \Z_n, 0) \text{.}
\end{equation}
This tell us that the model supports intrinsic $w=2$ lineons and $w=4$ fractons, but there are no intrinsic $w=3$ excitations. We note that while the excitation map $\epsilon$ is much more complicated than for the 3-planar X-cube case, the maps $\pi$ and $\omega$ take very similar and simple forms. We speculate that this pattern continues for X-cube models of higher planarity constructed in Ref.~\onlinecite{Slagle_2018}.

The planon mutual statistics are encoded in maps $\Phi_a$, with $a = x,y,z,w$, defined by
\begin{eqnarray}
\Phi_\mu (\Delta \mathbf q^Z_\mu) &=& - (\q^X_\mu)^* \\
\Phi_\mu (\bar{\Delta} \mathbf q^X_\mu) &=& - (\q^Z_\mu)^* \\
\Phi_w (\Delta \q^Z_w) &=& - \bar{t} ( \q^X_w)^* \\
\Phi_w (\bar{\Delta} \q^X_w) &=& - t ( \q^Z_w)^* \text{.}
\end{eqnarray}
Here we have denoted the standard basis elements of $Q = R_{xyzwxyzw}$ by $\q^Z_x, \q^Z_y, \q^Z_z, \q^Z_w, \q^X_x, \q^X_y, \q^X_z, \q^X_w$. Defining new basis elements $\q_{\fe a} = \q^Z_a$, $\q_{\fm \mu}= - \q^X_\mu$ and $\q_{\fm w} = - \bar{t} \q^X_w$, the conditions of Proposition~\ref{prop:foliated-rg} are satisfied. Since $\T$ acts trivially on $C$, we choose $m = 2$ and obtain a foliated RG where one layer of $\Z_n$ toric code is exfoliated per unit cell of $\T^{(2)}$ symmetry, for each of the four orientations $a=x,y,z,w$.

\subsection {Four Color Cube}
\label{section:fcc}

\begin{figure}
    \centering
    \begin{subfigure}[b]{0.35\textwidth}
        \centering
        \includegraphics[width=0.5\textwidth]{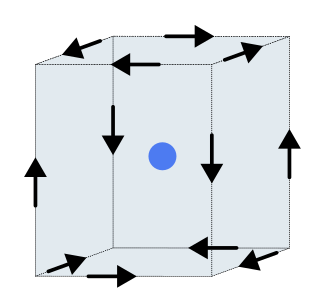}
        \caption{
            $\sigma( \mathbf A) = \overline{yz} \sigma_y \mathbf z_x - \overline{xy} \sigma_x \mathbf z_y$.
        }
    \end{subfigure}
    % \hspace{5 pt}
    % \\
    % \centering
    % \vspace{15pt}
    \begin{subfigure}[b]{0.35\textwidth}
        \centering
        \includegraphics[width=0.5\textwidth]{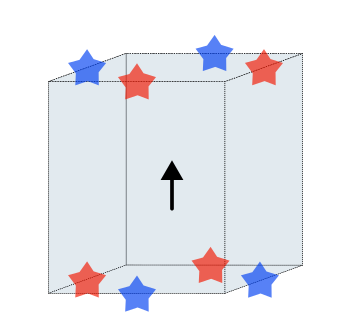}
        \caption{$\epsilon( \mathbf z_\mu) = \overline{\mu \nu} \sigma_\mu  \mathbf b$.
        }
    \end{subfigure}
    \\
    \caption{
        FCC stabilizers and locally generated excitations. The $\pm\mu$-oriented arrows represent translates of $\pm \mathbf z_\mu$ or $\pm \mathbf x_\mu$, where for $z$-oriented arrows we define $\mathbf z_z = -(\mathbf z_x + \mathbf z_y)$ and $\mathbf x_z = -(\mathbf x_x + \mathbf x_y)$, which makes the $x\to y \to z \to x$ cyclic permutation symmetry manifest. Blue (red) stars represent $+1$ ($-1$) excitations of the corresponding stabilizer.
    }
    \label{fig:fcc-stabilizer}
\end{figure}

\begin{figure*}
    \centering
    \includegraphics[width=0.8 \textwidth]{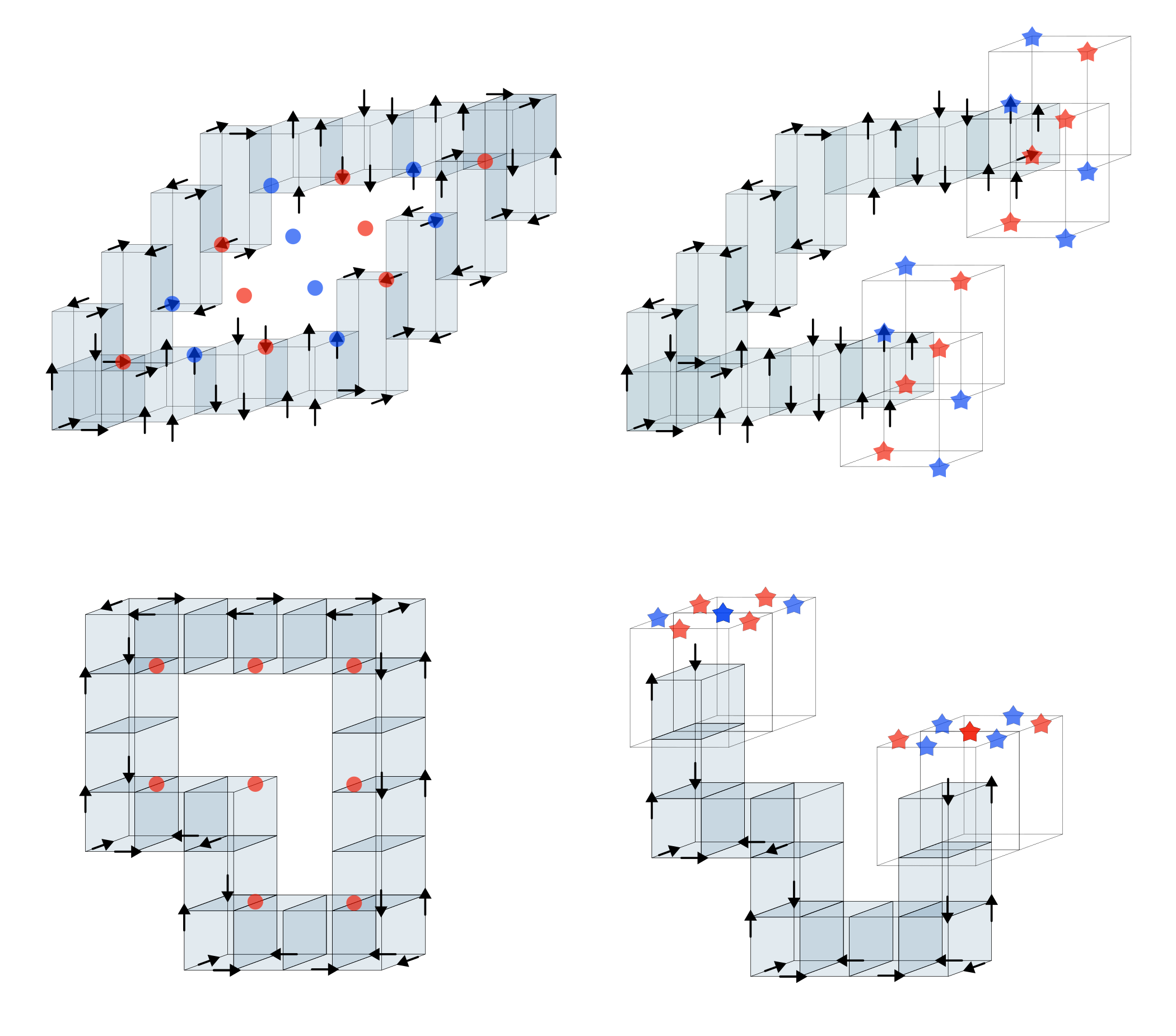}
    \caption{
        Closed (left) and open (right) string operators for $\{111\}$ (top) and $\{100\}$ planons in the FCC model. Arrows represent $\mathcal Z_\mu$ Pauli operators, with the arrow pointing in the direction of $\mu=x,y,z$. Each closed string is a product of stabilizers (inverse stabilizers) represented by blue (red) filled circles. On the right, excitations created by the open string are shown by blue (red) stars denoting +1 ($-1$) stabilizer eigenvalues.
    }
    \label{fig:planon-string-operators}
\end{figure*}

Finally, we analyze the four color cube (FCC) model. This model was introduced for $n = 2$ in Ref.~\onlinecite{Ma_2017} by coupling together a stack of four X-cube models, in a manner that was understood in terms of $p$-membrane condensation. The result was a model with two qubits and two stabilizers per site of the face centered cubic (fcc) lattice. Moreover the model is self dual, meaning it is invariant under a unitary transformation that exchanges $X$ and $Z$ Paulis. In Ref.~\onlinecite{Ma_2017}, it was shown that the model has lineon and fracton excitations, it was argued that the model has topological order (\emph{i.e.} the degenerate ground states with periodic boundary conditions are locally indistinguishable), and the ground state degeneracy was computed numerically. However, the mobility of excitations was not understood in detail; in particular, no planon excitations were identified, although it was found that the stabilizers obey constraints in both $\{ 100 \}$ and $\{ 111 \}$ planes. Here, we introduce a generalization of the FCC model to $n \geq 2$ and show it is $p$-modular and 7-planar, with planons mobile in the three $\{100\}$ planes and four $\{ 111 \}$ planes. We compute the phase invariants and show the $p$-theory admits a foliated RG.

The $\Z_n$ FCC model has two qudits at each site of the fcc lattice, with translation symmetry $\T = \langle xy, yz, zx\rangle$. Lattice sites can also be labeled by $g \in \T$. The Pauli operators at each site are $\mathcal X_x(g), \mathcal Z_x(g), \mathcal X_y(g), \mathcal Z_y(g)$ satisfying the nonstandard commutation relations
\begin{align}
    \mathcal Z_\mu(g) \mathcal X_{\mu'}(g')
    \mathcal Z_\mu^\dagger(g) \mathcal X_{\mu'}^\dagger(g')
    =
    \omega_n^{\delta_{gg'}(\delta_{\mu' \nu} - \delta_{\mu' \rho})},
    \label{eq:fcc-lambda}
\end{align}
where $\omega_n = e^{2 \pi i / n}$. Equation~(\ref{eq:fcc-lambda}) holds
for all values $\mu,\mu' = x,y,z$ once we define
$\mathcal X_z \equiv (\mathcal X_x \mathcal X_y)^{-1}$, $\mathcal Z_z \equiv (\mathcal Z_x \mathcal Z_y)^{-1}$. 
With this choice, both the stabilizers and commutation relations are invariant under cyclic permutations $x,y,z \mapsto y,z,x$; and therefore the excitation map $\epsilon = \sigma^\dagger \lambda$ is too.  We note that the commutation relations become the standard ones upon a relabeling $\cZ_1 = \cZ_x$, $\cX_1 = \cX_y$, $\cZ_2 = \cZ_y$, $\cX_2 = \cX_x^\dagger$, which we do not use below.

The use of Paulis with nonstandard commutation relations is motivated by the original presentation of the $\Z_2$ FCC model \cite{Ma_2017}.  There one starts with three qubits per site, with standard Paulis labeled by $\mu = x,y,z$ and subject to the on-site constraint $X_x X_y X_z = 1$.  In the constrained Hilbert space single $Z_\mu$ operators are unphysical and instead one works with products of two $Z$'s like $Z_\mu Z_\nu$.  The effective Pauli operators within the constrained subspace are $\cX_\mu = X_\mu$ and $\cZ_\mu = Z_\nu Z_\rho$, and these satisfy the $\Z_2$ version of the commutation relations above. The duality exchanging $X$ and $Z$ stabilizers is manifest in the $\cX$ and $\cZ$ variables; in particular, $\cX_\mu$ and $\cZ_\mu$ each create the same spatial pattern of excitations.

Returning to the case of general $n$, we have $P = R^4$ for the Pauli module, while the excitation and stabilizer modules are $G = E = R^2$, reflecting that there is one $Z$-stabilizer and one $X$-stabilizer for each fcc site. The basis vectors are denoted $\mathbf A$ and $\mathbf B$ for $G$, $\mathbf a$ and $\mathbf b$ for $E$, and $\mathbf z_x, \mathbf z_y, \mathbf x_x, \mathbf x_y$ for $P$. The model is defined by the excitation map
\begin{equation}
\epsilon = \left( \begin{array}{cccc}
 &  & \overline{xy} \sigma_x & \overline{yz} \sigma_y \\
- \overline{xy} \sigma_x & - \overline{yz} \sigma_y &  & 
\end{array} \right) \text{,}
\end{equation}
where we define $d_\mu = 1 - \mu^2$, $d_\mu^{\pm} = 1 \pm \nu\rho^{\pm 1}$, and
\begin{equation}
\sigma_\mu \equiv d_\mu d^+_\mu d^-_\mu = (1-\mu^2)(1 + \nu \rho)(1 - \nu \bar{\rho}) \text{,}
\end{equation}
which satisfy the identity $\overline{\sigma_\mu} = (\overline{\mu \nu})^2 \sigma_\mu$. Other useful identities are $\sum_{\mu = x,y,z} \overline{\mu \nu} \sigma_{\mu} = \sum_{\mu = x,y,z} \overline{\mu \nu} d^-_\mu = 0$. The non-standard commutation relations are encoded in the matrix
\begin{equation}
\lambda = \left( \begin{array}{cccc}
 & & 0 & 1  \\
  & & -1 & 0 \\
  0 & 1 & & \\
   -1 & 0 & &
  \end{array} \right) \text{,}
\end{equation}
and the stabilizer map is 
\begin{equation}
\sigma = \lambda \epsilon^\dagger = \left( \begin{array}{cc}
 \overline{yz} \sigma_y & \\
 - \overline{xy} \sigma_x & \\
 & - \overline{yz} \sigma_y \\
 & \overline{xy} \sigma_x 
 \end{array}\right) \text{.}
\end{equation}
The self-duality between $Z$ and $X$ sectors is manifest in the form of both the stabilizer and excitation maps.
Stabilizers and the pattern of excitations created by a single Pauli operator are shown in Fig.~\ref{fig:fcc-stabilizer}.

As illustrated in Fig.~\ref{fig:planon-string-operators}, the model supports planons with seven distinct orientations. The planar subgroup $\mathcal T_\mu = \langle \nu^2, \rho^2 \rangle$ describes mobility within the $\mu$-oriented $\{100\}$ plane. It turns out that $\T_\mu = \T_{S_\mu}$; that is, $\T_\mu$ is the group of translations leaving the superselection sector of $\{100\}$ planons invariant. Even though $\nu \rho$ lies in the plane of mobility of a $\mu$-oriented planon, this translation changes the superselection sector, \emph{i.e.} it acts nontrivially on $S_\mu$. To fix a counterclockwise orientation for braiding statistics, we choose the $+\hat{\mu}$ normal vector for $\mu$-oriented planes. To describe mobility within $\{111\}$ planes, we define $\mathcal T_{\pm\pm'} = \langle zx^{\pm 1}, zy^{\pm'1}\rangle$. 
We find it convenient to introduce three different notations for $\{111\}$ plane orientations, namely $(++,+-,-+,--) = (0,1,2,3) = (z^+, x^+, y^+, 3)$, and use $\alpha$ for an index running over the four orientations. With this choice, $\mathcal T_{\mu^+} = \langle  \mu\nu, \mu \rho \rangle$. While $\T_3 = \T_{S_3}$, this is not the case for $\T_{\mu^+}$, and we have $\T_{S_{\mu^+}} = \langle \nu \bar{\rho},\mu^2 \nu \rho \rangle$. We choose the normal vector $(\hat{x} + \hat{y} + \hat{z})/\sqrt{3}$ for $3$-oriented planes and $(\hat{\mu} - \hat{\nu} - \hat{\rho})/\sqrt{3}$ for $\mu^+$-oriented planes.

We identify representative planons of each orientation as follows. The excitations $\tilde{\mathbf p}^Z_{\mu 3} \equiv \overline{\mu \rho} d_\nu^+ d_\nu^- d_\mu \mathbf a$ and $\tilde{\mathbf p}^X_{\mu 3} \equiv  \overline{\mu \rho} d_\nu^+ d_\nu^- d_\mu \mathbf b$ are $\mu$-oriented planons, which we verified by finding the associated closed string operators and showing that these detect excitations within a disc enclosed by the string. The significance of the ``$3$'' subscript will become apparent below. Defining $c_\mu = d_\nu - d_\rho = \rho^2 - \nu^2$, we have $\mu^+$-oriented planons $\tilde{\mathbf p}^Z_{\mu^+} \equiv - \mu^2 d_\mu^+ d_\nu c_\nu \mathbf a$ and $\tilde{\mathbf p}^X_{\mu^+} \equiv - \mu^2 d_\mu^+ d_\nu c_\nu \mathbf b$. Finally, $\tilde{\mathbf p}^Z_3 \equiv \bar{x}^2 d_x^+ d_y^+ d_x d_y \mathbf a$ and $\tilde{\mathbf p}^X_3 \equiv \bar{x}^2 d_x^+ d_y^+ d_x d_y \mathbf b$ are planons with orientation $\alpha = 3$. The superselection sectors of these representative planons turn out to generate the submodules $S_a \subset S$ of $a$-oriented planons.

Using the planon string operators shown in Fig. \ref{fig:planon-string-operators}, we find $Q^Z = Q^X = R_{zxy0123}$, where $R_\mu = \Z_n \T / \T_\mu$ and $R_\alpha = \Z_n \T / \T_\alpha$, using the planar subgroups introduced above. The plane charge map is
\begin{equation}
\pi^Z = \pi^X =  \begin{pmatrix}
         p_z
         \\
         p_x
         \\
         p_y
         \\
         \epsilon_{z^+} p_{z^+}
         \\
         \epsilon_{x^+} p_{x^+}
         \\
         \epsilon_{y^+} p_{y^+}
         \\
         p_3
     \end{pmatrix} \text{,}
\end{equation}
where $p_a : R \to R_a$ are the projections and we have group homomorphisms $\epsilon_{\mu^+} : \T \to \Z_n^\times$ defined by $\mu\nu, \rho\mu \mapsto -1$, $\nu\rho \mapsto 1$. We then complete the $p$-modular sequence by defining $C = Q / \im \pi$ and $\omega : Q \to C$ the corresponding quotient map. Exactness of the sequence was checked along the lines described in Sec.~\ref{section:x-cube}. This shows that the FCC model is $p$-modular.

For the $p$-modular sequence to be useful for calculations, we need an explicit description of $C$ and $\omega$. In the other examples, we were able to guess a form of $C$ and $\omega$ and check the resulting sequence is exact, but here we take a different approach. As a consequence of self-duality, from now on we focus on $Z$-stabilizer excitations and omit $Z$-superscripts for elements of $Q^Z$; the corresponding discussion for $X$-excitations is identical. We start by coarsening the translation symmetry to $\T' = \langle x^2, y^2, z^2 \rangle$ and obtaining a more explicit description of $(Q^Z)^c$. We note that $R^c \cong (R')^4$, with $1,yz,zx,yx$ forming a basis for $R^c$; this is nothing but the standard description of the fcc lattice as a simple cubic lattice with basis. We can associate each basis site with a $\{ 111 \}$ plane orientation by observing that if we remove the basis site in question, the other three basis sites lie in a $\{ 111 \}$ plane. In this manner we associate $(1,yz,zx,xy) \leftrightarrow (3, x^+,y^+,z^+)$, and define fracton excitations $\mathbf f_3 = \mathbf a$, $\mathbf f_{x^+} = yz \mathbf a$, $\mathbf f_{y^+} = xz \mathbf a$ and $\mathbf f_{z^+} = xy \mathbf a$, which generate $(E^Z)^c \cong (R')^4$. It is convenient to write elements of $(Q^Z)^c$ as basis elements multiplied by one-variable Laurent polynomials in $\Z_n[t^\pm]$, where multiplication by $t$ represents the smallest perpendicular translation between parallel planes. There are four basis elements for $R^c_\mu \subset (Q^Z)^c$, which are denoted $\mathbf q_{\mu \alpha}$. The single generator of each $R^c_\alpha \subset (Q^Z)^c$ is denoted $\mathbf s_\alpha$. The generators are defined as follows:
\begin{eqnarray}
\mathbf q_{\mu 3} &=& -p_\mu(\mathbf f_3) \\
\mathbf q_{\mu' \mu^+} &=& p_{\mu'}(\mathbf f_{\mu^+} ) \\
 t  \mathbf s_3 &=& p_3(\mathbf f_3) \\
\mathbf s_{\mu^+} &=& - p_{\mu^+}(\mathbf f_3) \text{.}
\end{eqnarray}
The choices of minus signs and shifts by $t$ above are made so that the generators of $S^Z \cong \pi^Z(E^Z)$ take the following nice form:
\begin{eqnarray}
\pi^Z (\mathbf f_3) &=& - ( \mathbf q_3 + \mathbf s - (1+t) \mathbf s_3 ) \\
\pi^Z (\mathbf f_{\mu^+}) &=& \mathbf q_{\mu^+} + \mathbf s - (1+t) \mathbf s_{\mu^+} \text{.}
\end{eqnarray}
Here and below we adopt the convention that an omitted index is summed over, so \emph{e.g.} $\mathbf q_3 = \sum_{\mu} \mathbf q_{\mu 3}$. From now on we omit ``$c$'' superscripts to avoid notational clutter. 

The $R'$-module structure on $Q^Z$ is given by specifying by the action of 
 $\T'$ translations on the generators, which is
\begin{equation}
    \mu^{2a} \nu^{2b} \rho^{2c}
    \begin{pmatrix}
        \mathbf q_{\mu\alpha}
        \\
        \mathbf s_{\mu^+}
        \\
        \mathbf s_3
    \end{pmatrix}
    =
    \begin{pmatrix}
        t^a \mathbf q_{\mu\alpha}
        \\
        t^{b + c - a} \mathbf s_{\mu^+}
        \\
        t^{-(a + b + c)} \mathbf s_3
    \end{pmatrix}.
\end{equation}
Note that for each generator of $Q^Z$, given one of $x^2$, $y^2$ or $z^2$ that acts nontrivially on it, one has to make an arbitrary choice whether the translation acts as $t$ or $\bar{t}$. Once this choice is made, the action of the other two generators of $\T'$ is fixed.

Upon coarsening $\T \to \T'$, there are four types of $\{100\}$ planons of each orientation. As elements of $S^Z \cong \pi^Z(E^Z) \subset Q^Z$, these are generated by $\mathbf p_{\mu \alpha} = \pi^Z(\tilde{\mathbf p}^Z_{\mu \alpha}) = \Delta_1 \mathbf q_{\mu \alpha}$, where $\Delta_i \equiv t^i - 2 + \bar{t}^i$, and where we defined $\tilde{\mathbf p}^Z_{\mu' \mu^+} \equiv - \nu \rho \tilde{\mathbf p}^Z_{\mu' 3}$. Similarly, the $\{111\}$ planons are generated by $\mathbf p_\alpha = \pi^Z(\tilde{\mathbf p}^Z_\alpha) = \Delta_2 \mathbf s_\alpha$. After obtaining an expression for $\omega^Z$, it can be checked that these excitations generate all planons and composites of planons, \emph{i.e.} they generate $S^{(1)}$.

We now outline our strategy for obtaining an explicit description of $\omega^Z$ and $C^Z$ before carrying it out. As an abelian group (or a $\Z_n$-module), $Q^Z$ is generated by the infinite set $G_{Q^Z} = \{ t^{\Z} \mathbf q_{\mu \alpha} , t^\Z \mathbf s_\alpha \}$. We will choose a certain finite subset $\{ \mathbf u_1, \dots \mathbf u_m \} \subset G_{Q^Z}$, and define $\mathfrak u_i \equiv \omega^Z(\mathbf u_i)$, $i = 1,\dots,m$. Identifying certain elements of $\ker \omega^Z$ will give us relations that can be used to express $\omega^Z(\mathbf q)$ for any $\mathbf q \in G_{Q^Z}$ in terms of the $\mathfrak u_i$, which thus generate $C^Z$. This provides an explicit description of the $R'$-module structure on $C^Z$ in terms of the generators $\{ \mathfrak u_i \}$ via $r' \mathfrak u_i = \omega^Z(r' \mathbf u_i)$, where $r' \in R'$.

The next step is to show that the $\mathfrak u_i$ do not satisfy any relations and thus form a basis for $C^Z$ as a $\Z_n$-module, implying that $C^Z \cong \Z_n^m$. We let $\tilde{C}^Z$ be the free $\Z_n$-module with generators $\{ \tilde{\mathfrak u}_1, \dots, \tilde{\mathfrak u}_n \}$. There is an obvious surjective map $\alpha : \tilde{C}^Z \to C^Z$ defined by $\tilde{\mathfrak u}_i \mapsto \mathfrak u_i$.
The expression mentioned above for $\omega^Z(\mathbf q)$ in terms of the $\mathfrak u_i$ defines a group homomorphism $\tilde{\omega}^Z : Q^Z \to \tilde{C}^Z$ by replacing $\mathfrak u_i \to \tilde{\mathfrak u}_i$, which is easily seen to be surjective. Using the explicit description of the $R'$-module structure on $C^Z$, we find an $R'$-module structure on $\tilde{C}^Z$ that makes $\tilde{\omega}^Z$ and $\alpha$ into $R'$-module homomorphisms. Finally we have $\im \pi^Z \subset \ker \tilde{\omega}^Z \subset \ker \omega^Z$, where the first inclusion is easily checked by evaluating $\tilde{\omega}^Z$ on the generators $\pi^Z(\mathbf f_\alpha)$, and the second inclusion follows trivially from the fact $\omega^Z = \alpha \circ \tilde{\omega}^Z$. Since $\im \pi^Z = \ker \omega^Z$, we thus have $\ker \tilde{\omega}^Z = \ker \omega^Z$. It follows that $\alpha$ is an $R'$-module isomorphism and thus we can identify $\tilde{C}^Z$ with $C^Z$.

Following the above strategy, we begin by defining $\mathfrak q_{\mu \alpha} = \omega^Z(\mathbf q_{\mu \alpha})$. We also define
\begin{equation}
\mathfrak d_\alpha = \omega^Z\big( \varepsilon_{\mu \alpha} t^k (1-t) \mathbf q_{\mu \alpha^\mu} \big) \text{,}
\end{equation}
where $\alpha^\mu$ is the image of $\alpha$ under the permutation $(\mu^+ 3)(\nu^+ \rho^+)$, and
$\varepsilon_{\mu\alpha} = \delta_{\alpha \nu^+} + \delta_{\alpha\rho^+} - \delta_{\alpha\mu^+} - \delta_{\alpha 3}$.
The definition of $\mathfrak d_\alpha$ makes sense because the right-hand side is independent of $k \in \Z$ and the index $\mu$. To see this, we consider the $\langle 100 \rangle$ lineon excitation $\mathbf l_{\mu \nu} = (1 + \mu \nu)(1- \mu \bar{\nu}) \mathbf a$, which is mobile along the $\rho$-direction. The image of $\mathbf l_{\mu \nu}$ and its translates by arbitrary elements of $\T$ under $\pi^Z$ is 
$\pm [ \epsilon_{\mu \alpha} t^k (1-t) \mathbf q_{\mu \alpha^\mu} -
\epsilon_{\nu \alpha} t^\ell (1-t) \mathbf q_{\nu \alpha^\nu}]$ for arbitrary $\alpha = 0,1,2,3$ and $k,\ell \in \Z$, where $\alpha, k$ and $\ell$ depend on the translation applied to $\mathbf l_{\mu \nu}$. Using $\im \pi^Z = \ker \omega^Z$, it follows that $\mathfrak d_\alpha$ is well-defined. Finally we define $\mathfrak s_\alpha = \omega^Z(\mathbf s_\alpha)$.

To determine the value of $\omega^Z$ on arbitrary elements of $Q^Z$, we consider the $\{100\}$ planons $\mathbf p_{\mu \alpha}$. We have $\omega^Z(t^k \mathbf p_{\mu \alpha}) = \omega^Z(t^k \Delta_1 \mathbf q_{\mu \alpha}) = 0$, which gives a recursion relation $\omega^Z(t^{k+1} \mathbf q_{\mu \alpha}) = 2 \omega^Z(t^k \mathbf q_{\mu \alpha}) - \omega^Z(t^{k-1} \mathbf q_{\mu \alpha} )$. This can be solved to find
\begin{equation}
\omega^Z(t^k \mathbf q_{\mu \alpha} ) = \mathfrak q_{\mu \alpha} - k \varepsilon_{\mu \alpha} \mathfrak d_{\alpha^\mu}, \qquad k \in \Z \text{.} \label{eqn:omega1}
\end{equation}
We proceed similarly to determine $\omega^Z(t^k \mathbf s_\alpha)$. We have $\omega^Z(t^k \mathbf p_{\alpha}) = \omega^Z(t^k \Delta_2 \mathbf s_{\alpha}) = 0$, giving $\omega^Z(t^{k+2} \mathbf s_\alpha) = 2 \omega^Z(t^k \mathbf s_\alpha) - \omega^Z(t^{k-2} \mathbf s_\alpha)$. To solve this, we use $\omega^Z(\pi^Z( \mathbf f_\alpha)) = 0$ to express $\omega^Z(t \mathbf s_\alpha) = \mathfrak s - \mathfrak s_\alpha + \mathfrak q_\alpha$. Similarly, we use the fact that $\omega^Z \circ \pi^Z$ vanishes on certain translates of $\mathbf f_\alpha$ to write $\omega^Z(t^k \mathbf s_\alpha)$ in terms of $C^Z$ generators for $k = -1, 2$. This provides enough information to solve the recursion relation, and we find
\begin{equation}
\omega^Z(t^k \mathbf s_{\alpha} ) = \left\{ \begin{array}{ll}
\mathfrak (\mathfrak s_\alpha - k \mathfrak s) - \frac{k}{2}(\mathfrak q + \mathfrak d - \mathfrak d_\alpha), & k \in 2\Z \\
\mathfrak -(\mathfrak s_\alpha - k \mathfrak s) + \mathfrak q_\alpha + \frac{k-1}{2}(\mathfrak q + \mathfrak d_\alpha), & k \in 2\Z +1 \end{array}\right.  \text{.}  \label{eqn:omega2}
\end{equation}
The $R'$-module structure on $C^Z$ is given by
\begin{eqnarray}
\mu^2 \mathfrak q_{\mu \alpha} &=& \mathfrak q_{\mu \alpha} - \varepsilon_{\mu \alpha} \mathfrak d_{\alpha^\mu} \label{eq:mod-structure-first}\\
\nu^2 \mathfrak q_{\mu \alpha} &=& \rho^2 \mathfrak q_{\mu \alpha}  = \mathfrak q_{\mu \alpha} \\
x^2 \mathfrak d_\alpha &=& y^2 \mathfrak d_\alpha = z^2 \mathfrak d_\alpha = \mathfrak d_\alpha \\
x^2 \mathfrak s_3 &=& y^2 \mathfrak s_3 = z^2 \mathfrak s_3 = - (\mathfrak s_3 + \mathfrak s) + \mathfrak q_3 - (\mathfrak q + \mathfrak d_3)  \\
\mu^2 \mathfrak s_{\mu^+} &=& - (\mathfrak s_{\mu^+} + \mathfrak s) + \mathfrak q_{\mu^+} - (\mathfrak q + \mathfrak d_{\mu^+}) \\
\nu^2 \mathfrak s_{\mu^+} &=& \rho^2 \mathfrak s_{\mu^+} = - (\mathfrak s_{\mu^+} - \mathfrak s) + \mathfrak q_{\mu^+} \text{.} \label{eq:mod-structure-last}
\end{eqnarray}

Reinterpreting the $\mathfrak q_{\mu \alpha}, \mathfrak d_\alpha, \mathfrak s_\alpha$ as generators of the free $\Z_n$-module $\tilde{C}^Z \cong \Z_n^{20}$, Equations~\ref{eq:mod-structure-first}-\ref{eq:mod-structure-last} give an $R'$-module structure on $\tilde{C}^Z$. To verify this, it is enough to check that each of $x^2, y^2, z^2$ gives an automorphism acting on $\tilde{C}^Z$, and that these automorphisms commute. The expressions~\ref{eqn:omega1} and~\ref{eqn:omega2} define $\tilde{\omega}^Z : Q^Z \to \tilde{C}^Z$, which can be straightforwardly checked to be an $R'$-module homomorphism. The same property obviously holds for $\alpha : \tilde{C}^Z \to C^Z$. Finally, straightforward computation shows that $\tilde{\omega}^Z(\pi^Z(\mathbf f_\alpha)) = 0$; by the discussion above, we thus have $C^Z \cong \tilde{C}^Z \cong \Z_n^{20}$.  Including the X-sector where identical results hold, we have found the invariant $C \cong \Z_n^{40}$.

We now compute the weighted QSS, focusing again on the $Z$-sector. We first verify that the $\mathbf p_{\mu \alpha}$ and $\mathbf p_{\alpha}$ defined above generate $S^{(1)}$. This can be done straightforwardly by writing down the most general expression for a planon $\mathbf p$ of each orientation in $Q^Z$, and requiring $\omega^Z(\mathbf p) = 0$. We then find that the quotient $Q^Z / S^{(1),Z}$ can be identified with the free $\Z_n$-module with the 40 generators $\{ t^k \mathbf q_{\mu \alpha}, t^\ell \mathbf s_{\alpha} \}$, where $k = 0,1$ and $\ell = 0,\dots,3$. This is done by observing that an arbitrary element of $Q^Z$ can be ``cleaned up'' to a linear combination of the above generators by adding multiples of the $\mathbf p_{\mu \alpha}, \mathbf p_\alpha$ and their translates. At this point we carried out the computation numerically using Macaulay2. We have the obvious induced map $\omega_Z : Q^Z / S^{(1),Z} \to C^Z$, which can be expressed as a $20 \times 40$ matrix (we use the same symbol for $\omega^Z$ in a slight abuse of notation). We have $\mathcal W^Z_{1} = S^Z / S^{(1),Z} \cong \ker \omega^Z$. The modules $S^{(i),Z} / S^{(1),Z} \subset Q^Z / S^{(1),Z}$ can be constructed explicitly, and the quotients
\begin{equation}
\mathcal W^Z_{i} = S^Z / S^{(i),Z} = \frac{\ker \omega^Z}{S^{(i),Z}/S^{(1),Z}} 
\end{equation}
are straightforward to evaluate. 
It is convenient to work with $\Z$ coefficients, and obtain a result that gives the answer for $\Z_n$ coefficients upon tensoring with $\Z_n$. Combining with $Z$ and $X$ sectors we have the weighted QSS sequence for $\Z$ coefficients:
\begin{equation}
\mathcal W_{\Z} = (\Z^{40},\Z^{10} \times \Z^6_2, \Z^{10}, \Z_4^2 \times \Z_2^6, \Z_2^2, \Z_2^2, 0) \text{.}
\end{equation}
Tensoring with $\Z_n$ gives the weighted QSS sequence
\begin{equation}
\mathcal W = \left\{ \begin{array}{ll}
(\Z_n^{40},\Z_n^{10},\Z_n^{10},0,0,0,0), & n \text{ odd} \\
(\Z_n^{40}, \Z_{n}^{10} \times \Z_2^6 , \Z_n^{10}, \Z_2^8 , \Z_2^2, \Z_2^2, 0) ,
& n = 2 \mod 4 \\
(\Z_n^{40}, \Z_{n}^{10} \times \Z_2^6 , \Z_n^{10}, \Z_4^2 \times \Z_2^6 , \Z_2^2, \Z_2^2, 0) ,
& n = 0 \mod 4
\end{array}\right. \text{.}
\end{equation}

The planon mutual statistics are given by 
\begin{align}
    \Phi_\mu(\Delta_1 \mathbf q_{\mu\alpha}^X) = - \varepsilon_{\mu\alpha} \mathbf q_{\mu\alpha^\mu}^{Z*},
    \\
    \Phi_\mu(\Delta_1 \mathbf q_{\mu\alpha}^Z) = - \varepsilon_{\mu\alpha} \mathbf q_{\mu\alpha^\mu}^{X*},
\end{align}
and \begin{align}
    \Phi_\alpha(\Delta_2 \mathbf s_\alpha^X) = \mathbf s_\alpha^{Z^*},
    \\
    \Phi_\alpha(\Delta_2 \mathbf s_\alpha^Z) = \mathbf s_\alpha^{X^*} \text{.}
\end{align}
Now, having coarsened to $\T'$ translational symmetry, 
we note that the hypotheses of Proposition \ref{prop:foliated-rg} are satisfied after a redefinition of basis elements of $Q$. Coarsening by a further factor of $2n$ in linear dimension, the subgroup $\T'' = \langle x^{4n}, y^{4n}, z^{4n}\rangle$ acts trivially on the constraint module. We note that this is the only model among our examples where the amount by which we need to coarsen to trivialize the translation action depends on qudit dimension. Thus, letting $P_\text{FCC}$ be the $p$-theory of the FCC model with $\T'$-translation symmetry, we identify a foliated RG of the form
\begin{align}
    P_\text{FCC} \to_{2n+1} P_\text{FCC}^{(2n+1)} \ominus P_\text{TC}^{(xyz)^{(8n)}(0123)^{2n}}.
\end{align}
Here, upon coarsening by a factor of $2n+1$, we have a copy of the FCC model on a coarser spatial lattice together with $8n$ toric codes per unit cell in each $\{100\}$ orientation plus $2n$ toric codes per unit cell in each $\{111\}$ orientation. 

Due to the large size of the matrices involved after coarse-graining, even for small qudit dimension, it is challenging to find a circuit realizing such an RG. 
    Nevertheless, as a check, we note that the result above recovers the system-size-dependence of the logarithm (in base $n$) of ground
    state degeneracy $K(L/a)$ on a three-torus of linear system size $L/a$ as obtained for the $\Z_2$ FCC model in Ref.~\onlinecite{Ma_2017}.  This determination of $K(L/a)$ was based on non-rigorous arguments (as noted in Ref.~\onlinecite{Ma_2017}) supported by numerical calculations for finite system sizes. Setting $a = 1$, if we assume a form
    % \begin{align}
        $K(L) = c L + d$,
    % \end{align}
    then the RG implies 
    % \begin{align}
        $K(L)
        =
        K(\frac 1 {2n+1} L)
        + 32n \frac {2}{2n+1} L$,
    % \end{align}
    which yields $c = 32$ as expected from the $n=2$ result of Ref.~\onlinecite{Ma_2017}.

\section {Discussion}
\label{section:conclusion}

We conclude with a discussion of several open questions. One interesting question is whether every $p$-modular $p$-theory is realizable by some physical quantum system. For anyons in two spatial dimensions, the analog of a $p$-modular $p$-theory is an abelian group describing fusion of excitations together with a bilinear form encoding a modular braiding. It is well known that every such theory can be realized. However, we believe that not every $p$-modular $p$-theory can be realized. This will be treated in upcoming work of two of us (EW, MH) and Wilbur Shirley, where we also discuss additional conditions on a $p$-theory that may be sufficient for its realizability \cite{WSH_inprep}.

As noted in Sec.~\ref{section:definitions}, a $p$-theory does not incorporate all information on statistical properties of excitations needed for a full description of the universal properties of a phase. In particular, $p$-theories lack information on self exchange statistics, and there are examples of distinct fracton phases with the same $p$-theory and different self-statistics \cite{Pai_2019, Song_2024}. In future work, it will be interesting to modify the definition of $p$-theory to include self statistics in an efficient manner. With such a definition in hand, together with a better understanding of which $p$-theories are realizable, we conjecture that the problem of classifying $p$-modular fracton phases will reduce to an algebraic problem of classifying suitably enhanced $p$-theories.

However, there are issues for which including self statistics is likely not important. In CSS codes, any self statistics originates from mutual statistics between $X$ and $Z$ sector excitations. While bound states of such excitations can have non-trivial self statistics, it is not clear that this needs to be included as additional data within a $p$-theory. We speculate that classifying $p$-modular CSS codes is equivalent to a problem of classifying $p$-theories, incorporating suitable conditions for realizability. We also point out that the arguments of Sec.~\ref{section:RG} imply that the data of a $p$-theory is sufficient to understand entanglement RG flows of $p$-modular fracton orders, at the level of a system's non-trivial excitations. 

It will be important to identify more examples of $p$-modular fracton orders. Indeed, while the landscape of fracton models is large and complex, there are only a handful of $p$-modular fracton orders known to be realized in the simplest setting of $\Z_2$ CSS codes. We are aware of at least two examples from the literature not studied in this paper that are foliated and likely to be $p$-modular.  These are are the stacked kagome lattice X-cube model of Ref.~\onlinecite{Shirley_2019} and the triangular/honeycomb anisotropic Laplacian model studied in Ref.~\onlinecite{Xue_2024}.  To generate more examples, one way forward is to start with a trivial paramagnet with a planar subsystem symmetry, and gauge the symmetry to obtain a fracton model. Here, we mention one set of conjectures along these lines from work in progress of EW, MH and Shirley \cite{WSH_inprog}. Suppose we choose $k$ orientations of lattice planes labeled by $a_1, \dots, a_k$, and let $\mathcal S_{a_i}$ be the product of all $\Z_n$ planar symmetry operators of a given orientation. We suppose that the symmetry obeys global relations $\mathcal S_{a_i} = \mathcal S_{a_j}$ for any pair of orientations. For example, this describes the symmetry of the plaquette Ising model on the 3d cubic lattice \cite{Vijay_2016}.

Starting from a trivial paramagnet and gauging the symmetry, we define the resulting model to be a $k$-planar X-cube model; this is a reasonable definition because the 3- and 4-planar X-cube models, as well as the 2-planar anisotropic model, can all be obtained this way \cite{Vijay_2016,Shirley_2019_gauging,Shirley_2023}. We expect that the large class of X-cube models constructed in Ref.~\onlinecite{Slagle_2018} can also be obtained by gauging planar symmetries in this manner. We conjecture that the resulting fracton order is $k$-planar, foliated and $p$-modular, and only depends on $k$ and on the intersection properties of the chosen orientations. For example, it certainly matters whether a triple of lattice planes intersects at a point or along a line; we conjecture that this is the only geometrical data needed to specify a unique fracton order within this construction. We further conjecture that the constraint map of a $k$-planar X-cube model is of the form
\begin{align}
    \omega_\text{$k$XC}
    :
    R_{a_1\dots a_k a_1 \dots a_k}
    \xrightarrow{
        \begin{pmatrix}
            c_{a_1} & \dots & c_{a_k}
            & &&&
            \\
            &&& c_{a_1} & & & -c_{a_k}
            \\
            &&& & \ddots & & \vdots
            \\
            &&& && c_{a_{k-1}} & -c_{a_k}
        \end{pmatrix}
    }
    R_{\T}^{k} \text{,}
\end{align}
generalizing our results for the anisotropic model and 3- and 4-planar X-cube models. The resulting weighted QSS are easily computed to be $\mathcal W_\text{$k$XC} = (\Z_n^k, \Z_n, \dots, \Z_n, 0)$, where the zero module is the $k$th entry. We note that even though the fracton order depends on the intersection properties of the $k$ orientations, the conjectured constraint map and weighted QSS only depend on the integer $k$. If these conjectures hold, one consequence is that the $\Z_2$ FCC model is not phase-equivalent to any stack of $k$-planar X-cube models, which can easily be seen from the weighted QSS.

Finally, we comment on an interesting property that sets the FCC model apart from the other models studied. For any $p$-modular fracton order, the description in terms of plane charges becomes simpler upon coarsening the translation symmetry to obtain a trivial action on the constraint module $C$; this can always be done because $C$ is finite. In the $\Z_n$ FCC model, to achieve a trivial $\T$-action on $C$, we need to enlarge the unit cell by a factor of $\mathcal O(n)$ in linear dimension. In contrast, in the other models, the necessary enlargement of the unit cell is independent of $n$. This phenomenon is likely related to the fact that the plane charge content of planons in the FCC model has a quadrupolar form (\emph{e.g.} $(t - 2 + \bar{t}) \mathbf q_{\mu \alpha})$), while the corresponding expressions in the other models are dipoles of plane charges such as $(1-t) \mathbf q_\mu$. It will be interesting to explore this difference in future work, and to find more examples of fracton orders similar to the FCC model in this regard.

\acknowledgments

We are grateful for useful discussions with Agn\`{e}s Beaudry, Yifan Hong, Wilbur Shirley, Charles Stahl and David T. Stephen. We are also especially grateful for collaboration on related work in progress with Wilbur Shirley, and for helpful comments on the manuscript from Agn\`{e}s Beaudry, including pointing out an error in a draft version. The research of MQ (at University of Colorado Boulder, prior to August 2024), EW and MH is supported by the U.S. Department of Energy, Office of Science, Basic Energy Sciences (BES) under Award number DE-SC0014415. Research of MQ after August 2024 was supported, in part, by the Physical Science Division of the University of Chicago. Research of AD at Caltech (prior to 2024) was supported by the Simons Collaboration on Ultra-Quantum Matter (UQM), funded by Simons Foundation grant 651438. The work of all authors benefited from meetings of the UQM Simons Collaboration, supported by Simons Foundation grant 618615.

\appendix

\section{Technical results on fusion theories, $p$-theories and phase invariants}
\label{appendix:technical}

Here we obtain some basic technical results underpinning the discussion of the main text. In order to be clear about what is being assumed, this section is written in ``theorem-proof'' style. While the definitions of a fusion theory and $p$-theory were given in Sec.~\ref{section:definitions}, we re-state them here in a concise form. Some other definitions are given in Sections~\ref{section:definitions} and~\ref{sec:phase-invariants}. In this Appendix we do not assume that $p$-theories are $p$-modular unless explicitly stated. 

\begin{defn}
A fusion theory consists of the data $(\T, S, \eta)$, where:
\begin{enumerate}
\item $\T$ is a group $\T \cong \Z^3$.
\item $S$ is a finitely generated $\Z\T$-module, where every element is of finite order. We let $n$ be the smallest positive integer such that $n x = 0$ for all $x \in S$; such an $n$ is guaranteed to exist because $S$ is finitely generated. Note that the ideal $(n) \subset \Z\T$ annihilates $S$. We let $R = \Z\T/(n) = \Z_n \T$, and view $S$ as an $R$-module. Note that if $S=0$, $R$ is the zero ring.
\item For any $x \in S$, $\T_x \cong \Z^3$ implies $x = 0$.  This is the assumption that there are no non-trivial fully mobile excitations.
\item $\eta : \T \hookrightarrow \R^3$ is an injective group homomorphism, where the group structure on $\R^3$ is the standard vector addition. 
\item If $\{t_1, t_2, t_3\}$ is a set of generators for $\T$, then $\{ \eta(t_1), \eta(t_2), \eta(t_3)\}$ is a vector space basis for $\R^3$. 
\end{enumerate}
\end{defn}

Given a fusion theory, the set of orientations of planons is denoted $A$.

We begin with some basic facts about subgroups of $\Z^3$. These facts rely on the natural inclusion $\Z^3 \subset \R^3$, which allows us to view $\Z^3$ and subsets thereof as subsets of the vector space $\R^3$. The first proposition below tells us that subgroups of $\Z^3$ isomorphic to $\Z^k$ span a $k$-dimensional subspace of $\R^3$.

\begin{proposition}
Let $A \subset \Z^3$ be a subgroup, then $A \cong \Z^k$ for some $k \leq 3$. Supposing $k > 0$, if $\{a_1, \dots, a_k \}$ is a basis for $A$, then $\{ a_1, \dots, a_k \}$ is a linearly independent set of vectors in $\R^3$. \label{prop:sbgrp}
\end{proposition}

\begin{proof}
The first assertion is a basic fact about finitely generated free abelian groups; it follows from Theorem 1.6 of Ref.~\onlinecite{Hungerford_1974}. The second statement follows from Theorem~I in Chapter~4.1 of Ref.~\onlinecite{Samuel_1970}; nonetheless, we provide a proof.

We prove the second statement by induction; it obviously holds for $k=1$ since $a_1 \neq 0$. Suppose the statement has been proved for $k \to k-1$, then $\{ a_1, \dots, a_{k-1} \}$ is a basis for a subgroup of $\Z^3$ isomorphic to $\Z^{k-1}$, and is thus a linearly independent set. Suppose that $\{ a_1, \dots, a_k \}$ is not linearly independent, then there exist real numbers $c_i$, some of which must be nonzero, such that
\begin{equation}
a_k = c_1 a_1 + \cdots + c_{k-1} a_{k-1} \text{.} \label{eqn:lin-comb}
\end{equation}

Now we extend $\{a_1, \dots, a_{k-1} \}$ to $\{a'_1, a'_2, a'_3 \}$, a vector space basis of $\R^3$ such that $a'_i = a_i$ when $i \leq k-1$. We take all the $a'_i$ to have integer components. Viewing $\{ a'_1, a'_2, a'_3 \}$ as the basis for a Bravais lattice, we let $\{ b_1, b_2, b_3 \} \subset \R^3$ be a basis for the corresponding reciprocal lattice, with normalization chosen such that $a'_i \cdot b_j = \delta_{i j}$. With this choice of normalization, the vectors $b_i$ have rational components. 

Taking the dot product of Eq.~(\ref{eqn:lin-comb}) with $b_i$ for $i \leq k-1$ gives
\begin{equation}
a_k \cdot b_i = c_i
\end{equation}
The left-hand side is rational, so the coefficients $c_i$ are rational. Therefore there exists $N \in \N$ such that $N c_i$ is an integer for all $i=1,\dots,k-1$ and such that
\begin{equation}
N c_1 a_1 + \cdots + N c_{k-1} a_{k-1} - N a_k = 0 \text{.}
\end{equation}
This contradicts $\{a_1, \dots, a_k\}$ being a basis for $A$, and thus $\{ a_1, \dots, a_k\}$ must be a set of linearly independent vectors.
\end{proof}

\begin{proposition}
Let $A$ and $C$ be subgroups of $\Z^3$ with $A \cong \Z^k$ and $C \cong \Z^\ell$, assuming without loss of generality that $k \leq \ell \leq 3$. Viewing $\Z^3 \subset \R^3$, we define $m$ to be the dimension of the intersection of subspaces spanned by $A$ and $C$; that is,
\begin{equation}
m = \operatorname{dim} \Big( \operatorname{span} A \cap \operatorname{span} C \Big)\text{.}
\end{equation}
Then $A \cap C \cong \Z^m$.
\end{proposition}

\begin{proof}
By Proposition~\ref{prop:sbgrp}, $A \cap C \cong \Z^{m'}$, and $m' = \operatorname{dim} \operatorname{span} A \cap C$, which implies $m' \leq m$ since
$\operatorname{span} A \cap C \subset \operatorname{span} A \cap \operatorname{span} C$ for any subsets $A$ and $C$ of a vector space. Moreover, $m \leq k$ because $\operatorname{span} A \cap \operatorname{span} C \subset \operatorname{span} A$. To summarize, $m' \leq m \leq k$. We need to show $m' = m$. 

If $m = 0$, then clearly $A \cap C = \{ 0 \}$ and the claim holds. We now assume $m > 0$. Let $\{ a_1, \dots, a_k \}$ be a basis for $A$ (as a group); by Proposition~\ref{prop:sbgrp} this is a linearly independent set in $\R^3$, so we can extend it to a vector space basis, $\{ a_1, a_2, a_3 \}$, where we choose all the $a_i$ to have integer components. Let $\{ b_1, b_2, b_3 \}$ be a basis for the corresponding reciprocal lattice, with normalization $a_i \cdot b_j = \delta_{i j}$, so that the vectors $b_i$ have rational components. Similarly, let let $\{ c_1, \dots, c_\ell \}$ be a basis for $C$, and extend to $\{ c_1, c_2, c_3 \}$, a basis for $\R^3$. Finally we let $\{ d_1, d_2, d_3 \}$ be a basis for the corresponding reciprocal lattice.

By making a table of possibilities for $k$, $\ell$ and $m$, we see that $\operatorname{span} A \subset \operatorname{span} C$ (and thus $m = k$), except in the case $k = \ell = 2$ and $m = 1$ (two planes intersecting along a line). We first assume $\operatorname{span} A \subset \operatorname{span} C$ and then treat the latter case separately. For each of the $a_i$ basis vectors we have $a_i = \sum_{j=1}^3 \alpha_{i j} c_j$ for $\alpha_{i j} \in \R$. Taking the dot product with $d_j$ we have $\alpha_{i j} = a_i \cdot d_j$, which implies $\alpha_{i j} \in \Q$. Therefore there exist $N_i \in \N$ such that $N_i \alpha_{i j} \in \Z$ for $j = 1,2,3$, which implies that $N_i a_i$ is an integer linear combination of the $c_i$'s, and thus $N_i a_i \in C$. Therefore $N_i a_i \in A \cap C$ and $\{ N_1 a_i, \dots, N_k a_k \}$ is a basis for a $\Z^k$ subgroup of $A \cap C$. For $A \cap C$ to have a $\Z^k$ subgroup, we must have $m' \geq k$ (by Theorem 1.6 of Ref.~\onlinecite{Hungerford_1974}). But $m' \leq k$, so $m' = k = m$.

Now we consider the case $k = \ell = 2$ and $m = 1$. We know $m' \leq 1$ and we will show $m' = 1$ by proving $m' \geq 1$. Choose a nonzero vector $v \in \Span A \cap \Span C$, then
\begin{equation}
v = \alpha_1 a_1 + \alpha_2 a_2 = \gamma_1 c_1 + \gamma_2 c_2 \text{,} \label{eqn:v-eqn}
\end{equation}
where the coefficients are real numbers.  If $\alpha_1 \neq 0$, we multiply $v$ by a scalar to choose $\alpha_1 = 1$, otherwise we choose $\alpha_2 = 1$. With this choice, our goal is to show that $\alpha_i, \gamma_i \in \Q$. From this it follows that we can multiply $v$ by $N \in \N$ such that $N \alpha_i, N \gamma_i \in \Z$. Then $N v \in A \cap C$ and $N v$ generates a subgroup of $A \cap C$ isomorphic to $\Z$, implying $m' \geq 1$.

Extending $a_1, a_2$ to a basis as above, we have $c_i = \sum_{j = 1}^3 q_{i j} a_j$, which implies $q_{i j} = c_i \cdot b_j$, and thus $q_{i j} \in \Q$. We equate coefficients of $a_i$ on each side of Eq.~\ref{eqn:v-eqn} to obtain the three equations
\begin{eqnarray}
\alpha_1 &=& \gamma_1 q_{1 1} + \gamma_2 q_{2 1} \label{ag1} \\
\alpha_2 &=& \gamma_1 q_{1 2} + \gamma_2 q_{2 2} \label{ag2} \\
0 &=& \gamma_1 q_{1 3} + \gamma_2 q_{2 3} \text{.} \label{ag3}
\end{eqnarray}

At least one of $q_{13}$ and $q_{23}$ must be non-zero, since otherwise both $c_1$ and $c_2$ lie in $\operatorname{span} A$, which contradicts $m = 1$. First suppose $q_{13} \neq 0$ and $q_{23} = 0$. Then $\gamma_1 = 0$ by Eq.~\ref{ag3}. If we chose $\alpha_1 = 1$, we have $\gamma_2 \in \Q$ by Eq.~\ref{ag1}, and Eq.~\ref{ag2} implies $\alpha_2 \in \Q$. On the other hand, if we chose $\alpha_2 = 1$, Eq.~\ref{ag2} gives $\gamma_2 \in \Q$. In both cases, all the $\alpha_i$ and $\gamma_i$ are rational.  The same conclusion follows by essentially the same arguments when $q_{13} = 0$ and $q_{23} \neq 0$.

Now suppose both $q_{13}$ and $q_{23}$ are non-zero, then Eq.~\ref{ag3} implies both $\gamma_i$ are non-zero, and the ratio $\gamma_1 / \gamma_2 \in \Q$. Suppose we chose $\alpha_1 = 1$, then Eq.~\ref{ag1} can be written $1 = \gamma_1(q_{11} + (\gamma_2 / \gamma_1) q_{21} )$. This implies $\gamma_1 \in \Q$, and thus also $\gamma_2 \in \Q$, and Eq.~\ref{ag2} then implies $\alpha_2 \in \Q$. Suppose instead $\alpha_1 = 0$ and $\alpha_2 = 1$, then similarly Eq.~\ref{ag2} implies $\gamma_1, \gamma_2 \in \Q$. We have shown in both cases that all the coefficients $\alpha_i$ and $\gamma_i$ are rational.
\end{proof}

Now suppose we have a fusion theory $(\T, \eta, S)$, and let $t_1, t_2, t_3$ be generators for $\T$. Let $T: \R^3 \to \R^3$ be the linear isomorphism defined by $T \eta(t_i) = e_i$, where $\{ e_1, e_2, e_3 \}$ is the standard basis for $\R^3$. Then viewing $\Z^3$ as a subset of $\R^3$, we have an isomorphism $T \circ \eta : \T \to \Z^3$. This observation implies the following corollaries, which are just straightforward translations of the preceding propositions into the langauge of fusion theories:

\begin{corollary}
Given a fusion theory, let $G \subset \T$ be a subgroup with $G \cong \Z^k$ for some $k \leq 3$. Supposing $k > 0$, if $\{ g_1, \dots, g_k \}$ is a basis for $G$, then $\{ \eta(g_1), \dots, \eta(g_k) \}$ is a linearly independent set in $\R^3$. \label{cor:dim}
\end{corollary}

\begin{corollary}
Given a fusion theory, let $G$ and $H$ be subgroups of $\T$ with $G \cong \Z^k$ and $H \cong \Z^\ell$, assuming $k \leq \ell \leq 3$. Then $G \cap H \cong \Z^m$, where
\begin{equation}
m = \operatorname{dim} \Big( \Span \eta(G) \cap \Span \eta(H) \Big) \text{.}
\end{equation} \label{cor:int-dim}
\end{corollary}

It immediately follows that if $p, p' \in S$ are two planons of the same orientation, then $\T_p \cap \T_{p'} \cong \Z^2$.

\begin{defn}
Given a fusion theory, for each orientation $a \in A$, $S_a \subset S$ is the submodule generated by $a$-oriented planons.
\end{defn}

\begin{proposition}
Given a fusion theory and an orientation $a \in A$, every non-zero element of $S_a$ is an $a$-oriented planon. \label{prop:a-oriented}
\end{proposition}

\begin{proof}
A general element $x \in S_a$ can be written $x = \sum_{i = 1}^m r_i p_i$, where $r_i \in R$ and each $p_i$ is an $a$-oriented planon. We assume $x \neq 0$. Let $\tilde{\T}_x = \cap_{i = 1}^k \T_{p_i}$, then $\tilde{\T}_x \cong \Z^2$ by Corollary~\ref{cor:int-dim}, and $\tilde{\T}_x \subset \T_x$. Since $\T_x$ has a subgroup isomorphic to $\Z^2$, either $\T_x \cong \Z^2$ or $\T_x \cong \Z^3$. But the latter case would violate the assumption that there are no non-trivial fully mobile excitations, so $\T_x \cong \Z^2$. We have $a = \spn \eta(\tilde{\T}_x) = \spn \T_x$, so $x$ is an $a$-oriented planon.
\end{proof}

\begin{proposition}
Given a fusion theory and an orientation $a \in A$, we have $\T_{S_a} \cong \Z^2$.
\end{proposition}

\begin{proof}
The module $S_a$ is finitely generated, so let $p_1, \dots, p_m$ be generators, each of which is an $a$-oriented planon by Proposition~\ref{prop:a-oriented}. Then $\T_{S_a} = \T_{p_1} \cap \cdots \cap \T_{p_m} \cong \Z^2$, where the isomorphism follows because the $p_i$ are planons of the same orientation. 
\end{proof}

\begin{defn}
Given a fusion theory, the submodule of $S$ generated by all planons is denoted $S^{(1)} = \sum_{a \in A} S_a$.
\end{defn}

\begin{proposition}
Given a fusion theory, $S^{(1)} = \bigoplus_{a \in A} S_a$. \label{prop:planon-direct-sum}
\end{proposition}

\begin{proof}
It is enough to show that if $\sum_{i=1}^m p_i = 0$, where $p_i \in S_{a_i}$ and all the orientations $a_i$ are distinct, then all the $p_i$ are zero.

Suppose $\sum_{i=1}^m p_i = 0$ where the $p_i$ are as above and each $p_i \neq 0$. Pick $g \in \T$ such that $g p_m = p_m$ and $g p_i \neq p_i$ for all $i < m$. Such a $g$ always exists; to see this, observe that each $a(i)$ for $i < m$ intersects $a(m)$ along a line within $a(m)$ that passes through the origin. To choose $g$, we need only go sufficiently far from the origin to find a lattice point in $\eta(\T_{p_m})$ away from these lines of intersection.

Clearly $(1-g) p_m = 0$. Moreover $(1-g)p_i$ for $i < m$ is non-zero and thus an $a$-oriented planon; if it were zero, $p_i$ would be fully mobile, contradicting Proposition~\ref{prop:a-oriented}. Therefore, acting with $(1-g)$ on both sides of the expression $\sum_{i=1}^m p_i = 0$ gives a new expression $\sum_{i=1}^{m-1} p'_i = 0$, where each $p'_i \in S_{a_i}$ is non-zero. 

We repeat this procedure until we obtain an expression $p''_1 = 0$, where $p''_1 \in S_{a_1}$ is non-zero, a contradiction.
\end{proof}

\begin{corollary}
For any fusion theory, the set of orientations $A$ is finite. 
\end{corollary}

\begin{proof}
Because $R$ is Noetherian and $S$ is finitely generated, $S^{(1)} \subset S$ is finitely generated. Let $x_1, \dots, x_m \in S^{(1)}$ be generators. By Proposition~\ref{prop:planon-direct-sum}, each $x_i$ can be expressed uniquely as $x_i = p^i_{a_{i1}} + \cdots + p^i_{a_{i k_i}}$, where each $p^i_{a_{ij}} \in S_{a_{ij}}$ is non-zero, and for fixed $i$, the $a_{i j}$ are distinct orientations. If $A$ is infinite, there is some $a \in A$ such that no element of $S_a$ appears in any of the expressions for the $x_i$. Non-zero elements of this $S_a$ cannot be written as an $R$-linear combination of the generators, a contradiction.
\end{proof}

\begin{defn}
A $p$-theory is a fusion theory $(\T, S, \eta)$ together with, for each $a \in A$, a bilinear form $\langle \cdot , \cdot \rangle_a : S_a \times S \to \Z_n$ and a unit vector $\hat{n}_a \in \R^3$ normal to $a$. The following conditions are satisfied:
\begin{enumerate}
\item  If $p, p' \in S_a$, then $\langle p, p' \rangle_a = \langle p', p \rangle_a$.
\item For any $g \in \T$, $p \in S_a$ and $x \in S$, we have $\langle p, x \rangle_a = \langle gp, gx \rangle_a$.
\item For any $p \in S_a$ and any $x \in S$, let $g \in \T$ be such that $g^i p = p$ implies $i = 0$. Then there exists $m \in \N$ such that $|i| \geq m$ implies $\langle g^i p, x \rangle_a = 0$.
\end{enumerate}
\end{defn}

The data of a $p$-theory is the 5-tuple $P = (\T,S,\eta, \{ \langle \cdot, \cdot \rangle_a \}, \{ \hat{n}_a \})$. 

Physically, one expects that a planon has vanishing statistics with any excitation mobile along a direction transverse to the planon's orientation, because such an excitation can ``escape'' detection by the planon. We show that this physical expectation is implied by the structure of a $p$-theory:

\begin{proposition}
Given a $p$-theory, choose a planon $p \in S_a$ and $x \in S$ such that there is a non-trivial $g \in \T$ with $g x = x$ and $\eta(g) \notin a(p)$.  Then $\langle p , x \rangle = 0$.  \label{prop:planon-perp-mobile}
\end{proposition}

\begin{proof}
By locality of mutual statistics, we know there exists $\ell \in \N$ such that $\langle g^{-\ell} p, x \rangle = 0$.  But since $g^{\ell} x = x$, we have
\begin{equation}
0 = \langle g^{-\ell} p, x \rangle = \langle g^{\ell} g^{-\ell} p, g^{\ell} x \rangle = \langle p, x \rangle \text{.}
\end{equation}
\end{proof}

Below, we let $Q_a$ be the module of $a$-oriented plane charges defined in Sec.~\ref{section:definitions}. The following results tell us about the mobility of plane charges. 

\begin{proposition}
In any $p$-theory, given an orientation $a \in A$, $\T_{S_a} \subset \T_{Q_a}$. Moreover in a $p$-modular theory, $\T_{S_a} = \T_{Q_a}$. \label{prop:pc-mobility1}
\end{proposition}

\begin{proof}
Take $g \in {\cal T}_{S_a}$, and $q \in Q_a$.  We can write $q = \pi_a(x)$ for some $x \in S$, and $gq = \pi_a(g x)$.  For any $p \in S_a$, we have
\begin{equation}
\langle p, g x \rangle = \langle g^{-1} p, x \rangle = \langle p, x \rangle \text{,}
\end{equation}
since $g^{-1} p = p$.  Therefore $x \sim_a g x$ which implies $g q = q$, so $g \in {\cal T}_{Q_a}$.

Now suppose we have a $p$-modular theory. Take $g \in {\cal T}_{Q_a}$, then $g q = q$ for all $q \in Q_a$, and therefore $g x \sim_a x$ for all $x \in S$.  Then for all $p \in S_a$ and all $x \in S$ we have
\begin{equation}
\langle g^{-1} p, x \rangle = \langle p, g x \rangle = \langle p, x \rangle \text{.}
\end{equation}
We restrict $x \in S_a$ and write $x = p'$.  Then this becomes the statement that for all $p', p \in S_a$,
\begin{equation}
\langle p', g^{-1} p \rangle = \langle p', p \rangle \text{.}
\end{equation}
By Proposition~\ref{prop:planon-perp-mobile}, all planons of orientations different from $a$ have vanishing mutual statistics with both $p$ and $g^{-1}p$. Therefore $p$-modularity implies $g^{-1} p = p$ for all $p \in S_a$,  so $g \in {\cal T}_{S_a}$.
\end{proof}

In a $p$-modular theory it immediately follows that $\T_{Q_a} \cong \Z^2$. However, we also want to know about the mobility of individual plane charges. The following result does not rely on $p$-modularity and implies that $\T_{Q_a} \cong \Z^2$ in general:

\begin{proposition}
Given a $p$-theory, let $q$ be any nonzero element of $Q_a$. Then $\Span \eta(\T_q) = a$ and ${\cal T}_q \cong \Z^2$. In addition, ${\cal T}_{Q_a} \cong \Z^2$.  \label{prop:q-mobility}
\end{proposition}

\begin{proof}
Take $g \in {\cal T}$ such that $\eta(g) \notin a$.  We have $q = \pi_a(x)$ and $g q = \pi_a(g x)$ for some $x \in S$.  Now assume $g q = q$, then for all planons $p \in S_a$,
\begin{equation}
\langle g^{-1} p, x \rangle = \langle p, g x \rangle = \langle p, x \rangle \text{.}
\end{equation}
Since this holds for all planons, it holds putting $p \to g^{-1} p$, and we obtain $\langle g^{-2} p , x \rangle  = \langle p, x \rangle$.  We can keep iterating this to obtain
\begin{equation}
\langle g^{-k} p, x \rangle = \langle p, x \rangle \label{eqn:gk-contd}
\end{equation}
for all $k \in \N$.  Since $\pi_a(x) \neq 0$, there is some planon $p$ for which $\langle p , x \rangle \neq 0$, and Eq.~(\ref{eqn:gk-contd}) contradicts locality of statistics. Therefore $g \notin \T_q$, and we have shown $\eta(\T_q) \subset a$.

By Proposition~\ref{prop:pc-mobility1}, ${\cal T}_{S_a} \subset {\cal T}_{Q_a}$, and ${\cal T}_{Q_a} \subset {\cal T}_q$.  Therefore ${\cal T}_{S_a} \cong \Z^2 \subset {\cal T}_q$. The span of this $\Z^2$-isomorphic subgroup is a plane in $\R^3$, so it follows that $\Span \eta(\T_q) = a$, and thus $\T_q \cong \Z^2$. Because $\T_{Q_a}$ is finitely generated, it follows immediately that $\T_{Q_a} \cong \Z^2$.
\end{proof}

Recall from Sec.~\ref{section:definitions} that $Q = \bigoplus_{a \in A} Q_a$ and $\pi : S \to Q$ is given by $\pi = \sum_{a \in A} i_a \circ \pi_a$, where $i_a : Q_a \hookrightarrow Q$ is the natural inclusion. The constraint module is defined as the quotient $C = Q / \pi(S)$. At the risk of being overly pedantic but to avoid possible confusion, we say an $R$-module is finite if its underlying set is finite. (Confusion can arise because a finitely generated $R$-module need not be finite.)

\begin{theorem}
Given a $p$-theory, the constraint module $C = Q / \pi(S)$ is finite. \label{thm:finite-constraints}
\end{theorem}

We first obtain some intermediate results before proving the theorem.

\begin{proposition}
For each $a \in A$ there exists $\alpha_a \in R$ such that $\alpha_a Q_b = 0$ for $b \neq a$ and $\alpha_a q_a \neq 0$ if $q_a \in Q_a$ is non-zero.
\end{proposition}

\begin{proof}
If $|A| = 1$, then take $\alpha_a = 1$.  Now assume $|A| > 1$.  Fix $a \in A$, and take $b \in A$ with $b \neq a$.  Choose $g_{ab} \in {\cal T}_{Q_b}$ so that $\eta(g_{ab}) \notin a$ (this is always possible).  Then observe that $(1 - g_{ab})$ annihilates $Q_b$ and, using Proposition~\ref{prop:q-mobility}, satisfies $(1-g_{ab})q_a \neq 0$ for any non-zero $q_a \in Q_a$.  We make such a choice for each $b \neq a$ and define
\begin{equation}
\alpha_a = \prod_{b\neq a} (1 - g_{a b} ) \text{.}
\end{equation}
This obviously satisfies the desired properties.
\end{proof}

\begin{corollary}
We have $i_a(\alpha_a Q_a) \subset \pi(S)$. \label{cor:alQ-in-S}
\end{corollary}

\begin{proof}
Take $q_a \in Q_a$, then there exists $x \in S$ such that $q_a = \pi_a(x)$.  Moreover $\alpha_a q_a = \pi_a (\alpha_a x)$.  In addition, for $b \neq a$, $\pi_b (\alpha_a x) = \alpha_a \pi_b(x) = 0$ because $\alpha_a$ annihilates $Q_b$. Now $\pi = \sum_{c \in A} i_c \circ \pi_c$, so $\pi(\alpha_a x) = i_a \circ \pi_a(\alpha_a x) = i_a ( \alpha_a q_a)$.
\end{proof}

As an aside (we do not use the following corollary below), we can say more in a $p$-modular theory:

\begin{corollary}
In a $p$-modular theory, for any $x \in S$ with $\pi_a(x) \neq 0$,  $\alpha_a x$ is a non-trivial $a$-planon. \label{cor:alpha_a-planon}
\end{corollary}

\begin{proof}
Because $\pi(\alpha_a x) \in i_a( Q_a )$ and $\pi(\alpha_a x) = \alpha_a \pi(x) \neq 0$, we have ${\cal T}_{\pi(\alpha_a x)} \cong \Z^2$ by Proposition~\ref{prop:q-mobility}.  Because $\pi$ is injective, ${\cal T}_{\alpha_a x} = {\cal T}_{\pi(\alpha_a x)}$.
\end{proof}

\begin{lemma}
The quotient $Q_a / \alpha_a Q_a$ is finite.  \label{lem:qaq-finite-ab}
\end{lemma}

\begin{proof}
$Q_a$ is finitely generated as an $R$-module, so let $\tilde{q}_i$ be $R$-module generators for $Q_a$, with $i = 1,\dots,N$.  

The translation group ${\cal T}$ can be decomposed as a product ${\cal T} = {\cal T}_\parallel \times {\cal T}_\perp$, where ${\cal T}_\parallel \cong \Z^2$ is all translations whose image under $\eta$ lies in $a$, and ${\cal T}_\perp \cong \Z$ is generated by some translation not coplanar with $a$.  (Note that ${\cal T}_\perp$ does not necessarily consist of translations normal to the $a$-plane.)  We use polynomial notation and denote  generators of ${\cal T}_\parallel$ by $x$ and $y$, and denote the generator of ${\cal T}_\perp$ by $z$.

For each generator $\tilde{q_i}$, the set ${\cal T}_\parallel q_i = \{ g q_i | g \in {\cal T}_\parallel \}$ is finite because ${\cal T}_\parallel / {\cal T}_{\tilde{q_i}}$ is finite and ${\cal T}_{\tilde{q}_i}$ acts trivially on $q_i$.  Using this, we define a new and possibly larger finite set of generators $\cup_{i = 1}^N {\cal T}_\parallel \tilde{q}_i \equiv \{ q_1, \dots, q_M \}$.  Clearly ${\cal T}_\parallel$ acts on $\{ q_1, \dots, q_M \}$ by permutation, and therefore $Q_a$ as an abelian group is generated by the set $\{ z^\ell q_i | \ell \in \Z, i \in \{1, \dots, M\} \}$.  It follows that $Q_a / \alpha_a Q_a$ as an abelian group is generated by
$\Gamma = \{ z^\ell q_i + \alpha_a Q_a | \ell \in \Z, i \in \{1, \dots, M\} \}$.  

We will show that $\Gamma$ is a finite set, which implies that the group it generates is finite because every element is of finite order.  To do this we will show that there exists $k \in \N$ such that each element of $\Gamma$ can be written in the form
\begin{equation}
z^\ell q_i + \alpha_a Q_a = \sum_{j = 1}^M p_j(z) q_j + \alpha_a Q_a \text{,}  \label{eqn:desired-form}
\end{equation}
where each $p_j(z)$ is a polynomial of the form $p_j(z) = c_0 + c_1 z + \cdots + c_k z^k$ with $c_i \in \Z_n$.  There are only finitely many expressions of this form, so $\Gamma$ is finite.  

First we express $\alpha_a = \prod_{b \neq a} (1- g_{ab})$ in polynomial notation.  Each $g_{ab}$ is a monomial $g_{ab} = x^{n_x} y^{n_y} z^{n_z}$, with $n_z \neq 0$ since $g_{ab}$ is non-coplanar with $a$.  We can always take $n_z > 0$ (if $n_z < 0$, then instead of choosing $g_{ab}$, we choose $g^{-1}_{ab}$).  Therefore, for $k$ the maximum positive degree of $\alpha_a$ with respect to $z$, we have
\begin{equation}
\alpha_a = \sum_{j = 0}^k \alpha_j(x,y) z^j \text{,}
\end{equation}
where $\alpha_0 = 1$ and $\alpha_k(x,y)$ is a monomial.  

Now we consider the element of $\Gamma$ with coset representative $z^\ell q_i$.  If $0 \leq \ell \leq k$, then $z^\ell q_i$ is already in the desired form.  Suppose that $\ell < 0$, then we are free to change the coset representative by
\begin{equation}
z^\ell q_i \to z^\ell q_i  - \alpha_a z^{\ell} q_i \text{.}
\end{equation}
This expression can be written in the form $\sum_{j = 1}^M p_j(z) q_j$, where each $p_j$ has degree running from $\ell+1$ through $\ell+k$.  If now $\ell+1 = 0$, we are done.  If not, we can proceed iteratively, subtracting off terms of the form $\alpha_a z^{\ell+1} q_j$ to cancel the degree $\ell+1$ terms in the $p_j$'s.  Proceeding in this way, we can keep increasing the degree of terms appearing in the $p_j$'s until all degrees lie in the range zero through $k$.

If instead $\ell > k$, we can proceed in essentially the same manner to decrease the degree.  Here the fact that $\alpha_k$ is a monomial is important.  
\end{proof}

Now we can proceed to prove the theorem:

\begin{proof}[Proof of Theorem \ref{thm:finite-constraints}]
Since $i_a (\alpha_a Q_a) \subset \pi(S)$, we also have $\bigoplus_{a \in A} \alpha_a Q_a \subset \pi(S)$.  Then, by the third isomorphism theorem,
\begin{equation}
C = \frac{Q}{\pi(S)} \cong  \frac{Q / \bigoplus_a \alpha_a Q_a}{ \pi(S) / \bigoplus_a \alpha_a Q_a } \text{.}
\end{equation}
But
\begin{equation}
\frac{Q}{\bigoplus_a \alpha_a Q_a} \cong \bigoplus_a \frac{Q_a}{\alpha_a Q_a} \text{,}
\end{equation}
which is finite.  Therefore $C = Q / \pi(S)$ is a quotient of a finite group and is thus also finite.  
\end{proof}

\begin{theorem}
In a $p$-modular $p$-theory, the weighted QSS $\mathcal W_i = S / S^{(i)}$ are finite. \label{thm:finite-WQSS}
\end{theorem}

\begin{lemma}
In a $p$-modular $p$-theory, $Q_a / \pi_a(S_a)$ is finite. \label{lem:qasa-finite}
\end{lemma}

\begin{proof}
First we show $\alpha_a Q_a \subset \pi_a(S_a)$. By Corollary~\ref{cor:alQ-in-S}, $i_a(\alpha_a Q_a) \subset \pi(S)$. In a $p$-modular $p$-theory, elements of $i_a(Q_a) \cap \pi(S)$ are $a$-planons, so $i_a(\alpha_a Q_a) \subset i_a(Q_a) \cap \pi(S) \subset \pi(S_a) = i_a \circ \pi_a (S_a)$. The claim follows.

Then we have
\begin{equation}
\frac{Q_a}{\pi_a(S_a)} \cong \frac{ Q_a / \alpha_a Q_a}{\pi_a(S_a) / \alpha_a Q_a} \text{,}
\end{equation}
and $Q_a / \pi_a(S_a)$ is a quotient of a finite group by Lemma~\ref{lem:qaq-finite-ab}.
\end{proof}

\begin{proof}[Proof of Theorem \ref{thm:finite-WQSS}]
We have
\begin{equation}
\frac{Q}{\pi(S^{(1)})} \cong \bigoplus_{a \in A} \frac{Q_a}{\pi_a(S_a)} \text{,}
\end{equation}
so $Q / \pi(S^{(1)})$ is finite by Lemma~\ref{lem:qasa-finite}. Therefore
\begin{equation}
\mathcal W_1 \cong \frac{\pi(S)}{\pi(S^{(1)})} \subset \frac{Q}{\pi(S^{(1)})}
\end{equation}
is finite. Finally, the remaining weighted QSS are finite because
\begin{equation}
\mathcal W_i \cong \frac{\pi(S)}{\pi(S^{(i)})} \cong 
\frac{\pi(S) / \pi(S^{(1)})}{\pi(S^{(i)})/\pi(S^{(1)})} \text{.}
\end{equation}
That is, we can express each $\mathcal W_i$ as a quotient of a finite group.
\end{proof}

\begin{proposition}
A $p$-modular $p$-theory is essentially 1-planar if and only if $C = 0$.
\end{proposition}

\begin{proof}
First suppose $C = 0$. Then $Q = \pi(S)$, and thus $\pi$ gives an isomorphism $S \cong Q = \bigoplus_{a \in A} Q_a$. Non-zero elements of $i_a(Q_a)$ are $a$-oriented planons, so $S_a \cong Q_a$, and we have $S = S^{(1)}$.

Now suppose $S = S^{(1)}$. Take $q \in i_a(Q_a)$. Because $\pi$ is surjective, we have $q = \pi(x)$ for some $x \in S$. We have $x = \sum_{b \in A} p_b$ for $p_a \in S_a$. For $q$ to lie in $i_a(Q_a)$, we must have $p_b = 0$ for $b \neq a$, so $q = \pi(p_a) = i_a \circ \pi_a(p_a)$. This implies that $\pi_a$ is surjective, so $Q_a / \pi_a(S_a) = 0$. Then
\begin{equation}
C = \frac{Q}{\pi(S)} \cong \bigoplus_{a \in A} \frac{Q_a}{\pi_a(S_a)} = 0 \text{.}
\end{equation}

\end{proof}

Finally we give a definition of isomorphism of $p$-theories:

\begin{defn}
An isomorphism between two $p$-theories $P_1$ and $P_2$, with their respective data also labeled by ``1'' and ``2,'' is a pair of maps $\alpha_\T : \T_1 \to \T_2$ and $\alpha_S : S_1 \to S_2$ satisfying the following properties:
\begin{enumerate}
    \item The map $\alpha_\T$ is a group isomorphism satisfying $\eta_2 \circ \alpha_\T = \eta_1$. We use the same symbol to denote the induced ring isomorphism $\alpha_\T : R_1 \to R_2$, where $R_i = \Z_n \T_i$ for $i=1,2$.    
    \item The map $\alpha_S$ is a group isomorphism and also satisfies $\alpha_S (r x) = \alpha_\T(r) \alpha_S(x)$ for any $g \in \T_1$ and $r \in R_1$.
    \item The above properties give a bijection between the sets of planon orientations, which we use to identify $A \equiv A_1 = A_2$. Then, for all $a \in A$, we require $\hat{n}^1_a = \hat{n}^2_a$, and $\langle p, x \rangle^1_a = \langle \alpha_S(p), \alpha_S(x) \rangle^2_a$ for all planons $p \in S_1$ and all $x \in S_1$.
\end{enumerate}    
We write $P_1 \cong P_2$ when two $p$-theories are isomorphic.
\end{defn}

\section {A sufficient condition for foliated RG}
\label{appendix:foliated-RG}

\renewcommand{\forall}{\text{ for all }}

Throughout this Appendix, $R = \Z_n \T$.  Recall that the subgroup $\T^{(m)} \subset \T$ consists of all $g \in \T$ for which there exists $h \in \T$ such that $g = h^m$. If $\mathcal G \subset \mathcal T$ is a subgroup, let $R_{\mathcal G} = \Z_n \T / \mathcal G$. This is viewed as an $R$-module where the $R$-action is given by the surjective ring homomorphism induced by the quotient map $\T \to \T / \mathcal G$. For an element $\alpha \in R_{\G}$, we write $\alpha = \sum_{g \in \T/\G} (\alpha)_g g$. The \emph{constant part} of $\alpha$ is $(\alpha)_1$, and $\alpha$ is called \emph{constant} if $(\alpha)_g = 0$ for all $g \neq 1$. Given an integer $m > 1$, we define a $\Z_n$-module map $D_m : \Z_n[t^\pm] \to \Z_n[t^\pm]$ by $D_m(t^k) = t^{m k}$. Note that $D_m$ is a ring homomorphism, but is not a $\Z_n[t^{\pm}]$-module map using the standard action of $\Z_n[t^{\pm}]$ on itself for both domain and codomain. Given a $p$-modular $p$-theory with set of orientations $A$, we recall that for each $a \in A$ we define $R_a \equiv R_{\T_a}$. Moreover, as discussed at the end of Sec.~\ref{section:definitions}, the injective $R$-module homomorphisms $\Phi_a : \operatorname{ker} \omega|_{Q_a} \to Q^*_a$ encode the statistical data of the $p$-theory.

\begin{lemma}
    \label{lemma:compatibility-with-scaling}
    Let $\Delta \in \Z_n[t^\pm]$
    be a polynomial of the form $\Delta = d_0 + d_1 t + \dots + d_\delta t^\delta$,
    where $d_0$ and $d_\delta$ are units in $\Z_n$,
    and suppose $m > 1$ is an integer such that $\Delta | 1 - t^m$. 
    Then 
   $D_m(\Delta) = \Delta( 1 + \Theta)$ where $\Theta = \sum_{i \in \Z} e_i t^i$ satisfies $e_{m j} = 0$ for all $j \in \Z$.
\end{lemma}

\begin{proof}
    Let $\Gamma$ be the polynomial such that $\Gamma \Delta = 1 - t^{m-1}$. Because $d_0$ and $d_\delta$ are units, $\Gamma$ can only contain terms of non-negative degree, and the degree $\gamma = \operatorname{deg} \Gamma$ satisfies $\gamma + \delta = m-1$.
    Then we can express 
    \begin{align}
        D_m(\Delta)
        &=
        d_0 + d_1 t^m + \dots + d_\delta t^{\delta m}
        \nonumber
        \\
        &=
        d_0 + d_1 t + \dots + d_\delta t^{\delta}
        -
        d_1 t (1 - t^{m-1})
        -
        \dots 
        - d_\delta t^{\delta} (1 - t^{\delta(m-1)})
        \nonumber
        \\
        &=
        \Delta
        \left(
        1 
        -
        \Gamma\left(
        d_1 t
        + 
        d_2 t^2 (1 + t^{m-1})
        +
        \dots
        +
        d_\delta t^\delta (1 + t^{m-1}
        +
        \dots t^{(\delta - 1)(m -1)}
        )
        \right)
        \right)
        \nonumber
        \\
        &\equiv \Delta (1 + \Theta).
    \end{align}
    Notice that $\Theta$ is a sum of terms of the form
    $-t^{\zeta m}
    t^\eta \Gamma$,
    where $1 \leq \eta \leq \delta$; but since $\gamma + \delta = m - 1$, $\Theta $  is a linear combination of monomials the form $t^{\zeta m} t^\eta$ with $1 \leq \eta \leq m - 1$.
\end{proof}

\begin{proposition}
    \label{prop:rg-2}
    Let $P$ be a $p$-modular $p$-theory. Let $m > 1$ be an integer such that $\T^{(m-1)}$ acts trivially on the constraint module $C$; such an integer always exists because $C$ is finite. Suppose that for each $a \in A$:
    \begin{enumerate}
        \item $\T / \T_a \cong \Z$. (Note this implies $R_a \cong \Z_n[t^\pm]$; we choose an isomorphism.)
        \item $Q_a$ is a free $R_a$-module.
        \item There is an $R_a$-module basis $q^a_1, \dots, q^a_{\ell_a}$ for $Q_a$ together with a basis $\Delta^a_1 q^a_1, \dots, \Delta^a_{\ell_a} q^a_{\ell_a}$ for $\ker \omega|_{Q_a}$, where $\Delta^a_i \in \Z_n[t^{\pm}]$ and either $\Delta^a_i$ or $\overline{\Delta^a_i}$ is of the form $d^a_{i0} + d^a_{i1} t^1 + \dots + d^a_{i \delta_i} t^{\delta_i}$ where $d^a_{i0}$ and $d^a_{i\delta_i}$ are units in $\Z_n$, and such that $\Phi_a(\Delta^a_i q^a_i)(q^a_j)$ is constant for all $i, j$.
    \end{enumerate}
     Then $P \to_m P' \cong P^{(m)} \ominus P_2$, where $P_2$ essentially $1$-planar, and where $P^{(m)}$ was defined in Sec.~\ref{section:RG}.
\end{proposition}

Note that assumption \#1 can always be made to hold by suitable coarsening. 

\begin{proof}
    Anticipating coarsening the $p$-theory to the translation symmetry $\T^{(m)} \subset \T$, we denote $\T' = \T^{(m)}$ and $R' = \Z_n \T' \subset R$. Fix an orientation $a \in A$. We choose a basis for $Q_a$ and $\operatorname{ker} \omega|_{Q_a}$ as in assumption \#3. For each basis element $q^a_i \in Q_a$, there exists $g^a_i \in \T$ such that $g^a_i q^a_i = t q^a_i$. Using the given basis, we define a $\Z_n$-module map $D^a_m : Q_a \to Q_a$ by $D^a_m(t^i q^a_j) = D_m(t^i) q^a_j = t^{m i} q^a_j$. This can be extended to a $\Z_n$-module map $D_m : Q \to Q$ by defining $D_m = \oplus_{a \in A} D^a_m$, where in a slight abuse of notation we use the same symbol as for the map $D_m : \Z_n[t^{\pm}] \to \Z_n[t^{\pm}]$ defined above. While $D_m$ is neither an $R$-module map nor an $R'$-module map, we do have that $D_m(Q)$ is an $R'$-submodule of $Q$.

    We have $(1- t^{m-1})q^a_i \in \operatorname{ker} \omega|_{Q_a}$, because
    \begin{equation}
    \omega[ (1-t^{m-1}) q^a_i] = \omega[ (1-(g^a_i)^{m-1}) q^a_i] = (1 - (g^a_i)^{m-1} ) \omega(q^a_i) = 0 \text{,}
    \end{equation}
    since $\T^{(m-1)}$ acts trivially on $C$. Therefore $(1-t^{m-1}) q^a_i = \Gamma \Delta^a_i q^a_i$ for some $\Gamma \in R_a$, and $\Delta^a_i | 1 - t^{m-1}$. This implies that also $\bar{\Delta}^a_i | 1 - t^{m-1}$.
    
    Now consider the set
    \begin{align}
        \mathbf b_a = \{ t^{m\Z} q^a_i \} \cup \{ t^{\{1,\dots,m-1\} + m \Z} \Delta^a_i q^a_i \}
        \equiv
        \mathbf b_a^{(m)} \cup \mathbf c_a.
    \end{align}
    We will prove that $\mathbf b_a$ is a $\Z_n$-basis for $Q_a$; first we need to show $Q_a = \spn_{\Z_n} \mathbf b_a$, where $\spn_{\Z_n}$ is the span with $\Z_n$ coefficients. We have $D_m(\Delta^a_{i} q_i) \in \spn_{\Z_n} \mathbf b_a^{(m)}$.   We can apply Lemma~\ref{lemma:compatibility-with-scaling} to either $\Delta^a_i$ or $\overline{\Delta^a_i}$, obtaining $D_m (\Delta^a_i) = (1 + \Theta^a_i) \Delta^a_i$ for a polynomial $\Theta^a_i$ whose coefficients of $t^{lm}$ vanish for all $l \in \Z$. Therefore, $\Theta^a_i \Delta^a_i q^a_i \in \spn_{\Z_n} \mathbf c_a$, and so $\Delta^a_i q^a_i = (D_m(\Delta^a_i) - \Theta^a_i \Delta^a_i) q^a_i \in \spn \mathbf b_a$. Given an expression for $\Delta^a_i q^a_i$ as a $\Z_n$-linear combination of elements of $\mathbf b_a$, we can apply $t^{jm}$ to both sides for any $j \in \Z$ to obtain $t^{j m} \Delta^a_i q^a_i \in \spn_{\Z_n} \mathbf b_a$, because $\mathbf b_a$ is closed under multiplication by $t^{j m}$.  It follows that $t^k \Delta^i_a q^i_a \in \spn_{\Z_n} \mathbf b_a$ for all $k\in \Z$, and therefore $\ker \omega|_{Q_a} \subset \spn_{\Z_n} \mathbf b$. Using $\Delta^a_i | 1- t^{m-1}$, we have $t^k (1 - t^{m-1}) q^a_i \in \spn_{\Z_n} \mathbf b_a$ for all $k \in \Z$, from which it can be shown that $t^k q^a_i \in \spn_{\Z_n} \mathbf b_a$ for all $k \in \Z$, and thus $Q_a = \spn_{\Z_n} \mathbf b_a$.

    Next we have to show that $\mathbf b_a$ is $\Z_n$-linearly independent. First of all $\mathbf b^{(m)}_a$ is linearly independent, because a non-trivial linear dependence would contract the assumption that $q^a_1, \dots, q^q_{\ell_a}$ is an $R_a$-basis. We note that in fact $\mathbf b^{(m)}_a$ is a $\Z_n$-basis for $D^a_m(Q_a)$. Similarly, $\mathbf c_a$ is linearly independent because $\Delta^a_1 q^a_1, \dots, \Delta^a_{\ell_a} q^a_{\ell_a}$ is an $R_a$-basis. Therefore, if we can show
    $\spn_{\Z_n} \mathbf b_a^{(m)} \cap \spn_{\Z_n} \mathbf c_a = \{0 \}$, this implies $\mathbf b_a$ is linearly independent.  Observe that Assumption \#3 implies
    $[ \Phi_a(t^{m k + \eta} \Delta^a_i q^a_i )(t^{m \ell} q^a_j ) ]_1 = 0$ for all $k, \ell \in \Z$ and $\eta \in \{1, \dots, m-1\}$. It follows that 
    \begin{equation}
    p \in \spn_{\Z_n} \mathbf c_a \text{ and } q \in D^a_m(Q_a) \implies \langle p, q \rangle_a = 0 \text{.} \label{eqn:statzero}
    \end{equation}
   Now suppose that $p \in \spn_{\Z_n} \mathbf b_a^{(m)} \cap \spn_{\Z_n} \mathbf c_a$ and $p \neq 0$. Because $p \in \spn_{\Z_n} \mathbf c_a \subset \operatorname{ker} \omega|_{Q_a}$, it is an $a$-planon, and $p$-modularity implies there exists $p' \in \operatorname{ker} \omega|_{Q_a}$ with $\langle p, p' \rangle_a \neq 0$. Because $Q_a = \spn_{\Z_n} \mathbf b_a$, we can write $p'$ as a $\Z_n$-linear combination of elements of $\mathbf b_a$. That is, we can write $p' = p'_b + p'_c$, where $p'_b \in \spn_{\Z_n} \mathbf b^{(m)}_a = D^a_m(Q_a)$ and $p'_c \in \spn_{\Z_n} \mathbf c_a$. By Eq.~\ref{eqn:statzero}, we have $\langle p, p'_b \rangle_a = 0$, so $\langle p, p'_c \rangle_a = \langle p, p' \rangle_a \neq 0$. But since $p \in D^a_m(Q_a)$, Eq.~\ref{eqn:statzero} implies $\langle p'_c, p \rangle_a = 0$, a contradiction.

    We have thus obtained a direct sum decomposition of $\Z_n$-modules $Q_a = D^a_m(Q_a) \oplus \tilde{Q}_a$, where 
    $\tilde{Q}_a \equiv \spn_{\Z_n} \mathbf c_a \subset \operatorname{ker} \omega|_{Q_a}$. In fact, both direct summands are $R'$-submodules of $Q_a$, so this is a decomposition of $R'$-modules. For any $q \in D^a_m(Q)$ and $\tilde{q} \in \tilde{Q}_a$, we have $\langle \tilde{q}, q \rangle_a = 0$. Moreover if also $q \in \operatorname{ker} \omega|_{Q_a}$, then $\langle q, \tilde{q} \rangle_a = 0$. We have the $R'$-module decomposition $Q = Q_1 \oplus Q_2$, where $Q_1 = \oplus_{a \in A} D^a_m(Q_a) = D_m(Q)$ and $Q_2 = \oplus_{a \in A} \tilde{Q}_a$. Since $Q_2 \subset \operatorname{ker} \omega$, using $S = Q \cap \operatorname{ker} \omega$, we also have the direct sum decomposition of $R'$-modules $S = S_1 \oplus S_2$, where $S_2 = Q_2$ and $S_1 = Q_1 \cap \operatorname{ker}\omega = \operatorname{ker} \omega|_{Q_1}$. There is no mutual statistics between planons in $S_1$ and excitations in $S_2$, and vice versa, so we have a decomposition of $p$-theories $P \to_m P' = P_1 \ominus P_2$, where $P_2$ is essentially 1-planar because $S_2 = Q_2$.

    It remains to show that $P_1 \cong P^{(m)}$. The first step is to show that we have an $R'$-module isomorphism $S_1 \cong S^{(m)}$, where $S^{(m)}$ was defined in Sec.~\ref{section:RG}. We have $\omega = \omega \circ D_m$, because the two maps agree on $\Z_n$-module generators of $Q$, namely
    \begin{equation}
    \omega(D_m(t^i q^a_j)) = \omega(t^{m i} q^a_j ) = \omega( t^{(m-1)i} t^i q^a_j )
    = \omega( (g^a_j)^{(m-1)i} t^i q^a_j) = (g^a_j)^{(m-1)i} \omega( t^i q^a_j) = \omega( t^i q^a_j) \text{,}    
    \end{equation}
    where we used the fact that $\T^{(m-1)}$ acts trivially on $C$. This implies that $D_m(\operatorname{ker} \omega) = \operatorname{ker} \omega|_{D_m(Q)}$. Therefore $S_1 = D_m(\operatorname{ker} \omega) = D_m(S)$. Since $D_m : Q \to Q$ is clearly injective, we thus have an isomorphism of $\Z_n$-modules $D_m|_S : S \to S^1$. As a $\Z_n$-module, $S^{(m)} = S$, so we also have the $\Z_n$-module isomorphism $D_m|_S : S^{(m)} \to S^1$. It can be checked that this is in fact an $R'$-module isomorphism. Finally, we need to show that $\langle p , q \rangle_a = \langle D_m(p), D_m(q) \rangle_a$ for any $p \in \operatorname{ker} \omega|_{Q_a}$ and any $q \in Q_a$. This can be checked by expanding $p$ and $q$ as $\Z_n$-linear combinations of basis elements, and using the fact $\langle t^k \Delta^a_i q^a_i , t^\ell q^a_j \rangle_a = \delta_{k,\ell}\langle  \Delta^a_i q^a_i ,  q^a_j \rangle_a$ for all $k, \ell \in \Z$, which follows from Assumption \#3 and translation invariance.
    \end{proof}

The proposition below gives a sufficient condition for a $p$-theory to admit a foliated RG where layers of toric code are exfoliated. We define a $R$-module homomorphism $\Psi^a_\Delta : \Delta R_a \oplus \bar{\Delta} R_a \to R^*_{aa}$, where $\Delta \in R_a$, $R^*_{a a} = \overline{\operatorname{Hom}}_R (R_{aa}, R_a)$, and the map is defined by
\begin{align}
            (\Delta, 0) &\mapsto (0, 1)^* \text{.}
            \\
            (0, \bar \Delta) &\mapsto (1, 0)^* \text{.}
\end{align}

\begin{proposition}
    \label{prop:foliated-rg}
    Let $P$ be a $p$-modular $p$-theory, such that, for each $a \in A$,
    \begin{enumerate}
        \item $\T / \T_a \cong \Z$. (This implies $R_a \cong \Z_n[t^\pm]$; we choose an isomorphism.)
        \item $Q_a$ is a free $R_a$-module, $Q_a \cong R_a^{2 r_a}$ for $r_a \in \N$. We choose an isomorphism and identify $Q_a = R^{2 r_a}_a$.
        \item For $i=1,\dots,r_a$, there exist $\Delta_{a i} \in R_a$ whose highest and lowest degree terms have coefficients that are units in $\Z_n$, such that $\operatorname{ker} \omega|_{Q_a} = \oplus_{i=1}^{r_a} ( \Delta_{a i} R_a \oplus \bar{\Delta}_{a i} R_a )$, and
      $\Phi_a = \oplus_{i=1}^{r_a} \Psi^a_{\Delta_{a i}}$.
    \end{enumerate}
    Then, letting $m$ be an integer such that $\T^{(m-1)}$ acts trivially on $C$, up to isomorphism
    \begin{equation}
    P \to_m \cong P^{(m)} \ominus P_\ell \text{,}
    \end{equation}
    where
    \begin{equation}
    P_\ell = \bigominus_{a \in A}  P_{\textup{$a$-oriented $\Z_n$ toric code layers}}^{\ominus r_a (m-1)}.
    \end{equation}
    That is, $P_\ell$ is a stack of $\Z_n$ toric code layers, with $r_a(m-1)$ $a$-oriented layers per unit cell of $\T^{(m)}$ translation symmetry.
\end{proposition}

\begin{proof}
    Throughout the proof, we view $S$ as a submodule of $Q$. We first establish an isomorphism between $P$ and a $p$-theory $P'$ for which the conditions of Proposition~\ref{prop:rg-2} hold. Let $Q'_a = R^{2 r_a}_a$, consider $R$-module isomorphisms $\alpha_{ai} : R_{aa} \to R_{aa}$ defined by $\alpha_{ai}(x,y) = (x, t^{k_{ai}} y)$ for $k_{ai}\in \Z$, and define $\alpha_a : Q_a \to Q'_a$ by $\alpha_a = \oplus_{i=1}^{r_a} \alpha_{ai}$. We then have the $R$-module isomorphism $\alpha : Q \to Q' = \oplus_{a \in A} Q'_a $ given by $\alpha = \oplus_{a \in A} \alpha_a$. While $Q$ and $Q'$ are identical, we use different symbols because they represent the plane charge modules of two different $p$-theories. We denote the standard $R_a$-basis elements of $Q_a$ by $q_{a e 1}, q_{a m 1}, \dots, q_{a e r_a}, q_{a m r_a}$, and similarly write $q'_{a e 1}, \dots$ for the basis elements of $Q'_a$. Then $\alpha_a(q_{a e i}) = q'_{a e i}$ and $\alpha_a(q_{a m i}) = t^{k_{ai}} q'_{a e i}$. Assumption \#3 implies  
    \begin{equation}
    \langle t^k \Delta_{a i}  q_{a e i} , t^\ell  q_{a m j} \rangle_a = 
    \langle t^k \bar{\Delta}_{a i}  q_{a m i} , t^\ell  q_{a e j} \rangle_a 
    = \delta_{i j} \delta_{k \ell} \text{.}
    \end{equation}
    In order for $\alpha$ to give an isomorphism of $p$-theories, it must preserve statistics, so we require
    \begin{equation}
    \langle t^k \Delta_{a i} q'_{a e i} , t^{\ell+ k_{a j}}  q'_{a m j} \rangle'_a = 
    \langle t^{k + k_{a i}} \bar{\Delta}_{a i}  q'_{a m i} , t^\ell  q'_{a e j} \rangle'_a 
    = \delta_{i j} \delta_{k \ell} \text{.}
    \end{equation}
    This holds if we choose $\Phi'_a : Q'_a \to (Q'_a)^*$ to be $\Phi'_a = \oplus_{i =1}^{r_a} \Psi_{\Delta'_{ai}}$, where $\Delta'_{ai} = \bar{t}^{k_{ai}} \Delta_{ai}$. 
    We get an isomorphism of $p$-theories by setting $S' = \alpha(S)$. Note that $S'_a = \alpha(S_a) = \oplus_{i=1}^{r_a} ( \Delta_{a i} R_a \oplus \bar{\Delta}_{a i} R_a ) = \oplus_{i=1}^{r_a} ( \Delta'_{a i} R_a \oplus \bar{\Delta}'_{a i} R_a )$, so the choice of $\Phi'_a$ makes sense.
      
    The integers $k_{ai}$ can be chosen so that $\Delta'_{a i}$ has terms of non-negative degree, with a non-zero constant term. Since we are only interested in $P$ up to isomorphism, we can assume that $\Delta_{ai}$ has this form without loss of generality. 
    The conditions of Proposition~\ref{prop:rg-2} are thus satisfied, and up to isomorphism we have
    $P \to_m P^{(m)} \ominus P_\ell$, where $P_\ell$ is essentially 1-planar, so $P_\ell = \oplus_{a \in A} P_{\ell, a}$, where $P_{\ell, a}$ is a $p$-theory whose only non-trivial excitations are the $a$-oriented planons of $P_\ell$. 

    We let $\T' = \T^{(m)}$ and $R' = \Z_n \T'$. By the proof of Proposition~\ref{prop:rg-2}, we have a decomposition of $R'$-modules $Q_a = D_m(Q_a) \oplus Q^\ell_a$, where $Q^\ell_a = S_{\ell, a}$ is the module of superselection sectors of $P_{\ell,a}$; note that $Q^\ell_a$ is called $\tilde{Q}_a$ in the proof of Proposition~\ref{prop:rg-2}. Denoting $\T'_a = \T' \cap \T_a$ and $R'_a = \Z_n \T' / \T'_a$, because $\T'_a$ acts trivially on $Q_a$ and its submodules above, we can view these as $R'_a$-modules. There is an injective group homomorphism mapping $\T' / \T'_a \to \T / \T_a$ induced by the inclusion $\T' \hookrightarrow \T$; this in turn induces an injective ring homomorphism $R'_a \hookrightarrow R_a$, so we can view any $R_a$-module as an $R'_a$-module. Identifying $R_a = \Z_n[t^\pm]$, this leads us to identify $R'_a$ with the subring of polynomials of the form $\sum_{i \in \Z} c_i t^{i m}$. 
    
    To complete the proof, we need to show that $P_{\ell,a}$ is a stack of $r_a (m-1)$ copies of toric code layers. For simplicity, we take $r_a = 1$; the case of $r_a > 1$ is a trivial generalization. We denote the two standard $R_a$-basis elements for $Q_a = R_a^2$ as $q_{a \fe}$ and $q_{a \fm}$, where the labels $\fe$ and $\fm$ are chosen to refer to the electric and magnetic sectors of toric code layer excitations. Examining the proof of Proposition~\ref{prop:rg-2}, the $R'$-module decomposition $Q_a = D_m(Q_a) \oplus Q^\ell_a$ breaks up into separate decompositions for each basis label. That is, if we define $Q_{a\fe}$ ($Q_{a\fm}$) to be the $R_a$-module generated by $q_{a\fe}$ ($q_{a\fm}$), then we have $Q_{a \fe} = D_m(Q_{a \fe}) \oplus Q^\ell_{a \fe}$, where $Q^\ell_{a \fe}$ is the $R'_a$-module with basis $\{ t^i \Delta_a q_{a\fe} | i =1,\dots,m-1 \}$. The analogous statement holds for $Q_{a\fm}$, where $Q^\ell_{a \fm}$ is the $R'_a$-module with basis $\{ t^i \bar{\Delta}_a q_{a\fe} | i =1,\dots,m-1 \}$. Moreover we have $D_m(Q_a) = D_m(Q_{a\fe}) \oplus D_m(Q_{a \fm})$ and $Q^\ell_a = Q^\ell_{a \fe} \oplus Q^\ell_{a \fm}$.
    
    Now, for $i=1,\dots,m-1$, we have $t^i q_{a\fe} = p_{\fe i} + r'_i q_{a \fe}$, where $r'_i \in R'_a$, and  $p_{\fe i} \in Q^\ell_{a \fe}$ and $r'_i q_{a \fe} \in D_m(Q_{a \fe})$ are unique by the direct sum decomposition of $Q_{a \fe}$. Since $t^i q_{a\fe} \notin D_m(Q_{a\fe})$, we have $p_{\fe i} \neq 0$. Defining $p_{\fm i} = t^i \bar{\Delta}_a q_{a \fm}$ for $i=1,\dots,m-1$, we let
    \begin{equation}
    \mathbf p = \mathbf p_\fe \cup \mathbf p_\fm = \{ p_{\fe i} | i = 1,\dots,m-1 \}
    \cup \{ p_{\fm i} | i = 1,\dots,m-1 \} \text{.}
    \end{equation}
    We claim that $\mathbf p$ is an $R'_a$-module basis for $Q^\ell_a$.  It is clear that $\mathbf p$ is a $R'_a$-linearly independent; we need to show that $\spn_{R'_a} \mathbf p = Q^\ell_a$. It is enough to show that $t^i \Delta_a q_{a\fe} \in \spn_{R'_a} \mathbf p_\fe$. Each term in $t^i \Delta_a q_{a\fe}$ is of the form $d_k t^{i + k} q_{a\fe}$ for $k = 0, \dots, \operatorname{deg} \Delta_a$ and $d_k \in \Z_n$, which is either an element of $D_m(Q_{a\fe})$, or can be written in the form $s'_k ( p_{\fe j(k)} + r'_{j(k)} q_{a \fe})$ for some $s'_k \in R'_a$ and some $j(k) \in\{1,\dots,m-1\}$. We thus obtain an expression of the form $t^i \Delta_a q_{a\fe} = \sum_{k} s'_k p_{\fe j(k)} + q_{(m)}$ for some $ q_{(m)} \in D_m(Q_{a\fe})$. By the direct sum decomposition  $Q_{a\fe} = D_m(Q_{a\fe}) \oplus Q^\ell_{a\fe}$, we must have $q_{(m)} = 0$.
    
    Finally, it is straightforward to compute
    \begin{equation}
    \langle t^{mk} p_{\fm i}, t^{m\ell} p_{\fe j} \rangle_a = 
    \langle t^{i+mk} \bar{\Delta}_a  q_{a \fm}, t^{j + m\ell} q_{a \fe} \rangle_a
    = \delta_{i j} \delta_{k \ell} \text{,}
    \end{equation}
    and
    \begin{equation}
    \langle t^{mk} p_{\fe i},  t^{m \ell} p_{\fe j} \rangle_a = 
    \langle t^{mk} p_{\fm i}, t^{m\ell} p_{\fm j} \rangle_a = 0 \text{.}
    \end{equation}
    Therefore $P_{\ell,a}$ is a stack of $m-1$ copies of toric code layers as claimed.
\end{proof}

\section {Checkerboard equivalences}
\label{appendix:checkerboard}

Here we fix $n$ to be odd, and show that the $\Z_n$ checkerboard model decomposes as a stack of three $\Z_n$ anisotropic modes.
First, we show equivalence of $p$-theories, 
\begin{align}
    P_\text{Ch} \cong P_\text{An}^x \ominus P_\text{An}^y \ominus P_\text{An}^z,
\end{align}
then we 
show circuit equivalence. 
Of course, the latter implies the former. 
We also demonstrate the $p$-theory equivalence because it is (arguably) simpler, and because its structure is mirrored in part of the circuit equivalence. 

Let us recall the $p$-theories of the two models. 
In both cases, the translational symmetry group associated with the theory is $\mathcal T = \langle x,y,z\rangle$; the embedding of the lattice into real space is $\eta : x^a y^b z^c \mapsto a \hat x + b \hat y + c \hat z$; and the normal vectors asssociated to each plane of mobility are $\hat n_x = \hat x, \hat n_y = \hat y, \hat n_z = \hat z$. The constraint maps and planon mutual statistics are given by:

\textit{$\mu$-Anisotropic:}
\begin{align}
     P_\text{An}^\mu
    =
    (
    &\omega: \Z_n[t^\pm]\{ \mathbf q_\nu^s, \mathbf q_\rho^s \}
    \to \Z_n\{ \mathbf r_\nu^e, \mathbf r_\rho^m \},
    \\
    & \hspace{59pt}
    \Delta \mathbf q_{\mu'}^s
    \mapsto \begin{cases}
        \mathbf r_\nu^e & \text{$\mu' = \nu$ and $s = e$}
        \\
        - \mathbf r_\nu^e & \text{$\mu' = \rho$ and $s = e$}
        \\
        \mathbf r_\rho^m & \text{$s = m$}
    \end{cases},
    \nonumber \\ &
    \Phi_{\mu'} : \Z_n[t^\pm]\{ \Delta \mathbf q_{\nu}^s, \Delta \mathbf q_{\rho}^s\}
    \to \Z_n[t^\pm]\{ \mathbf q_{\nu}^s, \mathbf q_\rho^s \}^*,
    \nonumber \\ & \hspace{86pt}
    \Delta \mathbf q_{\mu'}^e \mapsto \mathbf q_{\mu'}^{m*},
    \nonumber \\ & \hspace{86pt}
    \bar \Delta \mathbf q_{\mu'}^m \mapsto \mathbf q_{\mu'}^{e*}
    )
    \nonumber
\end{align}

\textit {Checkerboard:}
\begin{align}
    \label{eq:check-data}
    P_\text{Ch} = (
        &\omega : \Z_n[t^\pm]\{ \mathbf q_{\mu\epsilon}^s \}
        \to \Z_n \{ \mathbf r_{\mu}^s \},
        \\ & \hspace{57pt}
        \mathbf q_{\mu\epsilon}^s
        \mapsto \begin{cases}
            \mathbf r_\nu^s - \mathbf r_\rho^s & \epsilon = 0
            \\
            \mathbf r_\mu^s & \epsilon = 1
        \end{cases},
        \nonumber \\ &
        \Phi: \Z_n[t^\pm]\{ \Delta \mathbf q_{\mu\epsilon}^s \} \to \Z_n[t^\pm]\{ \mathbf q_{\mu\epsilon}^s \}^*,
        \nonumber \\ & \hspace{57pt}
        \Delta \mathbf q_{\mu 0}^e \mapsto \mathbf q_{\mu 1}^{m*},
        \nonumber \\ & \hspace{57pt}
        \bar \Delta \mathbf q_{\mu 1}^m \mapsto \mathbf q_{\mu  0}^{e*},
        \nonumber \\ & \hspace{57pt}
        \bar \Delta \mathbf q_{\mu 1}^e \mapsto 
        - \mathbf q_{\mu 0}^{m*},
        \nonumber \\ & \hspace{57pt}
        \Delta \mathbf q_{\mu 0}^m \mapsto 
        - \mathbf q_{\mu 1}^{e*}
        )
        \nonumber
\end{align}

Here $\Delta = 1 - t$. To specify the $R$-module structure in each of the above spaces, it suffices to write that $\mu^a \nu^b \rho^c \mathbf q_{\mu I} = t^a \mathbf q_{\mu I}$ and that all translations act trivially on $\mathbf r_J$, where $I$ and $J$ contain any other appropriate indices. Note that we present the checkerboard data having already coarse-grained down to the cubic sublattice of the fcc lattice.

Let $\{ \mathbf q_{\mu\epsilon}^s \}$ and $\{ \mathbf r_{\mu}^s \}$ be checkerboard plane charge and constraint module generators, respectively, as above. 
Noting that since $n$ is odd, $2 \in \Z_n^\times$,
consider the change of bases to $\{ \tilde \q_{\nu}^{s\mu}, \tilde \q_{\rho}^{s\mu}\}$ and $\{\tilde {\mathbf r}^{s\mu}\}$ defined by 
\begin{align}
    \tilde \q_\mu^{e\nu} = - \q_{\mu 1}^e + \q_{\mu 0}^e,
    \\
    \tilde \q_{\mu}^{m \nu} = \frac 12 ( \q_{\mu0}^m - t {\q_{\mu 1}}^m),
    \\
    \tilde \q_\mu^{e\rho} = \q_{\mu0}^e + \q_{\mu 1}^e,
    \\
    \tilde \q_{\mu}^{m\rho} = \frac 12 (\q_{\mu 1}^m + t \q_{\mu0}^m),
\end{align}
and 
\begin{align}
    \tilde {\mathbf r}^{s\mu}
    = 2^{-\delta_{sm}} (\mathbf r^s_\nu + \mathbf r^s_\rho - \mathbf r^s_\mu).
\end{align}
With $\omega$ the checkerboard constraint map as given above, one computes
\begin{align}
    \omega : \hspace{2pt} &\tilde \q_{\nu}^{e\mu},
    \tilde \q_\rho^{e\mu}
    \mapsto \tilde {\mathbf r}^{e\mu},
    \\
    &\tilde \q_\nu^{m\mu}, - \tilde \q_\rho^{m\mu} \mapsto \tilde {\mathbf r}^{m\mu}.
\end{align}
This precisely concides with the constraint map of three anisotropic models, indexed by the upper $\mu$. Meanwhile, regarding remote detection, we compute
\begin{align}
    \llangle \Delta \tilde \q_\nu^{e\mu}, \tilde \q_\nu^{s \mu'} \rrangle 
    =
    \llangle \Delta \tilde \q_\rho^{e\mu}, \tilde \q_\rho^{s \mu'} \rrangle 
    =
    \delta_{sm} \delta_{\mu\mu'},
% \end{align}
% and
% \begin{align}
\\
    \langle \bar \Delta \tilde \q_\nu^{m\mu}, \tilde \q_{\nu}^{s\mu'} \rangle
    = 
    \langle \bar \Delta \tilde \q_{\rho}^{m\mu},
    \tilde \q_{\rho}^{s\mu'}
    \rangle =
    \delta_{se} \delta_{\mu\mu'}.
\end{align}
This shows that planon remote detection also decomposes into that of anisotropic models, and establishes the desired equivalence of theories.

Now we show the circuit equivalence.
The (coarse-grained) checkerboard excitation map is given by
\begin{align}
    \epsilon
    =
    \left(
    \begin{tabular}{ C C  C C | C C C C | C  C C C | C C  C C }
            1&1&1&1& -1&-{yz}&-{zx}&-{xy}& &&&& &&&
            \\
            \overline {yz} & 1 & \bar y & \bar z &
            -1 & -1 & - x & -x & &&&& &&&
            \\
            \overline {zx} & \bar x & 1 & \bar z & -1 & -y & -1 & -y &
            &&&& &&&
            \\
            \overline {xy} & \bar x & \bar y & 1 & -1 & -z & -z & -1 & &&&& &&&
            \\
            \hline
            &&&& &&&&  1&1&1&1& 1&{yz}&{zx}&{xy} 
            \\
            &&&& &&&&  \overline {yz} & 1 & \bar y & \bar z &
            1 & 1 &  x & x 
            \\
            &&&& &&&&  \overline {zx} & \bar x & 1 & \bar z & 1 & y & 1 & y 
            \\
            &&&& &&&&  \overline {xy} & \bar x & \bar y & 1 & 1 & z & z & 1  
    \end{tabular}
    \right).
\end{align}
If we take the stabilizer change of basis (which mirrors the above plane charge change of basis)
\begin{align}
    u =
    \left(
    \begin{tabular}{C C C C | C C C C }
        1 &&&& &&&
        \\
        -\tfrac 12 &  - \tfrac 12 & \tfrac 12 & \tfrac 12 &&&&
        \\
        - \tfrac 12 & \tfrac 12 & - \tfrac 12 & \tfrac 12 &&&&
        \\
        - \tfrac 12 & \tfrac 12 & \tfrac 12 & - \tfrac 12 &&&&
        \\
        &&&& 2 &&& \\
        &&&& -1 & - \overline {yz} & \overline {zx}  &  \overline {xy} \\
        &&&& -1 & \overline {yz} & - \overline {zx}  &  \overline {xy} \\
        &&&& -1 & \overline {yz} & \overline {zx}  &  - \overline {xy} \\
    \end{tabular}
    \right),
\end{align}
and the elementary symplectic transformation
\begin{align}
    v = 
    \left(
    \begin{tabular}{C C C C | C C C C | C C C C | C C C C}
        2 &&&& &&&& &&&&  &&&
        \\
        2 & 1 &&& && \bar f_z & \bar f_y & &&&& &&&
        \\
        2 && 1 && & \bar f_z && \bar f_x & &&&& &&&
        \\
        2 &&& 1 & & \bar f_y & \bar f_x &&  &&&&  &&&
        \\
        \hline
        2 &  1  &  1 & 1 &  1 & \bar g_x & \bar g_y & \bar g_z & &&&&  &&&
        \\
        2 \overline{yz} &&&& & 2\overline{yz} &&& &&&& &&&
        \\
        2 \overline{zx} &&&& && 2\overline{zx} && &&&& &&&
        \\
        2 \overline{xy} &&&& &&& 2\overline{xy} & &&&& &&&
        \\
        \hline
        &&&& &&&& \tfrac 12 & \tfrac 12 g_x & \tfrac 12 g_y & \tfrac 12 g_z & -1 & - \tfrac 12 & - \tfrac 12 & - \tfrac 12 
        \\
        &&&& &&&& & 1 &&& -1 &&&
        \\
        &&&& &&&& && 1 && -1 &&&
        \\
        &&&& &&&& &&& 1 & -1 &&&
        \\
        \hline
        &&&& &&&& && - \tfrac 12 \overline {yz} f_z &  - \tfrac 12 \overline {yz} f_y & \overline {yz} & \tfrac 12 \overline{yz} && \\
        &&&& &&&& & - \tfrac 12 \overline {zx} f_z &&  - \tfrac 12 \overline {zx} f_x & \overline {zx} && \tfrac 12 \overline{zx} & \\
        &&&& &&&& & - \tfrac 12 \overline {xy} f_y &  - \tfrac 12 \overline {xy} f_x && \overline {xy} &&& \tfrac 12 \overline{xy} \\
    \end{tabular}
    \right),
\end{align}
where $f_\mu = 1 + \mu$, $g_\mu = \nu + \rho$, then
we find
\begin{align}
    u \epsilon v 
    =
    \left(
    \begin{tabular}{C C C C | C C C C | C C C C | C C C C}
        0 &&&& -1 &&&& &&&&  &&&
        \\
        & -\bar d_x &&& & \bar d_{yz} &&& &&&& &&&
        \\
        && -\bar d_y && && \bar d_{zx} && &&&& &&&
        \\
        &&& -\bar d_z & &&& \bar d_{xy} & &&&& &&&
        \\
        &&&& &&&& 1 &&&& 0 &&&
        \\
        &&&& &&&& & - d_{yz} &&&  &- d_x &&
        \\
        &&&& &&&& && - d_{zx} &&  && - d_y &
        \\
        &&&& &&&& &&& - d_{xy} &  &&& - d_z 
    \end{tabular}
    \right),
\end{align}
which is the excitation map of three decoupled $\Z_n$ anisotropic models plus two ancilla qudits per (coarse-grained) unit cell.

\bibliography{bibliography}

\end{document}